\documentclass[11pt]{article}
\usepackage{graphicx,amsmath,amsfonts,amssymb,amsthm,fullpage,times,color}
\usepackage{ytableau}
\catcode`\@ 11
\catcode`\@ 12
\numberwithin{equation}{section}

\def\sf#1#2{{\textstyle\frac{#1}{#2}}} 
\newcommand{\myatop}[2]{\genfrac{}{}{0pt}{}{#1}{#2}}  
\def\Re{\mathop{\rm Re}\nolimits}
\def\Im{\mathop{\rm Im}\nolimits}
\def\Wr{\mathop{\rm Wr}\nolimits}

\allowdisplaybreaks

\theoremstyle{definition}  
\newtheorem{definition}{Definition}[section]

\newtheorem{lemma}{Lemma}[section]
\newtheorem{corollary}{Corollary}[section]
\newtheorem{proposition}{Proposition}[section]
\newtheorem{remark}{Remark}[section]
\newtheorem{example}{Example}[section]

\begin{document}

\title{Classification of KPI lumps}
\author{Sarbarish Chakravarty and Michael Zowada \\[1ex]
\small\it\
Department of Mathematics, University of Colorado, Colorado Springs, CO 80918 \\}
\date{}
\maketitle \kern-2\bigskipamount

\begin{abstract}
A large family of nonsingular rational solutions of the Kadomtsev-Petviashvili 
(KP) I equation are investigated.  These solutions are constructed via the Gramian 
method and are identified as points in a complex Grassmannian.
Each solution is a traveling wave moving with a uniform background velocity but have 
multiple peaks which evolve at a slower time scale in the co-moving frame.
For large times, these peaks separate and form well-defined wave patterns in the 
$xy$-plane. The pattern formation are described by the roots of well-known polynomials 
arising in the study of rational solutions of Painlev\'e II and IV equations.
This family of solutions are shown to be described by the classical Schur functions 
associated with partitions of integers and irreducible representations of the symmetric 
group of $N$ objects. It is then shown that there exists a one-to-one
correspondence between the KPI rational solutions considered in this article
and partitions of a positive integer $N$.
\end{abstract}

\thispagestyle{empty}

\section{Introduction}

An important example of nonlinear wave equations in $(2+1)$-dimensions  
is the Kadomtsev-Petviashvili (KP) equation, which is 
dispersive equation describing the propagation of small amplitude, 
long wavelength, uni-directional waves with small transverse variation.  
It was originally proposed by Kadomtsev and Petviashvili~\cite{KP70} 
to study ion-acoustic waves of small amplitude propagating in plasmas.
The KP equation has many physical applications and arises in such diverse
fields as plasma physics~\cite{IR00,L98}, fluid dynamics~\cite{AS81,A11,K18}, 
nonlinear optics~~\cite{PSK95,BWK16} and ferromagnetic media~~\cite{TF85}.
It is also an exactly solvable nonlinear equation with remarkably 
rich mathematical structure documented
in many research monographs (see e.g.~\cite{NMPZ1984,AC91,IR00,H04,K18}).

There are two mathematically distinct versions of the KP equation, referred
to as KPI and KPII. This article is concerned with the KPI equation which
can be expressed as
\begin{equation}\label{kp}
(4u_t+6uu_x+u_{xxx})_x=3u_{yy}.
\end{equation}
Here $u=u(x,y,t)$ represents the normalized wave amplitude at the point $(x,y)$ in
the $xy$-plane for fixed time $t$, and the subscripts denote partial derivatives.
The KPII equation is \eqref{kp} with $-3u_{yy}$ in the right hand side.
From the water wave theory perspective, KPI corresponds to large surface tension
while KPII arises in the small surface tension limit of the 
multiple-scale asymptotics~\cite{AS81,A11}.

The KPI equation admits large classes of exact rational solutions
known as lumps which are localized in the $xy$-plane and are non-singular
for all $t$. The simplest type of rational solutions was first discovered 
analytically by employing the dressing method~\cite{MZBIM77} and subsequently 
via the Hirota method~\cite{SA79}. These solutions consist of $N$ local maxima
(peaks) traveling with distinct velocities and their trajectories
remain unchanged before and after interaction.  
These solutions are often referred to as {\it simple lumps} in contrast
to yet another class of KPI rational solutions arise as bound states
formed by fusing the simple lump solutions in a certain manner. 
These are called the {\it multi-lump}  
solutions which were originally found in~\cite{JT78} by algebraic
techniques and further investigated by several 
authors~\cite{PS93,GPS93,ACTV00}. The multi-lump solution is
an ensemble of a finite number of localized structures (or peaks) 
interacting in a non-trivial manner
unlike the $N$-simple lump solution. The peaks move with the same center-of-mass
velocity but undergo anomalous scattering with a non-zero deflection angle
after collision. Furthermore, the peak amplitudes evolve in time and reaches
a constant asymptotic value which equals that of the simple $1$-lump peak.
Both simple and multi-lump rational solutions of the KPI equation 
are rational potentials associated with the time-dependent Schr\"odinger equation, 
and were studied via inverse scattering method~\cite{AV97,VA99}.

In this article, a large family of rational solutions of KPI 
described above as multi-lump solutions, are investigated. These solutions
constructed here via the Gramian method which expresses the solutions
in terms of the determinant (Gramian) of a Gram matrix. The entries of the Gram 
matrix are constructed out of the inner products of $n$ linearly independent
complex vectors which spans a $n$-dimensional subspace which represents
a non-degenerate point $V_P$ in a complex Grassmannian. Then the Gramian 
which is the $\tau$-function of the associated KPI solution is the
norm of the complex $n$-form obtained as the image of $V_P$  via the 
Pl\"ucker embedding of the Grassmannian. It is shown in this paper that this
geometric construction leads immediately to a positive definite polynomial
form of the $\tau$-function for all $(x,y,t) \in \mathbb{R}^3$ so that
the resulting KPI rational solution is nonsingular in the $xy$-plane
for all $t$ and decays as $(x^2+y^2)^{-1}$.

The next part of this paper describes how to interpret the multi-lump 
$\tau$-function as a sum of squares in terms of the classical Schur functions 
which arises in the theory of symmetric functions. The Schur functions are 
weighted homogeneous polynomials of degree $N$ in $N$ complex variables
with specified weights of each variable, and form a basis of 
symmetric polynomials over integers. The Schur functions play an important
role in the Sato theory of the KP equation~\cite{S81,OSTT88,K18}. 
However, the relationship between 
the multi-lump $\tau$-function and the Schur functions unfolded in this
article is new, and forms a key feature of the underlying multi-lump 
solution structure. The Schur functions then lead to a natural characterization 
that associates each KPI multi-lump solution to a unique partition $\lambda$ of 
a given positive integer $N$ that is the degree of the corresponding Schur function.
Finally, this last characterization can be utilized to formulate a comprehensive
classification of the multi-lump solutions which are thus referred to as the
$N$-lump solutions of the KPI equation. This classification is also new and 
completes the task that was initiated in an earlier comprehensive 
article~\cite{ACTV00}.

Finally, a detailed study of the long time behavior of the $N$-lump solutions
is carried out. Our investigation reveals that (a)\, for $|t| \gg 1$, the $N$-lump 
KPI solutions splits into $N$ distinct peaks which form a rich variety
of surface wave patterns in the $xy$-plane; and (b)\, the solution in the
$O(1)$-neighborhood of each peak is a single $1$-lump solution as $|t| \to \infty$
implying that the $N$-lump solution can be viewed as a superposition of $N$ $1$-lump
solutions for large times. The approximate location of the peaks are determined
in terms of the zeros of the Yablonskii-Vorob'ev and Wronskian-Hermite (heat)
polynomials which also describe certain classes of rational solutions of the
Painlev\'e II and IV equations, respectively. This occurrence is not accidental
but is related to the similarity reductions of the KPI $N$-lump solutions
although this matter will not be pursued further in this paper. Note that the 
conclusion in part (a) above is based on an assumption that the non-zero roots 
of the Wronskian-Hermite polynomials are simple~\cite{FHV12}.
It is also worth noting that some of the exotic surface patterns of the KPI
multi-lumps were observed earlier~\cite{GPS93,ACTV00,G18,DLZ21} and 
particularly, the long time asymptotics has been reported recently~\cite{YY21}
after our investigation was complete. However, the treatment presented in this
paper is new and different from the earlier works since it utilizes the special
property of the Schur functions as characteristics of irreducible representation
of the symmetric group $S_N$. In fact, this property plays a key role in our 
long time asymptotic analysis which shows that as a whole, the $N$-lump structure 
is a traveling wave which moves with a uniform velocity while the internal 
dynamics of the component lumps takes place at a slower time scale $O(|t|^{1/p})$,
where (generically) $p=2$ or $p=3$. In a recent paper~\cite{CZ21}, we analyzed the
lump dynamics for a special class of $N$-lump solutions, and showed
the dynamics corresponds to a special reduction of the Calogero-Moser system. 
We believe that is the case for the general KPI $N$-lumps although we do not 
pursue the matter in this article.

The paper is structured as follows: Section 2 describes the Gramian
construction of the $\tau$-function for the multi-lumps
and provides some illustrative examples of some simple multi-lump
solutions. A brief overview of the Schur functions and integer partition
theory is provided in Section 3, followed by a characterization of the
KPI solutions by a partition $\lambda$ of a positive integer $N$, that 
finally leads to a complete classification of the $N$-lump solutions 
in terms of integer partitions. The long time behavior of the $N$-lump
solutions is described in Section 4 which begins with a brief overview of
irreducible characters of the symmetric group $S_N$, that is essential
for the asymptotic analysis presented. Section 5
contains some concluding remarks including future directions
of this continued investigation of KPI lumps. 
\section{Construction of multi-lump solutions}
In this section we describe how to construct a large family
of multi-lump solutions via Gramians~\cite{H04} which arise from the
application of the binary Darboux transformation
to the KPI equation~\cite{MS91}. The building blocks for these solutions 
are given by a special type of complex polynomials called the generalized 
Schur polynomials introduced below.

\subsection{Generalized Schur polynomials}
Let $k$ be a complex parameter and $\theta := kx+k^2y+k^3t+\gamma(k)$ where 
$(x,y,t) \in \mathbb{R}^3$ and $\gamma(k)$ is an arbitrary, 
differentiable (to all orders) function of $k$. The generalized Schur polynomial 
$p_n(x,y,t,k)$ is then defined via
\begin{equation}
\phi_n:=\frac{1}{n!}\partial_k^n\exp(i\theta) = p_n\exp(i\theta)\,.
\label{sp}
\end{equation}
The $p_n$ is a polynomial in $n$ variables 
$\theta_1, \theta_2, \cdots, \theta_n$ with $\theta_j := i\sf{\partial_k^j\theta}{j!}$.
However, only the first three variables 
\begin{equation}
\theta_1 = i(x+2ky+3k^2t+\gamma_1), \quad \theta_2 =i(y+3kt+\gamma_2),
\quad \theta_3 = i(t+\gamma_3)
\label{theta}
\end{equation}
depend on $(x,y,t)$ while $\theta_j = i\gamma_j(k)$ for $j>3$ depend only on $k$.
For a fixed value of $k$, $\gamma_j(k) = \sf{\partial_k^j\gamma(k)}{j!}$, 
$j=1,\cdots,n$ are viewed as independent complex parameters which parametrize
each polynomial $p_n$.

A generating function for the $p_n$'s is given by the Taylor series
\begin{equation}
\exp(i\theta(k+h)) = \exp(i\theta) \sum_{n=0}^\infty p_n(k)h^n \,,  
\label{gen}
\end{equation}
which yields, after expanding $i\theta(k+h)=i\theta(k)+h\theta_1+h^2\theta_2+\cdots$,
and comparing with the right hand side, an explicit expression for $p_n$, namely
\begin{equation}
p_n(\theta_1,\ldots,\theta_n) = \hspace{-.2in} 
\sum_{\myatop{m_1,m_2,\cdots,m_n \geq 0}{m_1+2m_2+\cdots nm_n=n}}  
\prod_{j=1}^n\frac{\theta_j^{m_j}}{m_j!} \,.
\label{pn}
\end{equation}   
The first few generalized Schur polynomials are given by
\[p_0=1, \quad p_1=\theta_1, \quad p_2 = \sf12\theta_1^2+\theta_2 \quad
p_3 = \sf{1}{3!}\theta_1^3+\theta_1\theta_2+\theta_3 \ldots \,.\]
It follows from \eqref{pn} that $p_n$ is a {\it weighted}\, homogeneous
polynomial of degree $n$ in $\theta_j, \, j=1,\ldots,n$, i.e.,
$p_n(a\theta_1,a^2\theta_2,\cdots,a^n\theta_n)=
a^np_n(\theta_1,\theta_2,\cdots,\theta_n)$, where weight$(\theta_j)=j$.
Some useful properties for the $p_n$'s are listed below. They can
be derived using \eqref{sp} and \eqref{gen}, and will be
used throughout this article. 
\begin{subequations}
\begin{equation}
\partial^j_{\theta_1}p_n = \partial_{\theta_j} p_n = 
\left\{ \begin{matrix} p_{n-j}\,, & \quad j\leq n \cr
0\,, & \quad j>n \end{matrix} \right.  \label{propa} 
\end{equation}
\begin{equation}
p_{n+1}(k) = \sf{1}{n+1}\sum_{j=0}^n (j+1)\,\theta_{j+1}\,p_{n-j}\,, 
\quad n \geq 0, \qquad p_0=1 
\label{propb}
\end{equation} 
\begin{equation}
p_n(\theta_1+h_1, \theta_2+h_2, \cdots, \theta_n+h_n) =
\sum_{j=0}^n p_j(h_1,h_2,\cdots,h_j) \, p_{n-j}(\theta_1,\theta_2,\cdots,\theta_{n-j})
\label{propc}
\end{equation}
\end{subequations}
\begin{remark}
\begin{itemize}
\item[(a)] The generalized Schur polynomials were introduced in
the study of rational solutions for the Zakharov-Shabat and KP 
hierarchies~\cite{M79,P94,P98} where the phase was defined as 
a quasi-polynomial $\theta = kt_1+k^2t_2+k^3t_3+\cdots$ in the
multi-time variables $(t_1,t_2,t_3,\ldots)$.
In this paper, we restrict the $p_n$'s
to depend only on the first three variables $(t_1,t_2,t_3):=(x,y,t)$ while the
dependence on the remaining variables are parametric, through the complex
parameters $\theta_j= i\gamma_j$ for $j>3$.
\item[(b)] It is possible to recover the standard Schur polynomials 
which are important in the
Sato theory of the KP hierarchy~\cite{S81} (see also~\cite{OSTT88,K18})
by modifying the generating function \eqref{gen} as 
\begin{equation*}
\exp(\theta(k)) = \sum_{n=0}^\infty s_nk^n \,,  
\end{equation*}
where $\theta = kt_1+k^2t_2+k^3t_3+\cdots$. The Schur polynomials
$s_n(t_1,t_2,\ldots,t_n)$ are then given by \eqref{pn} by replacing
$\theta_j$ with $t_j$.
\end{itemize} 
\end{remark} 
\subsection{The multi-lump $\tau$-function}
The solution of the KPI equation \eqref{kp} can be expressed as
\begin{equation}
u(x,y,t) = 2 (\ln \tau)_{xx} \,,
\label{u}
\end{equation}
where the function $\tau(x,y,t)$ is known as the 
$\tau$-function~\cite{S81,H04}. We describe below an explicit
construction and the resulting properties of a $\tau$-function 
associated with a large family of rational multi-lump solutions of KPI.

Let $1 \leq m_1<m_2<\cdots<m_n$ denote $n$ distinct positive
integers and let $\phi_{m_j}(x,y,t,k)=p_{m_j}\exp(i\theta)$ 
be as in \eqref{sp} and $\bar{\phi}_{m_j}$ be its complex conjugate..
Define a $n \times n$ hermitian matrix $M$ whose entries are 
\begin{equation}
M_{ij} = \int_x^\infty \phi_{m_i}\bar{\phi}_{m_j}\,dx' =
\int_x^\infty p_{m_i}\bar{p}_{m_j}\exp i(\theta-\bar{\theta})\,dx' \,,
\label{M}
\end{equation}
where the parameter $k=a+ib$ is chosen
such that $b:=\Im(k)>0$ in order for the integral in \eqref{M} to converge. 
Then a $\tau$-function for KPI is 
\begin{equation}
\tau(x,y,t) = \det M \,, 
\label{tau}
\end{equation}
such that the corresponding function $u(x,y,t)$ in \eqref{u} satisfies
the KPI equation. This form of the $\tau$-function is called
the Gramian where $M$ is a Gram matrix which can be derived 
using a variety of algebraic techniques such as the binary Darboux
transformation~\cite{MS91} as well as the Hirota bilinear method~\cite{H04}.

The matrix $M$ is positive definite since for any vector 
$v \in \mathbb{C}^n$ 
\[v^{\dagger}Mv = \sum_{i,j}\bar{v}_iM_{ij}v_j =
\int_x^\infty \big(\sum_{i,j}\bar{v}_i\phi_{m_i}\bar{\phi}_{m_j}v_j\big)\,dx'
= \int_x^\infty \big(\sum_{i}|\bar{v}_i\phi_{m_i}|^2\big)\,dx' > 0 \,,
\]
so that $\tau(x,y,t)=\det M>0$. Consequently, the corresponding
KPI solution given by \eqref{u} is nonsingular in the $xy$-plane
for all $t$. Furthermore, evaluating the integral in \eqref{M} by 
integration by parts, yield
\begin{equation}
M_{ij}=\frac{\exp i(\theta-\bar{\theta})}{2b}H_{ij}\,, \qquad \quad
H_{ij} = \sum_{r=0}^{m_i+m_j}
\frac{\partial^r_x(p_{m_i}\bar{p}_{m_j})}{(2b)^r} \,.
\label{H}
\end{equation}
Consequently, \eqref{u} can be re-expressed as 
\[u = 2 (\ln \det M)_{xx} = 2 (\ln \det H)_{xx} \,,\] 
since the factor $\sf{e^{in(\theta-\bar{\theta})}}{(2b)^n}$ arising 
in $\det M$ is annihilated by $\ln(\cdot)_{xx}$. The generalized 
Schur polynomial $p_j$ is of degree $j$ in $x,y,t$, hence $\det H$ is 
also a polynomial. Consequently, $u(x,y,t)$ in \eqref{H} is a rational function
of its arguments and decays as $(x^2+y^2)^{-1}$ for fixed $t$.

It is evident from \eqref{H} that the matrix $H$ is also positive
definite. We explore further its underlying structure.
The expression for $H_{ij}$ in \eqref{H} can be expanded using
Leibnitz rule of derivatives as  
\begin{equation*}
H_{ij} = \sum_{r=0}^{m_i+m_j}\frac{1}{(2b)^r} 
\sum_{s=0}^r\binom{r}{s}\partial_x^sp_{m_i}\partial_x^{r-s}\bar{p}_{m_j} 
= \sum_{r=0}^{m_i}\sum_{s=0}^{m_j} 
\frac{1}{(2b)^{r+s}}
\binom{r+s}{s}\partial_x^rp_{m_i}\partial_x^s\bar{p}_{m_j}\,.   
\end{equation*}
Notice that the upper limits of both sums in the last equality above can be 
extended to $m_n$ without any loss of generality because
from \eqref{theta} and \eqref{propa} it follows
that $\partial_x^rp_j=i^r\partial^r_{\theta_1}p_j=i^rp_{j-r}$ if $r\leq j$
and $\partial_x^rp_j=0$ if $r>j$. Next consider $n$ complex vectors in $\mathbb{C}^{m_n+1}$
\[P_i := (p_{m_i}, \partial_xp_{m_i}, \cdots, \partial_x^{m_n}p_{m_i})^T\,,
\qquad i=1,2,\ldots, n \]
then the elements of $H_{ij}$ are given by the inner products  
\begin{equation}
H_{ij} = P_j^{\dagger}CP_i\,, \qquad 
\quad C_{rs} = \frac{1}{(2b)^{r+s}}\binom{r+s}{s} \,, 
\quad r,s=0,1,\ldots,m_n\,,
\label{gram}
\end{equation}
where $C$ is a real, symmetric $(m_n+1) \times (m_n+1)$ matrix. 
The matrix $H$ is known as the Gram matrix and $\det H$ is called
the Gramian of the vectors $P_1,P_2,\ldots,P_n$ which are linearly 
independent since $H$ is positive definite. A geometric interpretation
of the KPI $\tau$-function $\det H$ is given as follows:
Consider the complex vector space $\mathbb{C}^{m_n+1}$ endowed with
a hermitian inner product given by the matrix $C$. Then
$V_P := {\rm span}_{\mathbb{C}}\{P_1,P_2,\ldots,P_n\}$ is a complex
$n$-dimensional subspace of $\mathbb{C}^{m_n+1}$, i.e., a point in the
complex Grassmannian ${\rm Gr}_{\mathbb{C}}(n,m_n+1)$. A natural
representation of
$V_P \in {\rm Gr}_{\mathbb{C}}(n,m_n+1)$ is given by the 
$m_n+1 \times n$ matrix $P := (P_1,P_2,\ldots,P_n)$ whose 
$j^{\rm th}$ column is the vector $P_j$ such that the entries 
of $P$ are given by
\begin{equation}
P_{rj} = \partial_x^rp_{m_j} =  i^rW_{rj}\,, \quad
W_{rj} = \left\{\begin{matrix} p_{m_j-r}\,, 
& \quad r\leq m_j \cr 0\,, & \quad r>m_j \end{matrix} \right. \,,
\quad r=0,1,\ldots,m_n\,, \quad j=1,2,\ldots,n \,.  
\label{P}
\end{equation}
Since the point $V_P \in {\rm Gr}_{\mathbb{C}}(n,m_n+1)$
is independent of the choice of basis for the $n$-dimensional subspace,
its matrix representation is unique up to a right multiplication
$P \to PA$ for any $A \in {\rm GL}(n, \mathbb{C})$. Yet another way
to represent $V_P$ is via the Pl\"ucker map 
$V_P \to \Lambda^n\mathbb{C}^{m_n+1}$ whose image
is a one-dimensional subspace of the exterior product space 
$\Lambda^n \mathbb{C}^{m_n+1}$. Explicitly, this is given by the $n$-form
\[\omega_P = P_1 \wedge \cdots \wedge P_n = \hspace{-0.3in} 
\sum_{0\leq i_1<\cdots<i_n\leq m_n}\hspace{-0.3in} 
P(i_1,\ldots,i_n)\,e_{i_1}\wedge \cdots e_{i_n}\,, \]
where $P(i_1,\ldots,i_n)$ are the $n \times n$ maximal minors of
the matrix $P$, i.e., the determinants of the $\binom{m_n+1}{n}$
submatrices of $P$ with rows indexed by $0\leq i_1<\cdots<i_n\leq m_n$,
and $\{e_j\}_{j=0}^{m_n}$ is the standard basis of $\mathbb{C}^{m_n+1}$.
The maximal minors $P(i_1,\ldots,i_n)$ are called
the Pl\"ucker co-ordinates of the point $V_P \in {\rm Gr}_{\mathbb{C}}(n,m_n+1)$;
they are not all independent since they satisfy the Pl\"ucker relations
\[\sum_{r=1}^{n+1} (-1)^{r-1}P(l_1,\ldots l_{n-1},j_r)
P(j_1,\ldots,j_{r-1},j_{r+1},\ldots,j_{n+1}) =0\,,
\quad l_1<\ldots<l_{n-1}\,, \quad j_1<\ldots <j_{n+1} \,.\]
The Pl\"ucker co-ordinates are unique upto 
$P(i_1,\ldots,i_n) \to (\det A) P(i_1,\ldots,i_n), \,\, 
A \in {\rm GL}(n, \mathbb{C})$ due to a change of basis for $V_P$.
The inner product on $\mathbb{C}^{m_n+1}$ induces a natural inner
product on the vector space $\Lambda^n \mathbb{C}^{m_n+1}$ given by
$(v_1\wedge\cdots\wedge v_n, w_1\wedge\cdots\wedge w_n) = 
\det(v_i^{\dagger}C w_j)$. Then from \eqref{gram} it follows that
the polynomial form of the KPI $\tau$-function is simply the norm of 
the $n$-form $\omega_P$ representing the complex Grassmannian $V_P$. That is, 
\begin{equation*}
\tau_P(x,y,t) = \det H = (\omega_P,\omega_P) \,.
\end{equation*}

The hermitian matrix $C$ admits a unique
decomposition $C=U^{\dagger}DU$ where $U$ is a real, upper-triangular matrix
with 1's along its main diagonal and $D$ is a diagonal matrix with 
$D_{rr}=(2b)^{-2r}, \, r=0,1,\ldots,m_n$.
Hence, the matrix elements of $H$ from \eqref{gram} can be expressed as
\begin{equation}
H_{ij} = P_j^{\dagger}CP_i = Q_j^{\dagger}DQ_i \,, \quad
Q_i = UP_i\,, \quad 
U_{rs} = \left\{\begin{matrix}\frac{1}{(2b)^{s-r}}\binom{s}{r}\,,
& \quad r\leq s\cr 0\,, & \quad r>s \end{matrix} \right. \,,
\qquad r,s=0,1,\ldots,m_n \,.  
\label{Q}
\end{equation}
Let $Q=UP$ where $P$ is defined in \eqref{P}, then $H=Q^{\dagger}DQ$
which results in an explicit expression for $\tau_P$ as a sum of squares 
after using the Cauchy-Binet formula for determinants
\begin{subequations}
\begin{equation}
\tau_P=\det H = \hspace{-0.3in} \sum_{0 \leq l_1 < \cdots <l_n \leq m_n}
\hspace{-0.1in} \frac{|Q(l_1,\cdots,l_n)|^2}{(2b)^{2(l_1+\cdots+l_n)}}\,, 
\qquad Q(l_1,\cdots,l_n)=\hspace{-0.3in}\sum_{0 \leq r_1 < \cdots <r_n \leq m_n}
\hspace{-0.2in} U\binom{l_1 \cdots l_n}{r_1 \cdots r_n}P(r_1,\ldots,r_n) \,,
\label{square-a}
\end{equation}
where $U\tbinom{l_1 \cdots l_n}{r_1 \cdots r_n}$ is the $n \times n$ minor
of $U$ obtained from the submatrix whose rows and columns are indexed by
$(l_1,\dots l_n)$ and $(r_1,\dots r_n)$, respectively. Since $U$ is
upper triangular with $1$'s along its diagonal, it follows that
the principal minors $U\tbinom{l_1 \cdots l_n}{l_1 \cdots l_n}=1$ and
$U\tbinom{l_1 \cdots l_n}{r_1 \cdots r_n} \neq 0$ if and only if
$l_j \leq r_j$ for each $j=1,\ldots,n$.    

Since the matrix elements $H_{ij}$ in \eqref{H} are polynomials in 
the $p_n$, their complex conjugates and inverse powers of $2b$, it follows that
$\tau_P$ is a positive definite, weighted homogeneous polynomial of degree 
$2(m_1+m_2+\cdots+m_n)$ in $\theta_j, \bar{\theta}_j \, j=1,2,\ldots,m_n$ 
and $\Im(k)=b$ with weight$(\theta_j) =$ weight$(\bar{\theta}_j) =j$, and 
weight$(b)= -1$. Moreover, each of the generalized Schur polynomials $p_{m_j}$ may be
parametrized by an independent set of arbitrary complex parameters
$\{\gamma_r \in \mathbb{C}, \, r=1,2,\ldots,m_j\}$
that is distinct for each $j=1,2,\ldots,n$, 
by choosing $n$ distinct arbitrary functions for $\gamma(k)$
in the expression for $\theta$ in \eqref{sp}. In that case, $\tau_P$ would
depend on at most $2(m_1+m_2+\cdots+m_n)$ real parameters and $k := a+ib$.
While this is true when $n=1$, the actual number of independent real
parameters in $\tau_P$ is less than $2(m_1+m_2+\cdots+m_n)$ if $n>1$.
This will be clear from further analysis of the $\tau$-function in
\eqref{square-a}, which will be done next.
  
Before proceeding further, it is convenient to introduce a total
(lexicographic) ordering for the multi-index sets   
${\bf l}:=(l_1 \cdots l_n), \,\, 0\leq l_1<l_2 \cdots<l_n \leq m_n$,
defined as follows. 
\begin{definition}[Lexicographic Ordering]
Given two distinct multi-index sets 
${\bf l}, {\bf r}$, if $j$ is the {\it earliest} index where ${\bf l}$
and ${\bf r}$ differ, then ${\bf l} < {\bf r}$ if and only if 
$l_j<r_j$. 
\end{definition}
\noindent Notice that there is a unique first element with respect
to this total ordering namely $(0 1 \cdots n-1)$ which will henceforth
be denoted by ${\bf 0}$. 
\begin{example}
Suppose $n=2, m_n=3$. That is, each multi-index set is 
denoted by ${\bf l} = (l_1 l_2), \,\, 0 \leq l_1 < l_2 \leq 3$. Then 
${\bf 0} = (01) < (02) < (03) < (12) < (13) < (23)$ is the
lexicographic ordering of the six 2-index sets.
\end{example}
\noindent Using the above notations \eqref{square-a} can be re-expressed 
in an abridged form as
\begin{equation}
 \tau_P = \det H = \sum_{{\bf r} \geq {\bf 0}} 
\frac{|Q({\bf r})|^2}{(2b)^{2|{\bf r}|}}\,, \qquad \quad
Q({\bf r}) = \sum_{{\bf s} \geq {\bf r}}U\tbinom{{\bf r}}{{\bf s}}P({\bf s}) \,, 
\label{square-b}
\end{equation}
\label{square}
\end{subequations}
with $|{\bf r}| := r_1+\cdots + r_n$. The leading polynomial term
in $\tau_P$ is given by $\sf{|Q({\bf 0})|^2}{(2b)^{n(n-1)}}$, where  
\[ Q(0 1 \cdots n-1) = Q({\bf 0}) = P({\bf 0}) + 
\sum_{{\bf s} > {\bf 0}}U\tbinom{{\bf 0}}{{\bf s}}P({\bf s}) \] 
since $U\tbinom{{\bf 0}}{{\bf 0}} = U\tbinom{0 1\cdots n-1}{0 1 \cdots n-1}=1$.
The maximal minors $P({\bf r})=P(r_1 \cdots r_n)$ for any 
multi-index set ${\bf r}$ and in particular $P({\bf 0})$, are given by 
$P({\bf r})=i^{|{\bf r}|}W({\bf r})$ and 
$P({\bf 0})= i^{\sf{n(n-1)}{2}}W({\bf 0})$,
where 
\begin{equation}
W({\bf r}) = \left|\begin{matrix}
p_{m_1-r_1} & p_{m_2-r_1} & \cdots & p_{m_n-r_1}\\
p_{m_1-r_2}& p_{m_2-r_2} & \cdots & p_{m_n-r_2}\\
\vdots& \vdots & \vdots & \vdots  \\
p_{m_1-r_n}& p_{m_2-r_n} &\cdots & p_{m_n-r_n} 
\end{matrix}\right| \,, 
\quad  
W({\bf 0}) = \left|\begin{matrix}
p_{m_1} & p_{m_2} & \cdots & p_{m_n}\\
p_{m_1-1}& p_{m_2-1} & \cdots & p_{m_n-1}\\
\vdots& \vdots & \vdots & \vdots  \\
p_{m_1-n+1}& p_{m_2-n+1} &\cdots & p_{m_n-n+1} 
\end{matrix}\right|   
\label{schur} 
\end{equation} 
are the $n \times n$ minors of the $m_n+1 \times n$ matrix $W$ defined 
in \eqref{P}.
Notice that $P({\bf r})$ is a weighted homogeneous polynomial in 
$\theta_1, \theta_2, \ldots, \theta_{m_n}$ of degree 
$(m_1+\cdots m_n)-|{\bf r}|$. Particularly when ${\bf r} = {\bf 0}$,
the degree of the leading principal minor $P({\bf 0})$ given by  
$N := (m_1+\cdots m_n)-\sf{n(n-1)}{2}$ is the highest, whereas
when ${\bf r} = {\bf m}=(m_1, \cdots,m_n)$, $P({\bf m})$ is the
last non-vanishing minor in the lexicographic ordering with   
$|P({\bf m})|=1$. Finally, note that one can factor out
$(2b)^{-n(n-1)}$ from $\tau_P$ in \eqref{square} as well as the 
gauge freedom:\, $P({\bf r}) \to (\det A) P({\bf r}) \Rightarrow
\tau_P \to |\det A|^2 \tau_P, \,\, A \in GL(n, \mathbb{C})$ since 
neither of these contribute to the solution \eqref{u}. 
The results obtained thus far are summarized below.
\begin{proposition}
The KPI $\tau$-function $(2b)^{n(n-1)}\tau_P(x,y,t)$ is a positive definite,
weighted homogeneous polynomial in $x,y,t$ and parameter $b=\Im(k)$
of degree $2N = 2(m_1+\cdots m_n)-n(n-1)$ where the weights of $x,y,t$
are $1,2,3$ respectively and weight$(b)=-1$.
The corresponding KPI solution $u(x,y,t)$ is a non-singular rational
function in the $xy$-plane, decaying as $(x^2+y^2)^{-1}$ for any fixed
value of $t$. Moreover, the KPI $\tau$-function can be expressed as a 
graded sum of squares of real polynomials in $x,y,t$ and $b$ whose (weighted)
degrees are in descending order, namely $2N, 2(N-1), \ldots, 0$; the gradation
is marked by dividing each square by an appropriate factor 
$(2b)^{2r}, \, r=0,1,\ldots,N$ such that its overall degree is $2N$.
\label{prop1}
\end{proposition} 
\noindent In order to estimate the number of independent parameters
in the multi-lump $\tau$-function it suffices to consider the leading 
term $P({\bf 0})$ contributing to $\tau_P$ in \eqref{square}. As mentioned 
earlier, each of the polynomials
$p_{m_j}$ in $P({\bf 0})$ is parametrized by a distinct set of $m_j$ 
complex parameters for $j=1,\ldots,n$. However when $n=2$, it is
possible to eliminate 1 complex parameter from the determinant $W({\bf 0})$
by elementary column operations and subsequently redefining the
arbitrary parameters so that $P({\bf 0})$ depends on only $m_1+m_2-1$ 
complex parameters. Similarly, one can eliminate 2 parameters when $n=3$,
and by induction one finds that in general, $P({\bf 0})$ depends on 
$N$ free complex parameters. Instead of a technical proof, the following 
example illustrates the basic idea.
\begin{example}
Let $n=2$ and $W({\bf 0})=\left|
\begin{smallmatrix} p_2 & p_4 \\ p_1 & p_3 \end{smallmatrix}\right|,$
where $\theta_1, \theta_2$ in $p_2$ are linear in the arbitrary parameters 
$\gamma_1, \gamma_2$, and $\theta_j$ in $p_4$ are 
linear in another set of parameters 
$\tilde{\gamma}_j, \, j=1,\ldots,4$. Thus there are 6 free 
complex parameters, one of which can be eliminated by the following
process. Denote by   
$\delta_j = \tilde{\gamma}_j - \gamma_j, \, j=1,2$ and use \eqref{propc}
to expand $p_4(\theta_1+\delta_1, \theta_2+\delta_2,\theta_3,\theta_4)
= p_4(\theta_1,\ldots,\theta_4) + a_1p_3+a_2p_2+a_3p_1+a_4p_0$
and similarly $p_3(\theta_1+\delta_1, \theta_2+\delta_2,\theta_3) = 
p_3(\theta_1,\ldots,\theta_3) + a_1p_2+a_2p_1+a_3p_0$
where $a_j = p_j(\delta_1,\delta_2,0,0), \, j=1,\ldots,4$.
Substituting these in $W({\bf 0})$, the coefficient 
$a_2=\sf{\delta_1^2}{2}+\delta_2$ in the second column is then eliminated by 
elementary column operation.
In order to eliminate $\delta_2$ completely from $W({\bf 0})$, next
redefine the arbitrary parameter $\tilde{\gamma}_3 \to \tilde{\gamma}_3+a_3$
and use \eqref{propc} once more, to absorb the coefficient $a_3$.
The coefficient $a_4$ can similarly be absorbed by redefining $\tilde{\gamma}_4$.
Thus $W({\bf 0})$ depends on only the parameter $\delta_1=a_1$ together
with $\gamma_1, \gamma_2$ and the (redefined) 
parameters $\tilde{\gamma}_3, \tilde{\gamma}_4$. 
\end{example} 
\noindent This process of eliminating variables described above can be generalized 
using induction in a straight-forward although tedious fashion. This leads
to the fact that the $\tau$-function $\tau_P$ in \eqref{square} and
the corresponding KPI solution given by \eqref{u} is parametrized by
$2N$ independent real parameters. In the next subsection we will
give examples which indicate that the rational solutions KPI constructed
from $\tau_P$ admit $N$ distinct peaks in the $xy$-plane which evolve
in time. The $2N$ real parameters can be chosen arbitrarily to
specify the locations of the peaks in the $xy$-plane at a given
time $t=t_0$.  
\begin{remark}
\begin{itemize}
\item[(a)] Several authors~\cite{DM11,C18,YY21} have made different
choices for the $\phi_n$ in \eqref{sp} to construct the $n \times n$ 
matrix $M$ \eqref{M}. These are essentially of the form
$$\phi_n = \sum_{j=0}^n a_{nj}(k)p_j(\theta_1,\cdots,\theta_j)\exp (i\theta(k))$$  
where the sum can be reduced to a single generalized Schur polynomial
$p_n(\theta_1+h_1,\cdots,\theta_n+h_n)$ using \eqref{propc} for
suitable choices for the $h_j$'s which can then be absorbed in the arbitrary
constants $\gamma_j$ appearing in the $\theta_j$ variables for $j=1,\ldots,n$.
Thus, the classes of KPI rational solutions obtained by
such choices are the same as those obtained from \eqref{u}--\eqref{tau}
which are simpler.

\item[(b)] The determinant $W({\bf 0})$ introduced in \eqref{schur}
is the wronskian form of the $\tau$-function for the KPI equation. 
However, the corresponding rational
solution $u(x,y,t)$ given by \eqref{tau} is singular in the $xy$-plane for 
any given $t$. The singularities occur at the zeros of the $W({\bf 0})$.

\item[(c)] Since the KPI equation admits a constant solution 
$u(x,y,t)=c, \,\, c \in \mathbb{R}$, one may also consider
the multi-lump solutions in a constant background. 
In this case, equation \eqref{u} should read as $u = c + 2(\ln \tau)_{xx}$ 
where $\tau$ is given by \eqref{tau} except that the $\phi_n$
in \eqref{M} are given by  
$\phi_n = \sf{1}{n!}\,\partial_k^n\exp(i\theta') = p'_n\exp(i\theta')$ where
$\theta' = kx+(k^2-c)y+(k^3-\sf34 kc)t+\gamma(k)$.
The new generalized Schur polynomials $p'_n$ are defined in the same way
as before, the only change being $x \to x-\sf34 ct$ in the definition
of $\theta_1$ in \eqref{theta}.

\item[(d)] The construction of $\tau_P$ in Section 2.2 requires that the
positive integers $m_j$ labeling the polynomials $p_{m_j}$ be distinct, else 
some columns of the matrices $P$ and $H$ will be identical causing
$\tau_P=\det H=0$. However, it is possible to choose $m_1=0$ instead
of $m_1 \geq 1$. Then the columns $P_1=Q_1=e_0 \in \mathbb{C}^{m_n+1}$
in \eqref{Q} since $p_{m_1}=p_0=1$. As a result the $n \times n$ minors 
of $P$ will have the
form $P(l_1,\cdots,l_n) = 0, \,\, l_1 \neq 0$ and if $l_1=0$, then
$P(0,l_2,\cdots,l_n) = \tilde{P}(l_2,\cdots,l_n), \,\, 
1 \leq l_2 < \cdots < l_n \leq m_n$ where $\tilde{P}$ is
the $m_n \times (n-1)$ matrix obtained by stripping off the first row
and columns of $P$ in \eqref{P}. Thus after re-indexing the columns
of $\tilde{P}$ as $\tilde{m}_j = m_{j+1}-1, \,\, j=1,\ldots,n-1$
(and pulling out an insignificant factor of $i$), one obtains
a new $m_n \times (n-1)$ $P$-matrix instead of the one in \eqref{P}.
This is equivalent to constructing a {\it reduced} KPI solution from the
generalized Schur polynomials $p_{m_{j+1}-1}, \, j=1,\ldots,n-1$
which forms a new $(n-1) \times (n-1)$ matrix $M$ in \eqref{M}.  
If now $\tilde{m}_2=0$, i.e., $m_2=1$, then the reduction 
process described above is repeated once more, and in fact iterated
until $m_j > j-1$ to obtain a non-trivial $\tau$-function. Clearly, the initial 
choice of $\{m_1,\ldots,m_n\} = \{0,1,\ldots,n-1\}$ leads trivially to
$(2b)^{n(n-1)}\tau_P=1$. 

\end{itemize}
\end{remark}
\subsection{Examples of multi-lumps}
We further illustrate the construction outlined in Section 2.2 
with some examples of the KPI multi-lump solutions that are obtained
via this method. It is convenient to first introduce a set of coordinates 
$r,s$ defined in terms of co-moving coordinates $x',y'$ as follows 
\begin{equation}
r=x'+2ay', \quad s = 2by', \quad \text{where} \quad
x'=x-3(a^2+b^2)t, \quad y'=y+3at \,, 
\label{rs}
\end{equation}
and $a=\Re(k), b=\Im(k)$. It will be clear from the discussion below that 
instead of the $x,y$, the natural coordinates to use are the $r,s$
coordinates, in terms of which $\theta_j$ in \eqref{theta} are given by
\begin{equation}
\theta_1 = i(r+is+\gamma_1),  \quad \theta_2 = \frac{is}{2b} -3bt +i\gamma_2,
\quad \theta_3=i(t+\gamma_3)
\label{thetars}
\end{equation}
\subsubsection{The 1-lump solution} The simplest rational solution of the 
class described in this paper is obtained by choosing $n=m_1=1$.
In this case both $M$ and $H$ are $1 \times 1$ matrices in \eqref{M}
and \eqref{H}, respectively. Moreover, in \eqref{P} 
$P = (p_1, ip_0)^T$ is a single column vector. Then \eqref{Q} together with 
the definition of $D$ above it, yields
\[ U = \begin{pmatrix} 1 & \sf{1}{2b} \\ 0 & 1 \end{pmatrix}\,, \qquad
Q = UP\,, \quad 
D=\begin{pmatrix} 1 & 0 \\ 0 & \sf{1}{4b^2} \end{pmatrix} \,, \qquad  
H= Q^{\dagger}DQ = |p_1+\sf{i}{2b}p_0|^2+\sf{1}{4b^2}p_0 \,,   
\]
where the last expression is also the sum of squares form for 
$\tau_P = \det H$ as in \eqref{square} since $H$ is a $1 \times 1$ matrix.
After using $p_0=1, p_1=\theta_1$ and \eqref{thetars}, one obtains 
the $\tau$-function
\[ \tau_1= |\theta_1+\sf{i}{2b}|^2 +\sf{1}{4b^2} = 
(r+\sf{1}{2b}-r_0)^2+(s-s_0)^2 + \sf{1}{4b^2}, 
\qquad r_0=-\Re(\gamma_1), \,\, s_0=-\Im(\gamma_1) \,,\]
and then from \eqref{u}, the KPI solution is given by
\begin{equation}
u_1(x,y,t) = 2(\ln \tau_1)_{xx} =  
4\,\frac{-(r+\sf{1}{2b}-r_0)^2+(s-s_0)^2+\sf{1}{4b^2}}
{\big[(r+\sf{1}{2b}-r_0)^2+(s-s_0)^2+\sf{1}{4b^2}\big]^2}\,.    
\label{u1}
\end{equation}
\begin{figure}[t!]
\begin{center}
\raisebox{0.1in}{\includegraphics[scale=0.65]{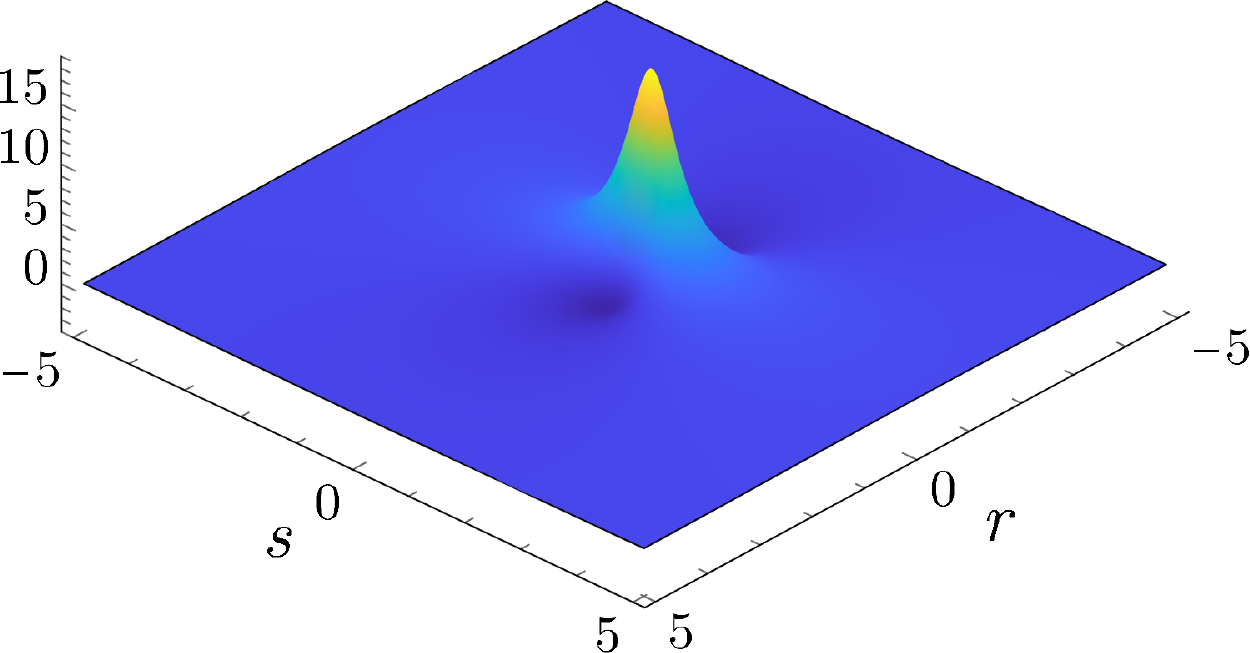}} \qquad
\includegraphics[scale=0.6]{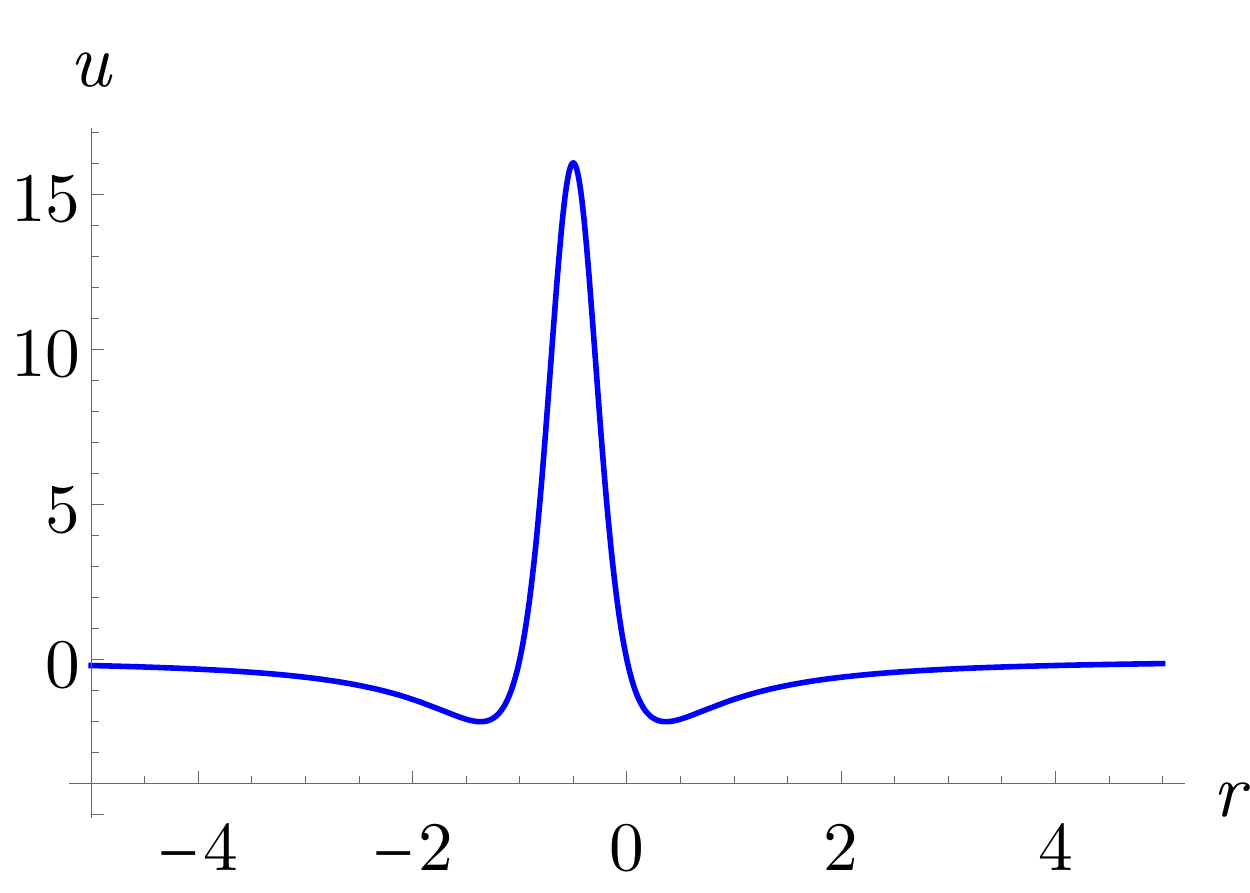}
\end{center}
\vspace{-0.2in}
\caption{$1$-lump solution of the KPI equation:\, $a=\gamma_1=0, b=1$.
Right: Vertical cross-section showing the maximum and two minima.}
\label{1lump}
\end{figure}
Notice that the solution is {\it stationary} in the $rs$-plane since
the only time dependence enters via the co-moving coordinates $x',y'$.
Hence, \eqref{u1} represents a rational traveling waveform with a single peak 
(maximum) at $(r_0-\sf{1}{2b}, s_0)$, of height $16b^2$, and two local minima 
symmetrically located from the peak at 
$(r_0-\sf{1}{2b}\pm \sf{\sqrt{3}}{2b}, \, s_0)$ 
and depth $-2b^2$ determined by $b$ and
the complex parameter $\gamma_1$ as illustrated in Figure~\ref{1lump}.
The wave moves in the $xy$-plane with
a uniform velocity $(3(a^2+b^2), -3a)$ at an angle $\tan^{-1}(-a/(a^2+b^2))$
with the positive $x$-axis. This solution is referred to as the $1$-lump 
solution since it has a 
single peak which also coincides with the fact that $N=1$
in this case. It will be shown in Section 4 that $N$ indeed
represents the number of peaks when $|t| \gg 1$ for this class of rational
solutions.

Since $u_1$ is the $x$-derivative of
the rational function $F_{1x}/F_1$ which decays as $|x| \to \infty$ 
for all $y,t$, one has
that $\int_{-\infty}^\infty u_1\,dx=0$. However, $u_1$ is not a 
$L^1(\mathbb{R}^2)$ function although $u_1 \in L^2(\mathbb{R}^2)$ and 
$\int\!\!\int_{\mathbb{R}^2}u_1^2 = 16\pi b$; the latter is a conserved 
quantity for the KPI equation \eqref{kp}.  
\subsubsection{A 2-lump solution}    
Next we consider the simplest case for $n=2$ where $m_1=1, m_2=2$.
Here $M$ is a $2 \times 2$ matrix and there are 2 column
vectors in $P$ which is a $3 \times 2$ matrix. The matrices $P, U$
and $D$ are given by
\[ P = \begin{pmatrix} p_1 & p_2 \\ ip_0 & ip_1 
\\ 0 & -p_0 \end{pmatrix}\,, \quad \qquad
U= \begin{pmatrix} 1 & \sf{1}{2b} & \sf{1}{4b^2} \\ 
0 & 1 & \sf{1}{b}  \\ 0 & 0 & 1 \end{pmatrix}\,, \qquad \quad
D= \mathrm{diag} (1, \,\, \sf{1}{2b}, \,\, \sf{1}{4b^2})  
\]
from which the $2 \times 2$ minors are computed as follows: 
\begin{gather*}
P({\bf 0})= P(01) = i(p_1^2-p_0p_2), \,\, P(02)=-p_1p_0, \,\,  
P(12)=-ip_0^2, \\
U\tbinom{01}{01}=U\tbinom{02}{02}=U\tbinom{12}{12} = 1,
\quad U\tbinom{02}{01} = U\tbinom{12}{01}=U\tbinom{12}{02}=0, \quad
U\tbinom{01}{02} = \sf{1}{b}=2U\tbinom{02}{12}, \quad 
U\tbinom{01}{12}=\sf{1}{4b^2} \,.  
\end{gather*}
\begin{figure}[t!]
\begin{center}
\includegraphics[scale=0.42]{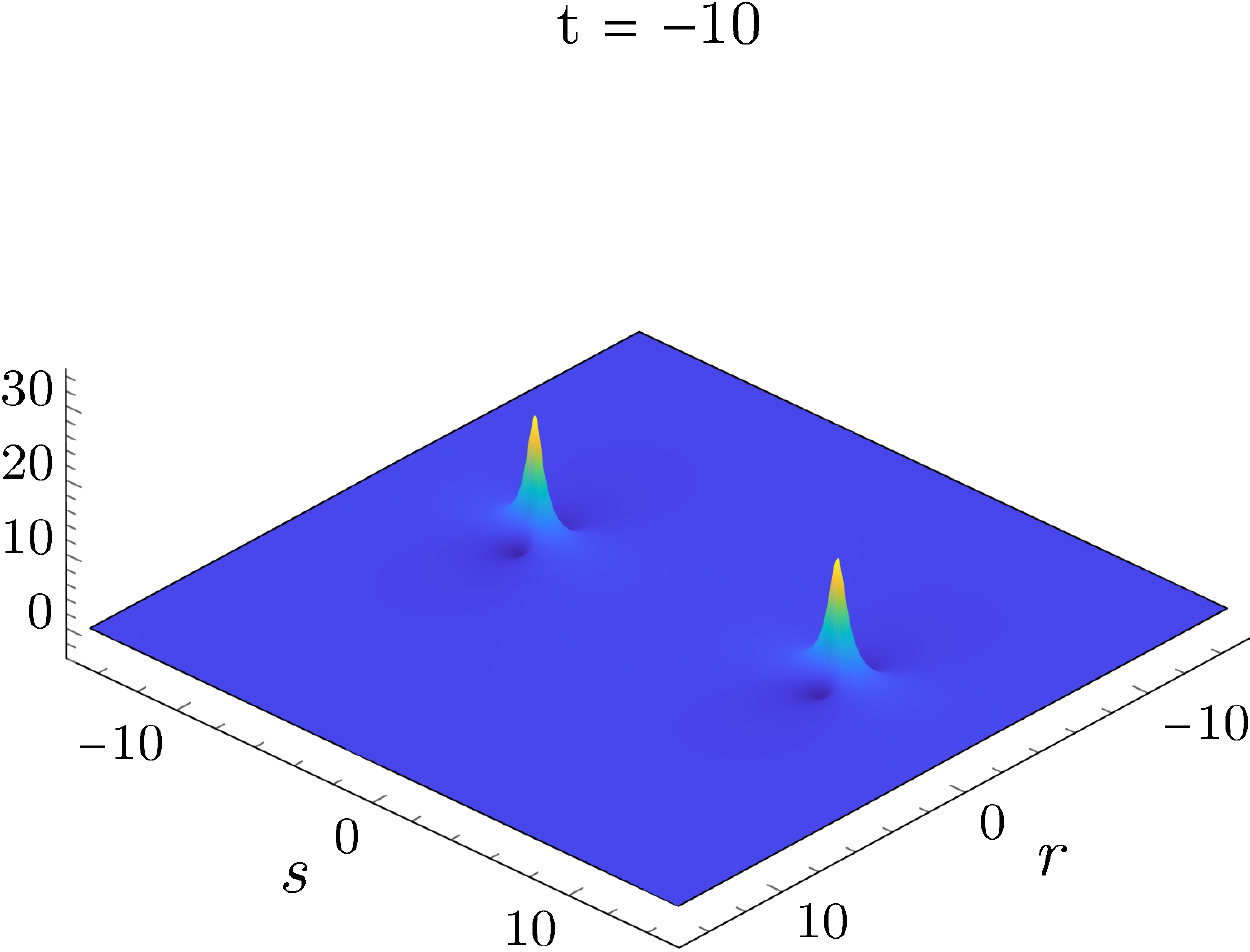} \quad
\includegraphics[scale=0.42]{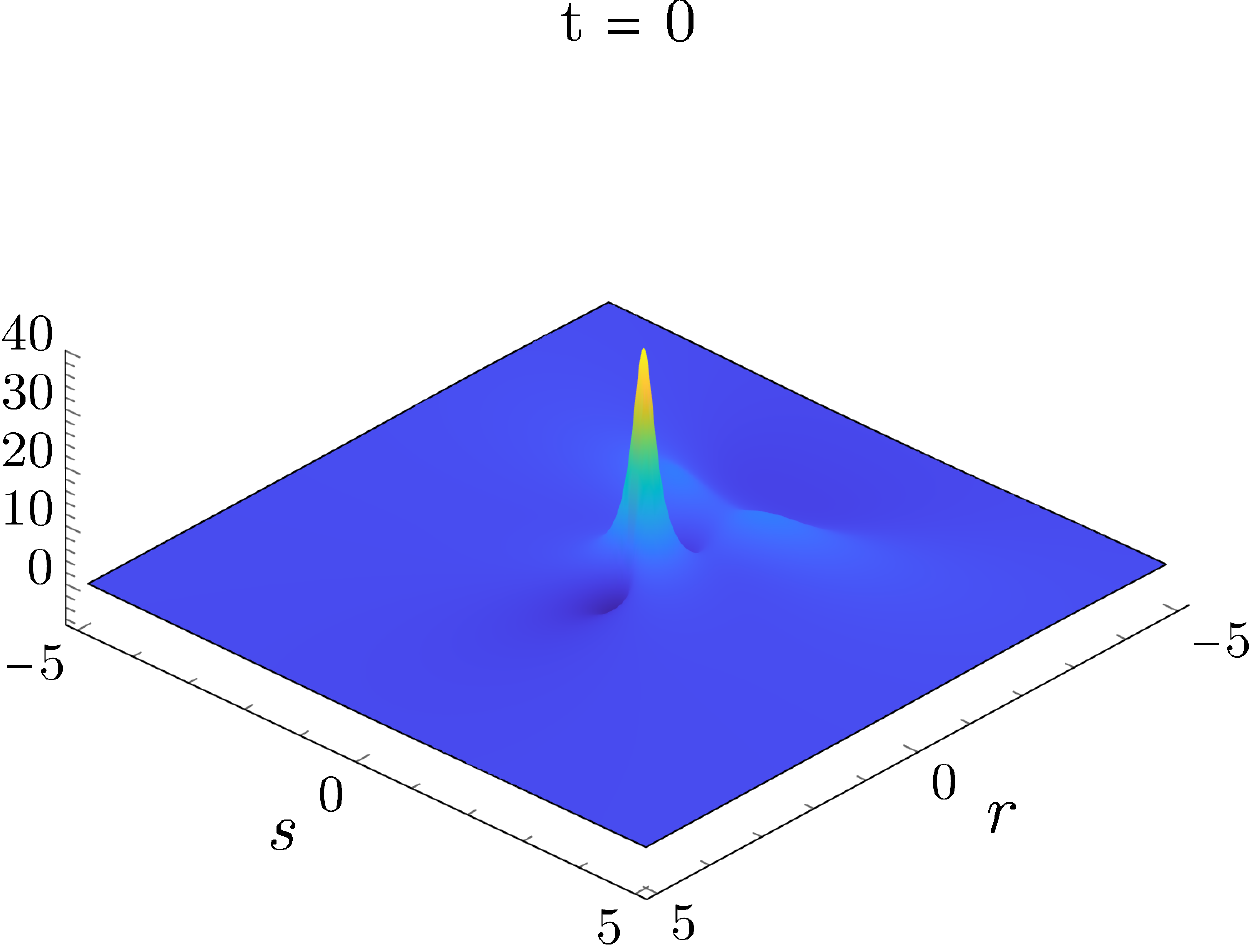} \quad
\includegraphics[scale=0.42]{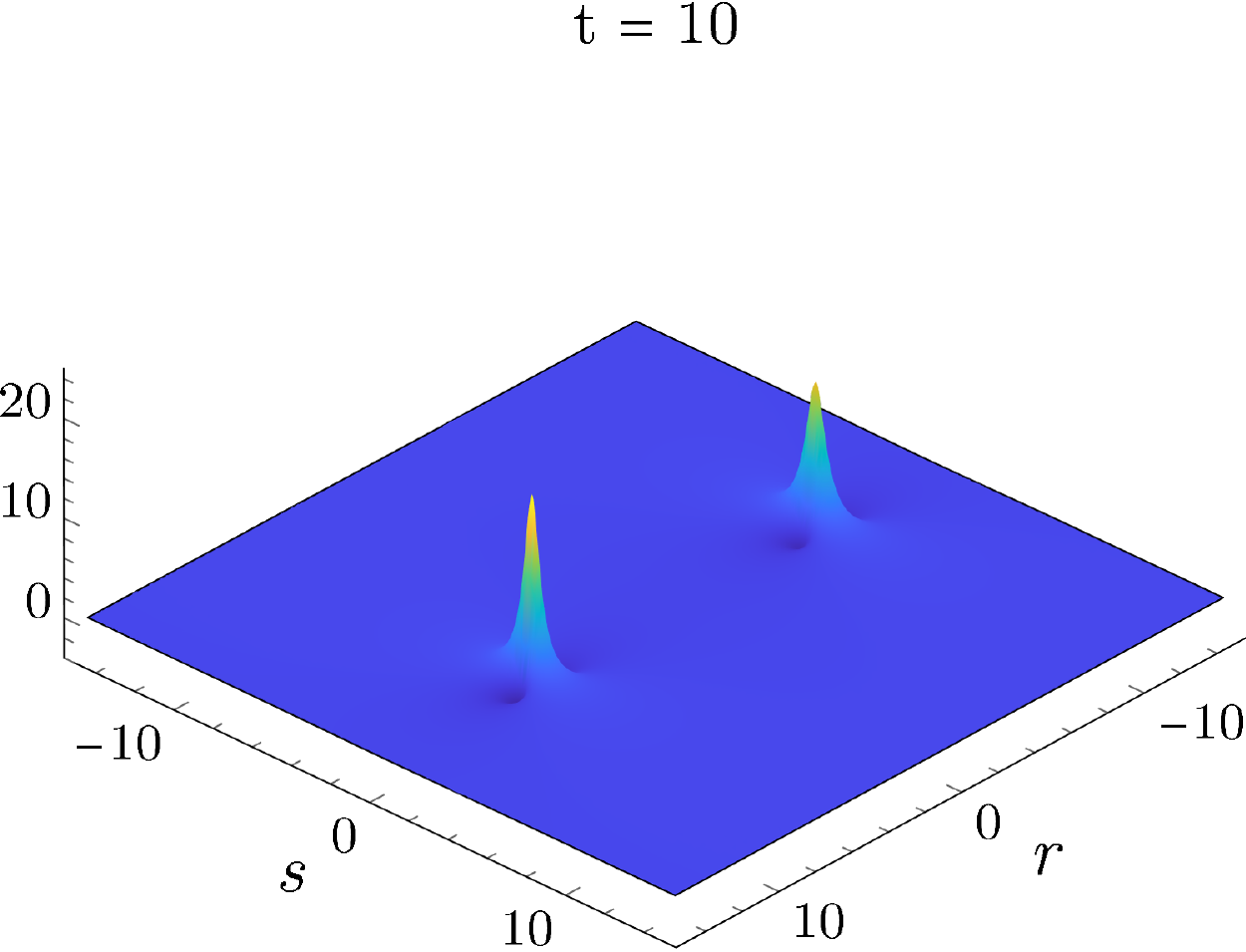}
\end{center}
\vspace{-0.2in}
\caption{A $2$-lump solution interacting in the $rs$-plane 
corresponding to $\tau_2$ in \eqref{tau2}.
Parameters:\, $a=\gamma_1=\gamma_2=0, b=1$.}
\label{2lump}
\end{figure}
From \eqref{pn} it follows that $P(01), P(02)$ and $P(12)$
are weighted homogeneous polynomials in $\theta_1,\theta_2$ of degree 
$(m_1+m_2)-|{\bf r}|$, with $|{\bf r}| = r_1+r_2 =1,2,3$, respectively,
as shown in Section 2.2.  Next, calculating the minors 
using \eqref{square} leads to the sum of squares expression
for the $\tau$-function  
\[\tau_P = \frac{|Q(01)|^2}{(2b)^2} + 
\frac{|Q(02)|^2}{(2b)^4}+\frac{|Q(12)|^2}{(2b)^6}=
\sf{1}{4b^2}\left[|(\sf{\theta_1^2}{2}-\theta_2)+\sf{1}{b}\theta_1
-\sf{1}{4b^2}|^2 + \sf{1}{4b^2}|\theta_1+\sf{i}{2b}|^2 
+ \sf{1}{16b^4}\right] \,,\]
where \eqref{pn} is used to express $p_0=1, p_1, p_2$. Note that $\tau_P$
is a weighted homogeneous polynomial of degree $2(m_1+m_2) =6$ 
in $\theta_1,\theta_2, b$ as mentioned below \eqref{square-a}.  
Finally using \eqref{thetars} and factoring out $\sf{1}{4b^2}$ from 
$\tau_P$ above one obtains the reduced $\tau$-function
\begin{equation}
\tau_2 = |\sf12((r+\sf{1}{b})^2-s^2)-3bt-\sf{1}{4b^2}+ i(r+\sf{3}{2b})s|^2 + 
\sf{1}{4b^2}|(r+\sf{1}{2b})+is|^2+\sf{1}{16b^4} \,,
\label{tau2}
\end{equation}
where we have set the constants $\gamma_1=\gamma_2=0$ for simplicity.
In terms of the variables $r,s,t$ (equivalently, $x,y,t$)
and $b=\Im(k)$, $\tau_2$ is a weighted homogeneous polynomial
of degree $2N$ with $N=(m_1+m_2)-\sf{n(n-1)}{2}=2$ as in Proposition~\ref{prop1}.
The leading order (in $r,s,t$) contribution arises from the
from the first square term  $|Q({\bf 0})|^2$ through 
$P({\bf 0})=i\big(\sf{1}{2}\theta_1^2-\theta_2\big)$ which is of degree
$N=2$ in $\theta_1, \theta_2$.

Notice that unlike $\tau_1$, the polynomial $\tau_2(r,s,t)$ 
in \eqref{tau2} does depend 
explicitly on $t$ so that the solution $u_2$ obtained from \eqref{u} 
is non-stationary in the co-moving $rs$-plane. The explicit expression
for $u_2(x,y,t)$ is complicated, so it is not included here.
Figure~\ref{2lump} illustrates that
$u_2$ consists of 2 localized lumps (local maxima) along the 
$s$-axis that are well 
separated as $t \ll -1$; these lumps get attracted to each other
and overlap for finite $t$ to form a transient large amplitude peak 
that splits again into two localized lumps which then recede from
each other when $t \gg 1$ but along the $r$-axis. Furthermore, the 
height of each peak
also evolve with time and approaches the constant height of the 1-lump
solution as $|t| \to \infty$. The interaction process is an example of 
anomalous scattering rather than the usual solitonic interaction
of the simple $n$-lump solutions of KPI found in~\cite{MZBIM77}.
This particular $2$-lump solution and another solution which
corresponds to $n=1, m_1=2$, were found earlier in~\cite{GPS93,VA99,ACTV00}.
The latter solution was also studied more recently in \cite{CZ21}, where 
the structure of the solution was analyzed in details. 
The analysis for the solutions corresponding
to $\tau_2$ in \eqref{tau2} is similar, so we do not include it here.
The locations of the lump peaks for these two solutions
are related by time reversal symmetry $t \to -t$ when $|t| \gg 1$. 
Note that $N=2$ in both cases. The classification scheme developed in 
Section 3 will establish that there are {\it only} 2 possible $2$-lump 
solutions obtained from the method outlined in Section 2.2. 
\subsubsection{A $3$-lump solution}
Our next illustrative example is $n=2, m_1=1, m_2=3$ which
corresponds to $N=3$. In this case, the matrices $P, U, D$
are as follows:
\[P = \begin{pmatrix} p_1 & p_3 \\ ip_0 & ip_2 
\\ 0 & -p_1 \\ 0 & -ip_0 \end{pmatrix}\,, \quad \qquad
U= \begin{pmatrix} 1 & \sf{1}{2b} & \sf{1}{(2b)^2} & \sf{1}{(2b)^3}\\ 
0 & 1 & \sf{2}{2b} & \sf{3}{(2b)^2} \\ 0 & 0 & 1 & \sf{3}{2b}\\
0 & 0&0&1\end{pmatrix}\,, \qquad \quad
D= \mathrm{diag} \big(1,\,\,\sf{1}{(2b)^2},\,\,\sf{1}{(2b)^4}\,\,\sf{1}{(2b)^6}\big)\,,  
\]
and the index set for the maximal minors 
${\bf r}=(r_1 r_2) \in \{(01), (02), (03), (12), (13), (23)\}$.
Then using \eqref{square-b} and \eqref{pn} the $\tau$-function 
can be expressed as
\begin{equation}
\tau_3 = (2b)^2 \tau_P = 
|(\sf{\theta_1^3}{3}-\theta_3)+\sf{i}{b}(\theta_1+\sf{i}{2b})^2|^2 +
\sf{1}{(2b)^2}|(\theta_1+\sf{i}{b})^2+\sf{1}{4b^2}|^2 +
\sf{1}{(2b)^4}|2(\theta_1+\sf{i}{b})|^2 + \sf{1}{(2b)^6} \,.  
\label{tau3} 
\end{equation}
\begin{figure}[t!]
\begin{center}
\includegraphics[scale=0.44]{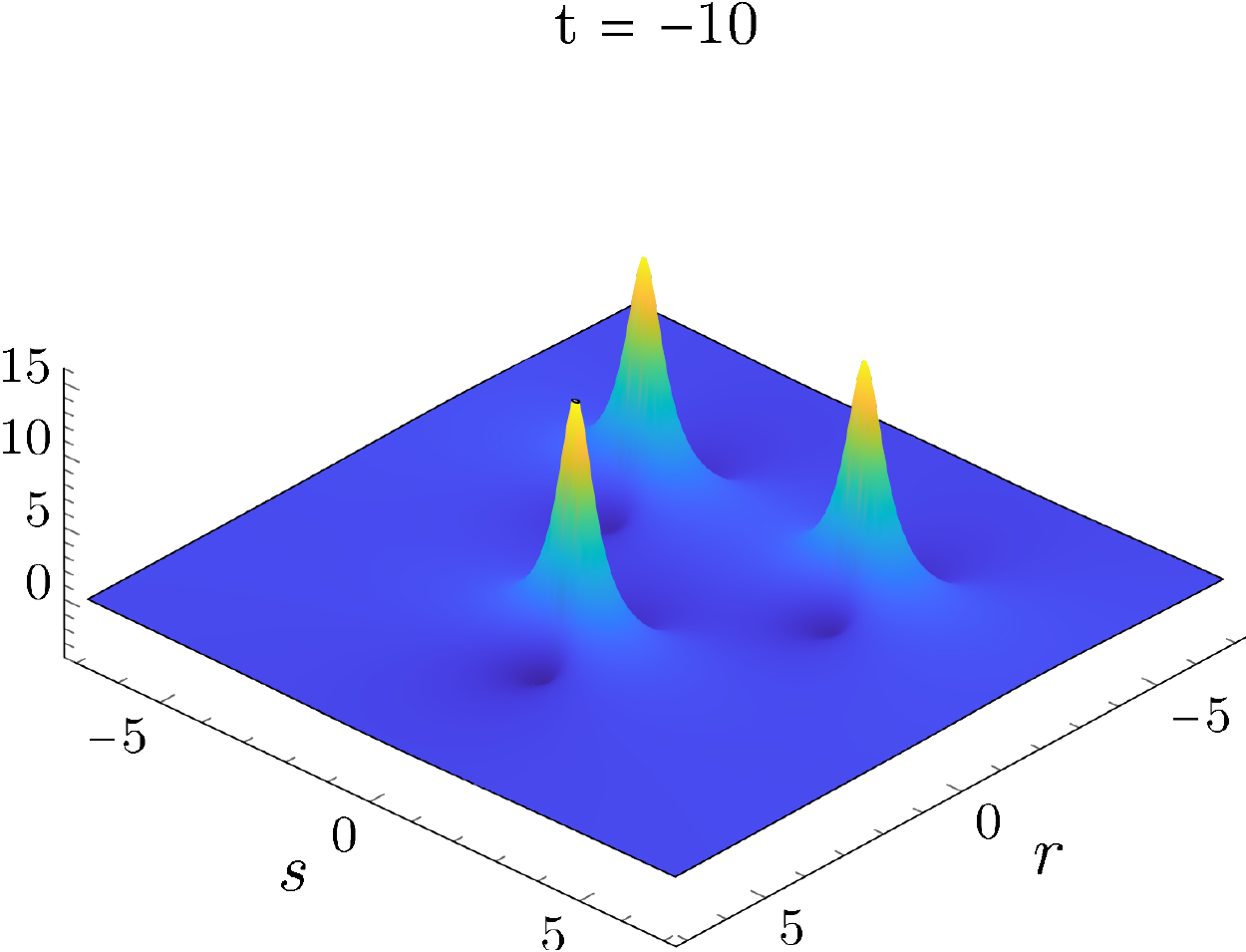}
\includegraphics[scale=0.44]{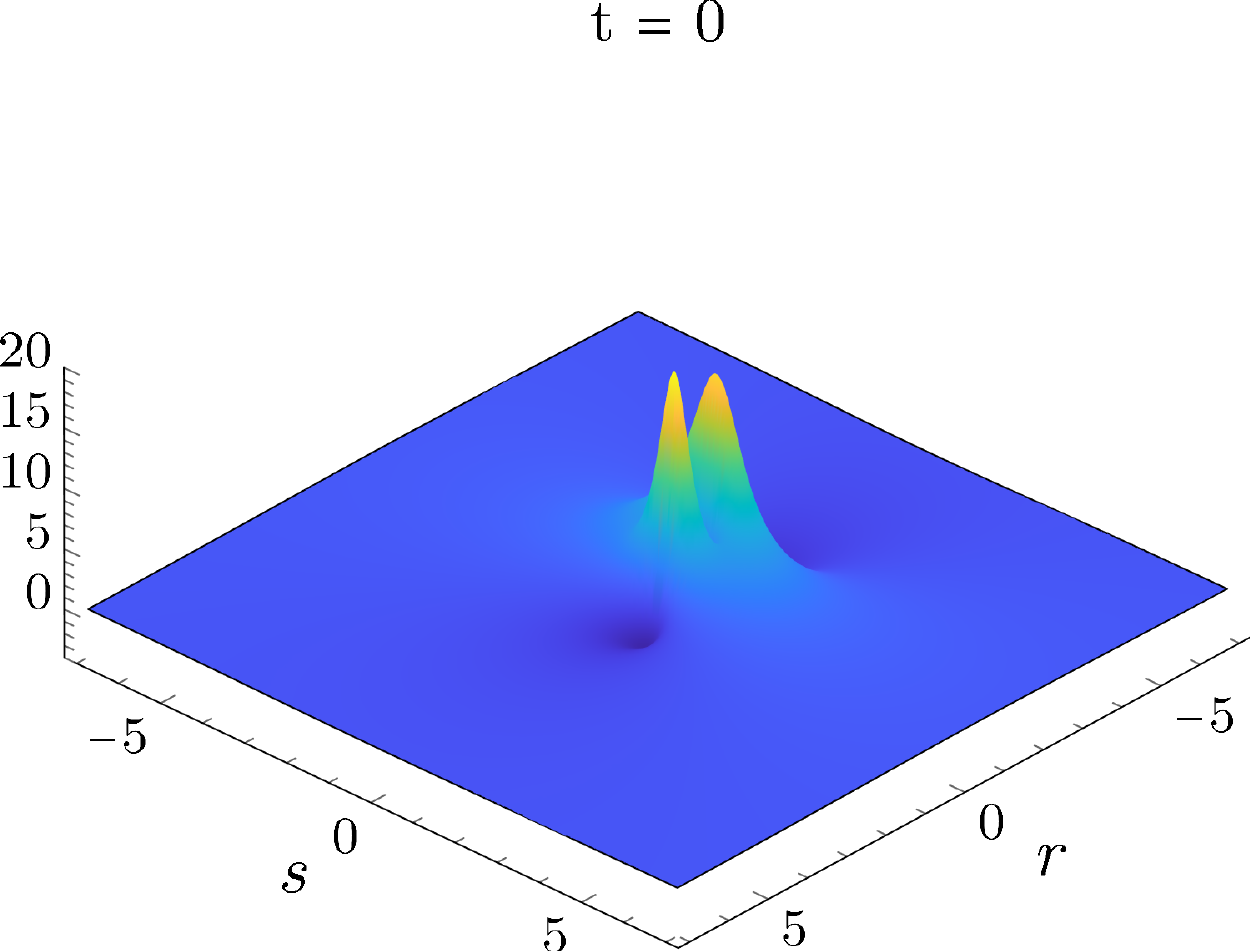}
\includegraphics[scale=0.44]{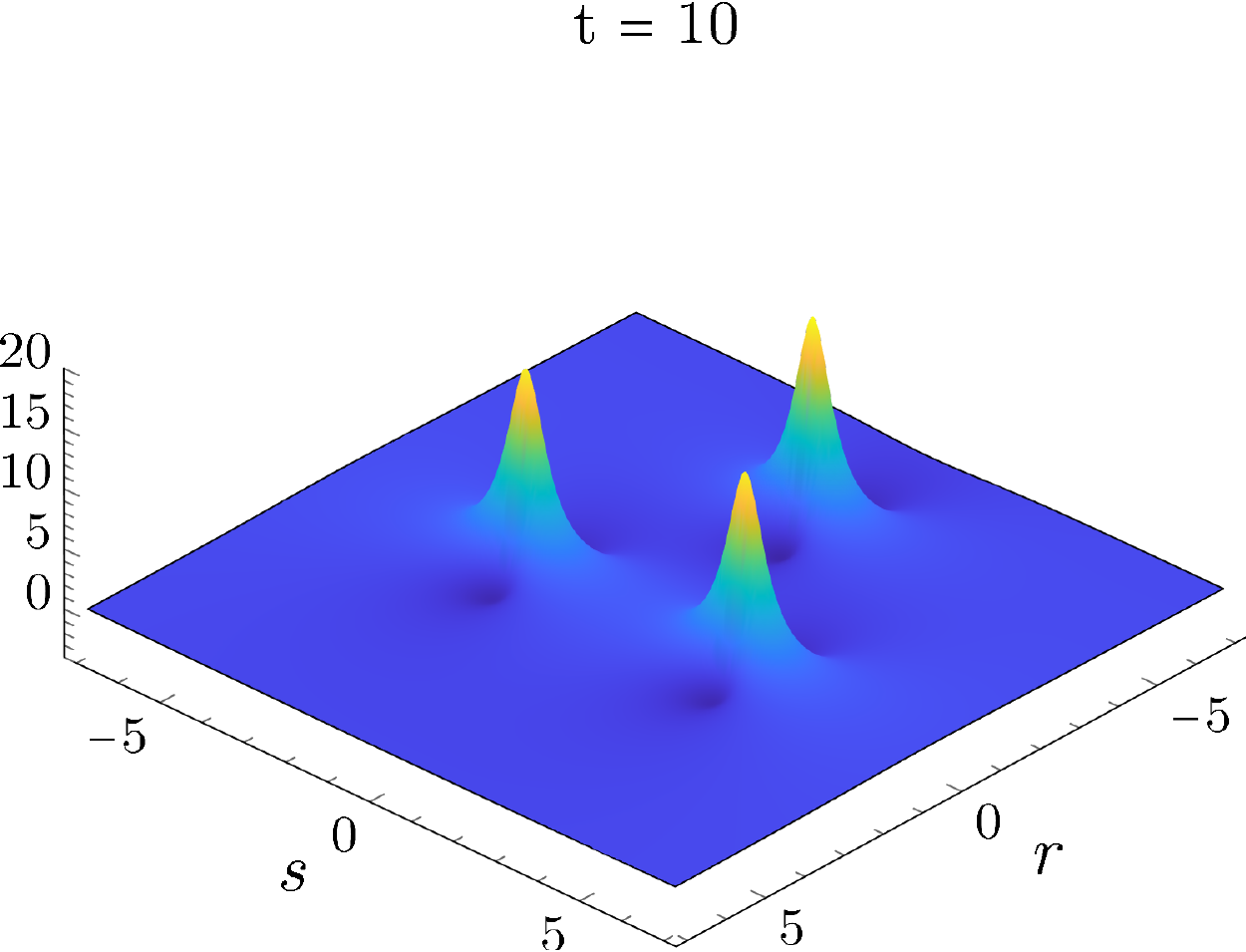}
\end{center}
\vspace{-0.2in}
\caption{$3$-lump solution of the KPI equation:\,
$a=\gamma_1=\gamma_2=\gamma_3=0, b=1$.}
\label{3lump}
\end{figure}
Figure~\ref{3lump} shows that the corresponding KPI solution $u_3(x,y,t)$  
is a $3$-lump solution which forms a triangular pattern in the $rs$-plane
for $|t| \gg 1$ with the 3 peaks located at the vertices of the triangle. 
As time progresses starting from large negative values, the lumps are 
attracted to each other and overlap, then
for $t \gg 1$ the 3 peaks re-appear and recede from each other
forming a time-reversed triangular structure. One of the peaks is located
along the $r$-axis while the other two are located symmetrically from
the $r$-axis. It will be shown in Section 4 that the approximate peak locations
for $|t| \gg 1$ can be estimated from the zeros of the first $|\cdot|^2$ term
i.e., by setting $Q({\bf 0})=0$ in \eqref{tau3}. Indeed, after substituting
\eqref{thetars} in \eqref{tau3}, $Q({\bf 0})=0$ implies that
$(z+\sf{1}{b})^3+3t-\sf{1}{4b^3} = 0$ where $z=r+is$. The peak locations
are then given by
\[z_j = r_j+is_j \sim -\frac{1}{b} + x_jt^{1/3} + 
\frac{1}{12x_j^2b^3}t^{-2/3} + O(|t|^{-5/3})\,, \quad \qquad |t| \gg 0 \,,\]
where $x_j, \, j=1,2,3$ are the roots of $x^3+3=0$. These
are shown in Figure~\ref{3lumppeaks}. Note from Figure~\ref{3lumppeaks}
that the peak locations exhibit reflection symmetry across the line $r=-\sf{1}{b}$ 
as $t \to -t$.
\begin{figure}[h!]
\begin{center}
\includegraphics[scale=0.42]{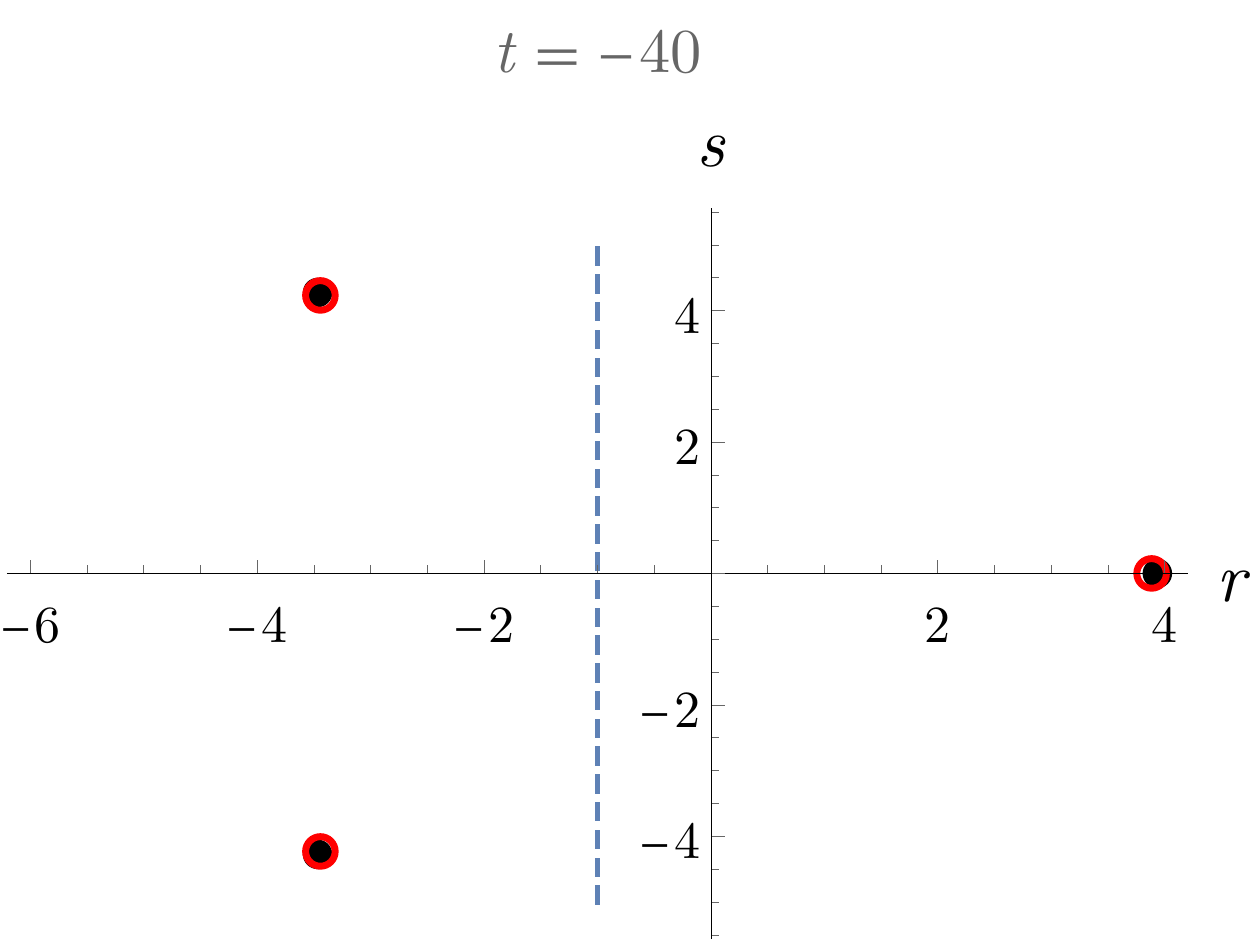} \qquad \quad
\includegraphics[scale=0.42]{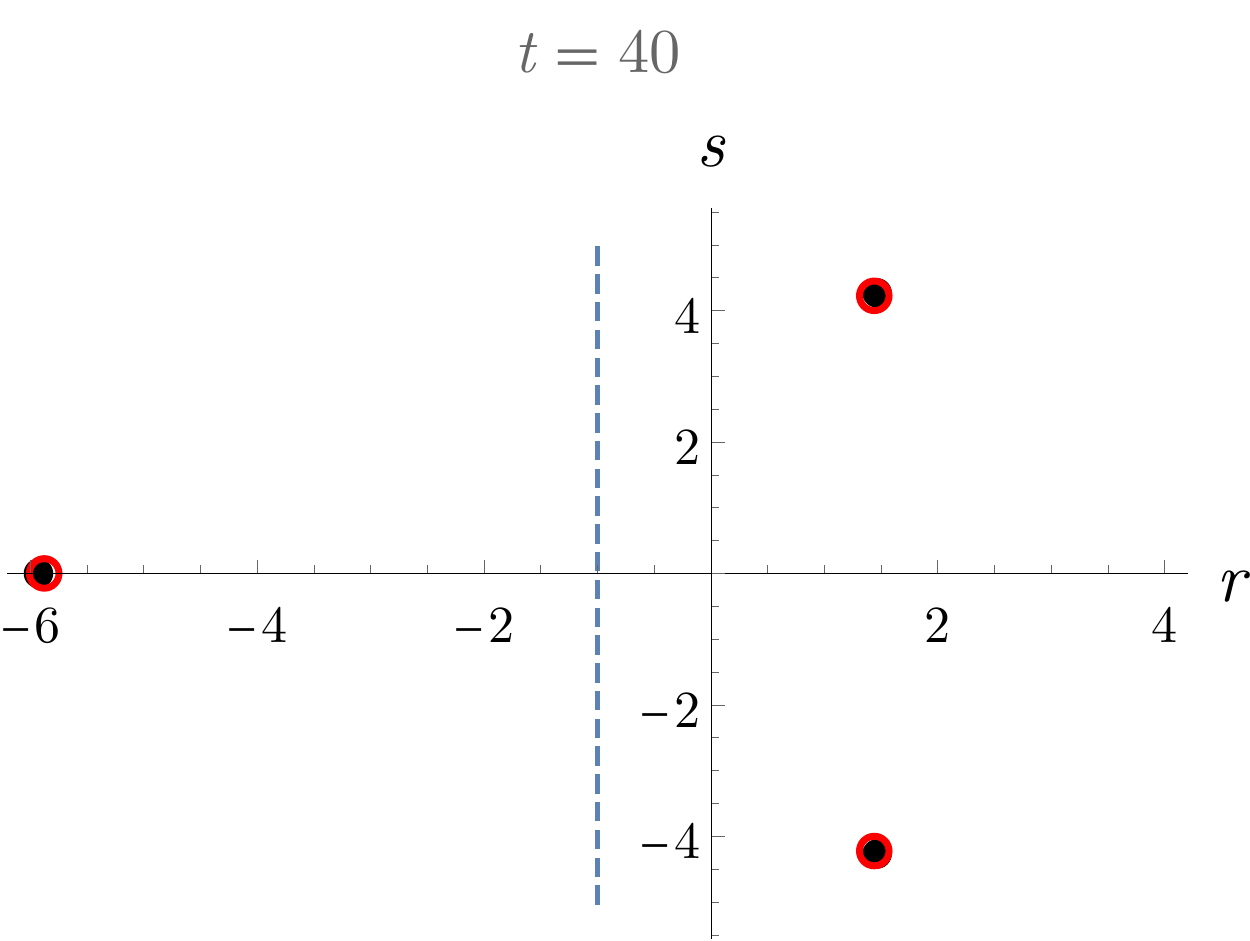}
\end{center}
\vspace{-0.2in}
\caption{$3$-lump peak locations at $t=-40$ (left panel) and $t=40$ (right panel).
Black dots represent locations $z_j$ where $Q({\bf 0})=0$, red open circles are exact locations.
As $t \to -t$, the peak locations are reflected across the vertical 
line $r=-\sf{1}{b}$. KP parameters: same as in Figure~\ref{3lump}.}
\label{3lumppeaks}
\end{figure}
Finally, we remark that in local co-ordinates 
$(r,s)=(r_j+h, s_j+k), \,\, (h,k) \sim O(1)$ near each peak, 
$\tau_3$ in \eqref{tau3} reduces to a $1$-lump 
$\tau$-function $\tau_1(h,k)$ for $|t| \gg 1$. Specifically,
$\tau_3 \sim (3|t|)^{4/3}\left(h^2+k^2+\sf{1}{4b^2}+O(|t|^{-1/3})\right)$.
This implies that the $3$-lump solution $u_3$ is a superposition
of three $1$-lump solution as $|t| \to \infty$. This result also holds
for general multi-lumps as will be shown in Section 4.2.
\section{Classification of KPI multi-lumps}
In this section we provide a simple yet useful characterization of
the KPI rational solutions constructed in Section 2.2, in terms of
integer partitions. Then we classify these solutions utilizing
certain ideas from the integer partition theory. Given a set of
distinct positive integers $\{m_1,\ldots, m_n\}$ indexing the
generalized Schur polynomials $p_{m_i}, \, i=1,\ldots,n$, the key observation
is that there is a one-to-one correspondence between the $\tau$-function
$\tau_P$ and the leading order principal minor $P({\bf 0})$.
Indeed, from $P({\bf 0})=i^{|{\bf 0}|}W({\bf 0})$ where 
the determinant $W({\bf 0})$ is given by \eqref{schur}, one can
immediately identify the generalized Schur polynomials $p_{m_i}$ which
are used to construct the KPI solution as in \eqref{H}. On the other hand,
$P({\bf 0})$ can be extracted from the first square term of   
the expression of $\tau_P$ in \eqref{square-b}, or even directly from
\eqref{P} as the leading principal minor of the matrix $P$ which is 
used to construct the Gram matrix $H$ in \eqref{gram} that then yields
the $\tau$-function $\tau_P=\det H$ in \eqref{square}. 
\subsection{Characterization of multi-lumps via integer partition}
In order to establish the relationship between the KPI lumps
and partitions of integers we first introduce some basic ideas from 
partition theory and its connection with the symmetric polynomials.
\begin{definition}[Partition]
A {\it partition} is a decomposition of a nonnegative integer $N$ 
given by a non-decreasing sequence 
$\lambda = (\lambda_1,\ldots \lambda_n)$  such
that $0 \leq \lambda_1 \leq \lambda_2 \leq \cdots \leq \lambda_n$
and $|\lambda| := \lambda_1+\cdots + \lambda_n=N$. The number of
non-zero parts $l(\lambda) \leq n$ of $\lambda$ is called its {\it length} 
and $|\lambda|=N$, its {\it size}. The strictly increasing sequence $m=(m_1, \ldots, m_n)$
defined by $m_j=\lambda_j+(j-1), \,\, j=1,\ldots,n$ is called the 
{\it degree vector} of the partition $\lambda$.
\end{definition}
\noindent It is often convenient to describe a partition $\lambda$ of $N \in \mathbb{N}$
in terms of its associated Young diagram $Y_{\lambda}$ which is a rectangular 
array of left-justified boxes (or dots) such that the i$^{\mathrm th}$ row 
from the top contains $\lambda_{n-i+1}$ boxes, $i=1,\ldots,n$. Thus $Y_{\lambda}$
consists of $n$ rows, $l(\lambda)$ of which contain non-zero number of
boxes, and $\lambda_n$ columns. The total number of boxes in $Y_\lambda$
is, of course $N$. The {\it conjugate} $\lambda'$ of a partition
$\lambda$ is a partition whose Young diagram $Y_{\lambda'}$ is the transpose
of $Y_{\lambda}$ obtained by interchanging its rows and columns.
Clearly, $|\lambda'|=|\lambda|=N, (\lambda')' = \lambda$, the length
$l(\lambda')=\lambda_n$, also $\lambda'$ has $l(\lambda) \leq n$
columns. A partition is called {\it self-conjugate} if $\lambda=\lambda'$.
\begin{example}
Following are the Young diagrams of the partition $\lambda=(1,2,4)$, its conjugate 
$\lambda' = (1,1,2,3)$ and a self-conjugate partition $\mu=(1,2,3)$
also known as a staircase (or triangular) partition.
\[ Y_{\lambda}=\raisebox{0.5em}{\ytableausetup{boxsize=0.8em}\ydiagram{4,2,1}}
\hspace{1in}
Y_{\lambda'}=\raisebox{0.85em}{\ytableausetup{boxsize=0.8em}\ydiagram{3,2,1,1}}
\hspace{1in}
Y_{\mu}=\raisebox{0.5em}{\ytableausetup{boxsize=0.8em}\ydiagram{3,2,1}}
 \]
Here, $|\lambda|=|\lambda'|=7$ which is the number of boxes in the diagrams 
$Y_\lambda$ and $Y_{\lambda'}$. The number of non-zero rows in 
$Y_\lambda = l(\lambda)=3=$ number of columns in $Y_{\lambda'}$. 
Similarly, the number of non-zero rows in
$Y_{\lambda'} = l(\lambda')=4=\lambda_4=$ number of columns in $Y_{\lambda}$.
For the self-conjugate partition, $|\mu|=6$, and $l(\mu) = 3 = \mu_3$. 
\end{example}

For a given pair of partitions $\mu$ and $\lambda$, $\mu \subset \lambda$ means
that $Y_{\mu} \subset Y_{\lambda}$, i.e., $\mu^i \leq \lambda^i, 
\,\, 1 \leq i \leq l(\mu)$ where $\mu^i, \lambda^i$ are the number of boxes in 
the i$^{\mathrm th}$ row (from top) of $Y_{\mu}, Y_{\lambda}$, respectively. 
Alternatively, suppose $\lambda$ has $n$ parts then the number
of parts in $\mu$ can be extended to $n$ by appending $0$ parts
if necessary. Then $\mu \subset \lambda$ implies $\mu_i \leq \lambda_i, \,\,
1 \leq i \leq n $. The set difference $Y_\lambda - Y_\mu$ is called a 
skew diagram $Y_{\lambda/\mu}$ which may be disconnected.
\begin{example}
Let $\lambda = (1,3,4)$ and $\mu = (1,2)$ then
\[ Y_{\lambda}=\raisebox{0.5em}{\ytableausetup{boxsize=0.8em}\ydiagram{4,3,1}}
\hspace{1in}
Y_{\mu}=\raisebox{0.3em}{\ytableausetup{boxsize=0.8em}\ydiagram{2,1}}
\hspace{1in}
Y_{\lambda/\mu}=\raisebox{0.5em}{\ytableausetup{boxsize=0.8em}\ydiagram{2+2,1+2,1}}
\]
If $\rho = (2,3)$ but $\lambda$ is the same as above then
\[Y_{\rho}=\raisebox{0.3em}{\ytableausetup{boxsize=0.8em}\ydiagram{3,2}}
\hspace{1in}
Y_{\lambda/\rho}=\raisebox{0.5em}{\ytableausetup{boxsize=0.8em}\ydiagram{3+1,2+1,1}}
\]
\end{example}
\noindent {\bf Schur and skew-Schur functions}:\,
Associated with each partition $\lambda$ of $N \in \mathbb{N}$
and degree vector $m = (m_1,\ldots,m_n)$, there is a unique weighted homogeneous
polynomial of degree $N$ in the variables $\theta_j$ defined by the
second $n \times n$ determinant in \eqref{schur} namely,
$W({\bf 0}) = \det(p_{m_j-i+1})=\det(p_{\lambda_j+j-i})$.
It will be referred to as the {\it Schur function} 
throughout this article and will be denoted by $W_{\lambda}$.
Historically, the Schur functions were introduced in the study of symmetric
functions in $n$ variables $(x_1,\ldots,x_n)$ with integer coefficients,
and was defined as a quotient
of determinants namely $W_{\lambda}(x_1,\ldots,x_n) =
\det(x_i^{m_j})/\det(x_i^{j-1})$. The last formula can be re-expressed via
Jacobi-Trudi identity as $W_{\lambda} = \det(h_{\lambda_j+j-i})$ where
$h_r = \sum x_{i_1} \cdots x_{i_r}$, $1 \leq i_1 \leq \cdots \leq i_r \leq n$
are the complete symmetric polynomials of degree $r$. Next, introducing the
power sums $\theta_j = \sf{1}{j}(x_1^j+\cdots+x_n^j), \, j\geq 1$ it can be shown that
$h_r(x_1,\ldots,x_n)=p_r(\theta_1,\ldots,\theta_r)$, and one recovers
the expression for $W_{\lambda}$ in terms of the generalized Schur polynomials
$p_r$. 

Just as each Young diagram $Y_\lambda$ is associated to the 
Schur function $W_{\lambda}$, similarly each skew diagram $Y_{\lambda/\mu}$
is associated with a {\it skew Schur function} defined as follows.
Let $\lambda, \mu$ be two partitions with degree vectors $m$ and $n$, i.e.,
$m_i = \lambda_i+(i-1), \, n_i = \mu_i+(i-1), \,\, i=1,\ldots,n$, and
let $\mu \subset \lambda$. Then the Schur function $W_\lambda$ and the
skew Schur function $W_{\lambda/\mu}$ are defined as   
\begin{equation}
W_\lambda = \det(p_{m_j-i+1}) = \Wr(p_{m_1},\ldots,p_{m_n})\,,
\qquad
W_{\lambda/\mu} = \det(p_{m_j-n_i}) = \det(p_{\lambda_j-\mu_i+j-i})\,,  
\label{ss}
\end{equation}  
where the wronskian in the first expression is with respect to $\theta_1$,
and is the same epression as $W({\bf 0})$ in \eqref{schur}. 
Both $W_\lambda$ and $W_{\lambda/\mu}$ are weighted homogeneous polynomials
in the $\theta_j$s of degree $|\lambda|$ and $|\lambda|-|\mu|$ respectively, 
which are also equal to the number of boxes in the Young diagram $Y_\lambda$
and the skew diagram $Y_{\lambda/\mu}$. It follows
that if $\mu = \lambda$ i.e., $Y_{\lambda}-Y_{\mu} = \emptyset$ (a diagram with
no boxes) then $W_{\emptyset}=\Wr(p_0,p_1,\ldots,p_{n-1})=1$.
Also, if $\mu = \emptyset$ then $W_{\lambda/\emptyset} = W_{\lambda}$, and if
$\mu \supset \lambda$ then $W_{\lambda/\mu} =0$.
The Schur functions $W_{\lambda}$ forms a basis of the symmetric functions
over the ring of integers. The Schur function together with its properties and
its role in the representation theory of the symmetric group $S_N$
are well documented in the literature~\cite{M37,FH91,M95} (and references therein).

Notice that the multi-index sets ${\bf r}=(r_1 \ldots r_n)$ with 
$0 \leq r_1 < r_2 < \cdots < r_n \leq m_n$ introduced to label the
$n \times n$ minors in \eqref{square} form the degree vector of a partition
denoted by $\lambda({\bf r})$. Each part $\lambda_i({\bf r}) = r_i-i+1$
for $1 \leq i \leq n$ and $|\lambda({\bf r})| = |{\bf r}| - |{\bf 0}|$ where
recall that $|{\bf r}|=r_1+\cdots+r_n$.
Furthermore, in addition to the lexicographic ordering in Definition 2.1,
there is a natural partial ordering among the multi-index sets defined by
$$ {\bf r} \preceq {\bf s} \quad \Leftrightarrow \quad r_i \leq s_i\,,
\qquad 1 \leq i \leq n\,.$$
The corresponding partitions then satisfy $\lambda({\bf r}) \subset \lambda({\bf s})$
and are partially ordered by the inclusion 
$Y_{\lambda({\bf r})} \subset Y_{\lambda({\bf s})}$ of the corresponding Young 
diagrams as mentioned after Example 3.1. Note that if ${\bf r} = {\bf 0}$ then the partition
$\lambda({\bf 0})=\emptyset$ which has only $0$ parts, and when ${\bf r} = {\bf m}$
the last multi-index in the lexicographic order, $\lambda({\bf m})=\lambda$.
\begin{example}
Consider the multi-index sets of Example 2.1. The ordering described above
corresponds to the partially ordered set (poset) shown below where the arrows 
indicate the ordering ${\bf r} \preceq {\bf s}$. 
\[ \begin{matrix}
 &  & \quad\nearrow & \raisebox{0.1 in}{(03)} & \hspace{-.1 in} \searrow &  \\
(01) & \rightarrow & (02) &  & (13) & \rightarrow & (23) \,, \\
 &  & \quad \searrow & \raisebox{-0.1 in}{(12)} & \hspace{-.1 in} \nearrow &  \\
\end{matrix} \]  
Note that the index sets $(03)$ and $(12)$ are incomparable 
although both $(03) \preceq (13), \,\, (12) \preceq (13)$.
The partitions $\lambda({\bf r})$ corresponding to ${\bf r} = (r_1 r_2)$ are
given by \\
$\lambda((01)) = (0,0) = \emptyset, \quad \lambda((02)) = (0,1), \quad 
\lambda((03))=(0,2), \quad \lambda((12)) = (1,1), \quad \lambda((13))=(1,2)$, and  
$\lambda((23))=(2,2)$. Then the associated Young diagrams are 
\[ Y_{\emptyset} \,, \qquad  
Y_{(0,1)}=\ytableausetup{boxsize=0.6em}\ydiagram{1}\,, \qquad 
Y_{(0,2)}=\ytableausetup{boxsize=0.6em}\ydiagram{2}\,, \qquad
Y_{(1,1)}=\raisebox{0.4em}{\ytableausetup{boxsize=0.6em}\ydiagram{1,1}}\,, \qquad
Y_{(1,2)}=\raisebox{0.4em}{\ytableausetup{boxsize=0.6em}\ydiagram{2,1}}\,, \qquad
Y_{(2,2)}=\raisebox{0.4em}{\ytableausetup{boxsize=0.6em}\ydiagram{2,2}} \,. 
\]
Notice that neither $Y_{(0,2)}$ nor $Y_{(1,1)}$ are contained in each other but
both are subdiagrams of $Y_{(1,2)}$. 
\end{example}
The main purpose of the background in partition theory is to relate the
expressions for the KPI $\tau$-function $\tau_P$ in \eqref{square} with
the Schur and skew Schur functions. It was already mentioned earlier
that in \eqref{ss} $W_\lambda=W({\bf 0})$ which is the second determinant 
in \eqref{schur}. The second formula
in \eqref{ss} states that for a given multi-index set ${\bf r}$ 
and the associated partition $\mu=\lambda({\bf r})$ the skew Schur function 
$W_{\lambda/\mu}=W_{\lambda/\lambda({\bf r})}=W({\bf r})$ where $W({\bf r})$ is the
first determinant in \eqref{square}. Then using the fact that the 
$n \times n$ minors of the upper triangular
matrix $U$ in \eqref{square-a} satisfy $U\tbinom{{\bf r}}{{\bf s}} = 0$ unless
${\bf r} \preceq {\bf s}$, the minors in \eqref{square-b} can now be re-expressed 
in terms of the Schur and the skew Schur functions in the following manner
\begin{subequations}
\begin{equation}
Q({\bf 0}) = i^{|{\bf 0}|}\Big(W_{\lambda} + \sum_{\lambda(s)\neq\emptyset} \!\!
i^{|\lambda(s)|}U\tbinom{{\bf 0}}{{\bf s}}W_{\lambda/\lambda({\bf s})} \Big) \,,
\label{skewschur-a}
\end{equation}
where the sum is over all non-empty partitions $\lambda({\bf s})$.
If $\lambda({\bf s}) \subset \lambda$, i.e.,
partitions whose Young diagram is contained in $Y_{\lambda}$, then
the corresponding skew Schur function $W_{\lambda/\lambda({\bf s})}$ is
a weighted homogeneous polynomial in the $\theta_j$ of degree 
$|\lambda|-|\lambda({\bf s})|, \,\, 1 \leq |\lambda({\bf s})| \leq |\lambda|$.
Otherwise, $W_{\lambda/\lambda({\bf s})}=0$ when $\lambda({\bf s}) \supset \lambda$.
Similarly, for each ${\bf r} > {\bf 0}$, the second expression in \eqref{square-b}
gives    
\begin{equation}
Q({\bf r})  = i^{|{\bf r}|}\Big(W_{\lambda/\lambda({\bf r})} + 
\sum_{\lambda({\bf s})} 
i^{|\lambda({\bf s})|-|\lambda({\bf r})|}U\tbinom{{\bf r}}{{\bf s}}
W_{\lambda/\lambda({\bf s})} \Big) \,, \qquad \lambda({\bf r}) \subset \lambda({\bf s}) 
\subseteq \lambda \,,  
\label{skewschur-b}
\end{equation}  
where the sum is over all $\lambda({\bf s})$ such that the Young diagrams
satisfy $Y_{\lambda({\bf r})} \subset Y_{\lambda({\bf s})} \subseteq Y_\lambda$. 
The skew Schur functions $W_{\lambda/\lambda({\bf s})}$ in this sum are  
weighted homogeneous polynomials in the $\theta_j$ of degree 
$|\lambda|-|\lambda({\bf s})|$ with
$|\lambda({\bf r})| < |\lambda({\bf s})| \leq |\lambda|$.
\label{skewschur} 
\end{subequations}
\begin{example}
Consider the $3$-lump solution in Section 2.3.3. In this case,
the partition corresponding to the degree vector $m=(m_1,m_2)=(1,3)$
is $\lambda = (1,2)$. The $2 \times 2$ minors of $P, Q, U$ and $D$
are labeled by the set of multi-indices $\{{\bf r}=(r_1 r_2)\}$ 
which corresponds to the poset in Example 3.3. Using the partitions 
$\lambda({\bf r})$ enumerated in Example 3.3 and employing either
equation \eqref{schur} or \eqref{ss},
the Schur function $W_{(1,2)}=W({\bf 0})$ and the skew schur functions 
$W_{\lambda/\lambda({\bf r})}=W({\bf r})$ for this example are found to be
\begin{gather*} 
W((01))=W_{(1,2)} = \begin{vmatrix} p_1 & p_3 \\ p_0 & p_2 \end{vmatrix}\,, \quad  
W((02))=W_{(1,2)/(0,1)} = \begin{vmatrix} p_1 & p_3 \\ 0 & p_1 \end{vmatrix}\,, \quad
W((03))=W_{(1,2)/(0,2)} = \begin{vmatrix} p_1 & p_3 \\ 0 & p_0 \end{vmatrix}\,, \\
W((12))=W_{(1,2)/(1,1)} = \begin{vmatrix} p_0 & p_2 \\ 0 & p_1 \end{vmatrix}\,, \quad
W((13))=W_{(1,2)/(1,2)} = \begin{vmatrix} p_0 & p_2 \\ 0 & p_0 \end{vmatrix} \quad
W((23))=W_{(1,2)/(2,2)} = \begin{vmatrix} 0 & p_1 \\ 0 & p_0 \end{vmatrix}\,,  
\end{gather*}
Notice that $W_{(1,2)/(1,2)}=W_{\emptyset}=1$ and $W_{(1,2)/(2,2)}=0$ since
the partition $(2,2) \supset (1,2)$. Then from \eqref{skewschur-a} and the
matrix $U$ given in Section 2.3.3, 
$$Q({\bf 0})= Q((01))= i\big[W_{(1,2)} +\sf{i}{b}W_{(1,2)/(0,1)} 
-\sf{3}{4b^2}W_{(1,2)/(0,2)}-\sf{1}{4b^2}W_{(1,2)/(1,1)}-\sf{i}{4b^3}W_{\emptyset}\big]\,,$$  
which is the argument inside the leading $|\cdot|^2$ term of $\tau_3$ in 
\eqref{tau3} when expressed in terms of $\theta_j$. The arguments inside the 
remaining $|\cdot|^2$ terms in
$\tau_3$ are obtained from \eqref{skewschur-b}. For example,
$$Q((02)) = -\big[W_{(1,2)/(0,1)}+\sf{3i}{2b}W_{(1,2)/(0,2)}+
\sf{i}{2b}W_{(1,2)/(1,1)}-\sf{3}{4b^2}W_{\emptyset}\big]\,,$$ 
and the rest follow in a similar fashion.
\end{example} 
We rename the KPI $\tau$-function $\tau_P$ as $\tau_\lambda$. Then the characterization 
of $\tau_\lambda$ in terms of integer partition can be stated as follows.
\begin{proposition}
Let $\lambda = (\lambda_1, \ldots, \lambda_n)$ be a partition of a 
positive integer $N$ and $W_{\lambda}$ be the associated Schur function. 
Then there is a unique (upto a $GL(n,\mathbb{C})$ gauge freedom)
KPI $\tau$-function given by \eqref{square-b} and \eqref{skewschur}  
\[ \tau_{\lambda} = \frac{1}{(2b)^{2|{\bf 0}|}}\Big|W_{\lambda} + 
\sum_{\lambda(s) \neq \emptyset} \!\!
i^{|\lambda(s)|}U\tbinom{{\bf 0}}{{\bf s}}W_{\lambda/\lambda({\bf s})} \Big|^2 +
\sum_{{\bf r} > {\bf 0}}\frac{1}{(2b)^{2|{\bf r}|}}\Big|W_{\lambda/\lambda({\bf r})} + 
\sum_{\lambda({\bf r}) \subset \lambda({\bf s}) \subseteq \lambda} \!\!\!
i^{|\lambda({\bf s})|-|\lambda({\bf r})|}U\tbinom{{\bf r}}{{\bf s}}
W_{\lambda/\lambda({\bf s})} \Big|^2 \,,   
\]
where ${\bf r} = (r_1 \ldots r_n), \,\,
0 \leq r_1 < \cdots < r_n \leq m_n$ is the degree vector of the partition
$\lambda({\bf r})$. $\tau_\lambda$ is expressed as a sum of squares, where each 
square term is a weighted homogeneous
polynomial in $\theta_j, \, j=1,\ldots,\lambda_n+n-1$ and their complex conjugates,
of degree $2(|\lambda|-|\lambda({\bf r})|), \,\, 0 \leq |\lambda({\bf r})| \leq |\lambda|$ 
with $|\lambda|=N$.
\end{proposition}
\begin{remark}
\begin{itemize}
\item[(a)] Throughout this article the parts $\lambda_j$ of a partition
$\lambda=(\lambda_1,\ldots, \lambda_n)$ are consistently listed in 
{\it non-decreasing} order (cf. Definition 3.1) contrary to the more 
standard convention of listing the parts in non-increasing order.
\item[(b)] Yet another way to denote a partition of a positive integer $N$ 
is to indicate the number of times (multiplicity)
each positive integer $1,2,\ldots$ occur as a part in the partition.
That is, $\lambda = (1^{\alpha_1},2^{\alpha_2},\ldots, N^{\alpha_N})$, where
$\alpha_i \geq 0, \,\, i=1,2,\ldots, N$ satisfy 
$\alpha_1+2\alpha_2+\cdots N\alpha_N=N=|\lambda|$. Let $Y_{\lambda}$ be
the Young diagram of $\lambda$. The sum of boxes in each column
of $Y_{\lambda}$ starting from the leftmost column is successively given by
$\mu_N = \alpha_1+\cdots+\alpha_N, \,\, \mu_{N-1} = \alpha_2+\cdots+\alpha_N,
\ldots \,\, \mu_1 = \alpha_N$ so that $\mu_1 \leq \mu_2 \leq \cdots \mu_N$.
Thus, $\mu = (\mu_1, \ldots, \mu_N)$ is also a partition with $|\mu|=N$.  
In fact, $\mu = \lambda'$ is the partition conjugate to $\lambda$ where
the difference between successive rows from the top 
$\mu_j - \mu_{j-1} = \alpha_{N-j+1}, \,\, j=N, N-1, \ldots, 2, \, \mu_1 = \alpha_N$.
\item[(c)] The Schur function $W_\lambda$ is a wronskian as shown
in \eqref{ss} but the skew Schur functions $W_{\lambda/\mu}$ are {\it not} wronskians. 
They are in fact, components of the $n$-form $\omega_P$ representing the 
Grassmannian $P \in {\rm Gr}_{\mathbb{C}}(n,m_n+1)$ introduced in Section 2.2.
However, viewed as functions of $\theta=(\theta_1, \theta_2, \ldots)$, it can
be shown that $W_{\lambda/\mu}(\theta) = W_\mu(\tilde{\partial})W_\lambda(\theta)$,
where $\tilde{\partial} = (\partial_{\theta_1}, \sf{1}{2}\partial_{\theta_2}, 
\sf{1}{3}\partial_{\theta_3}, \ldots)$~\cite{K18}.
\end{itemize} 
\end{remark} 
\subsection{The $N$-lump solutions}
Propositions 2.1 and 3.1 demonstrate that the 
$\tau$-function $\tau_\lambda$ and the corresponding rational solution of KPI
given by \eqref{u} can be identified with a partition $\lambda$
of a positive integer $N$. In particular, each $\tau_\lambda$ is uniquely
characterized by its Schur function $W_{\lambda}$ 
which is a weighted homogeneous polynomial in $\theta_j$ of degree
$|\lambda|=N$. Henceforth in this paper, we refer to these solutions
as the {\it $N$-lump} solutions. Our nomenclature will be further justified
in Section 4 where it will be shown that the $N$-lump solution
separates into $N$ distinct peaks whose heights approach the
$1$-lump peak height asymptotically as $|t| \to \infty$. Evidence of
the latter feature has already been seen in the examples of Section 2.3.
The $N$-lump solutions for a given positive integer $N$ can be enumerated
according to the underlying partition $\lambda$ of $N$. To avoid redundancies,
we will consider only those partitions $\lambda=(\lambda_1, \ldots, \lambda_n)$ 
in this subsection, whose smallest part $\lambda_1 > 0$ so that the
length of the partition $l(\lambda)=n$. It will be useful to
introduce a lexicographic ordering among the partitions of the same 
size $|\lambda|=N$ in a similar way as in Definition 2.1.
\begin{definition} (Lexicographic ordering)
Let $\lambda$ and $\mu$ be two distinct partitions of $N$. If $j$ is the 
earliest index such that $\mu_j \neq \lambda_j$ then $\mu < \lambda$
if and only if $\mu_j < \lambda_j$.
\end{definition} 
\noindent Note that the lexicographic ordering is a total ordering. The
{\it smallest} partition $\lambda = (1,1,\ldots,1):= (1^N)$ which
has $N$ parts, and the {\it largest} partition $\lambda = (N)$ 
that has only one part. 
\begin{example}
All possible partitions for $N=4$ are:\, $(1^4) < (1^2,2)< (1,3)< (2^2)<(4)$
and the associated Young diagrams are given by
\[ Y_{(1^4)}=\raisebox{1.0em}{\ytableausetup{boxsize=0.6em}\ydiagram{1,1,1,1}}\,, \qquad 
Y_{(1^2,2)}=\raisebox{0.6em}{\ytableausetup{boxsize=0.6em}\ydiagram{2,1,1}}\,, \qquad
Y_{(1,3)}=\raisebox{0.4em}{\ytableausetup{boxsize=0.6em}\ydiagram{3,1}}\,, \qquad
Y_{(2^2)}=\raisebox{0.4em}{\ytableausetup{boxsize=0.6em}\ydiagram{2,2}}\,, \qquad
Y_{(4)}=\ytableausetup{boxsize=0.6em}\ydiagram{4} \,. 
\]
There are two conjugate pairs, namely 
$(1^4)'=(4), \, (1^2,2)'=(1,3)$, while $(2^2)$ is self-conjugate.
Hence there are 5 types of $4$-lump solutions indexed by the Schur
functions:\, $W_{(1^4)} = \Wr(p_1,p_2,p_3,p_4), \,\,W_{(1^2,2)}=\Wr(p_1,p_2,p_4), 
\,\, W_{(1,3)}=\Wr(p_1,p_4), \,\, W_{(2^2)}=\Wr(p_2,p_3), \,\, W_{(4)} = p_4$.
In each case, $W_{\lambda}$ is a weighted homogeneous polynomial
in $\theta_j, \,\, j=1,\ldots,4$ of degree 4. 
\end{example}
For a given positive integer $N$, the total number of distinct 
$N$-lump solutions of the KPI equation is given by $p(N)$, which is 
the total number of ways to partition $N$. The Euler generating 
function for $p(N)$ is given by
\begin{equation}
p(x) = \sum_{N} p(N)x^N = \prod_{j\geq 1}\frac{1}{1-x^j} 
= \frac{1}{(1-x)(1-x^2)(1-x^3) \cdots} \,,
\label{pN}
\end{equation}
$p(0)=0$ and $p(j)=0$ if $j<0$. In addition, 
there exists a recurrence formula 
\[p(N) = \sum_{k\geq 1}(-1)^{k-1}\Big(p\big(N-\sf{3k^2-k}{2}\big)
+p\big(N-\sf{3k^2+k}{2}\big)\Big) 
=p(N-1)+p(N-2)-p(N-5)-p(N-7)+\cdots \,, 
\] 
that follows from Euler's pentagonal number identity
\[\frac{1}{p(x)}=\prod_{j\geq 1}(1-x^j) = 
1+\sum_{k \geq 1}(-1)^k(x^{(3k^2-k)/2} + x^{(3k^2+k)/2})\,. \]
The first few values of $p(N)$ are $p(1)=1, p(2)=2, p(3)=3, p(4)=5,
p(5)=7, \ldots$. It is also possible to further refine the class
of $N$-lump solutions according to a fixed number of parts $n$ of 
the partition. The number of partitions of $N$ with $n (\leq N)$ 
parts $p(N,n)$ is obtained from the generating function
\begin{equation} 
p(x,y)=\sum_{N}p(N,y)x^N = \prod_{j \geq 1}\frac{1}{(1-yx^j)}\,, \qquad
p(N,y)= \sum_{n=1}^Np(N,n)y^n \,.
\label{pNn} 
\end{equation}
Setting $y=1$ above, yields $\sum_{n}p(N,n) = p(N)$. It is easy to
verify from \eqref{pNn} that $p(N,1)=p(N,N) =1$. Moreover, the $p(N,n)$
satisfy the recurrence relation $p(N,n) = p(N-1,n-1)+p(N-n, n)$.   
\begin{example}
For $N=4, \, p(4)=5$, which is the coefficient of $x^4$ in \eqref{pN}.
These 5 partitions corresponding to 5 distinct $4$-lump solutions 
of KPI equation are enumerated in Example 3.5. Furthermore, from \eqref{pNn}
the polynomial $p(4,y) = y^4+y^3+2y^2+y$ which implies that there are
2 partitions each with 2 parts, and 3 partitions each with 1, 3 and 4 parts 
as illustrated in Example 3.5.
\end{example}  
\noindent Since conjugation preserves the size of a partition, i.e.,
$|\lambda|=|\lambda'|=N$, it is natural to divide the set of all partitions
of $N \in \mathbb{N}$ into two distinct classes, namely  
\begin{center}
(I)\, non-self-conjugate partitions ($\lambda \neq \lambda')$ and
(II)\, self-conjugate partitions ($\lambda=\lambda'$). 
\end{center}
\noindent Class (I) has
even number of partitions that are further characterized by the fact
that the largest part $\lambda_n \neq n$, the number of parts. More
precisely class (I) consists of partitions $\{\lambda: \lambda_n \geq n\}$
and their conjugates $\lambda'$. Class (II) consisting of the 
self-conjugate partitions are contained in the set $\{\lambda: \lambda_n=n\}$
although not necessarily {\it equal} to that set. (For example,
$\lambda = (2^2,3)$ satisfy $\lambda_3=3$ but it is not self-conjugate
since $\lambda'=(1,3^2)$).
However, the self-conjugate partitions are in bijection with the partitions
with distinct, odd parts as illustrated below
\[\raisebox{0.6em}{\ytableausetup{boxsize=0.6em}\ydiagram{4,4,3,2}}\quad = \quad 
\raisebox{0.6em}{\ytableausetup{boxsize=0.6em}\ydiagram{4,1,1,1}}\quad + \quad
\raisebox{0.6em}{\ytableausetup{boxsize=0.6em}\ydiagram{3,1,1}}\quad + \quad
\ytableausetup{boxsize=0.6em}\ydiagram{1} \quad \longrightarrow
\quad \raisebox{0.6em}{\ytableausetup{boxsize=0.6em}\ydiagram{7,5,1}}  
\]
where the Young diagram of a self-conjugate partition is decomposed along
its diagonal into ``hooks'' each of which has odd number of boxes.
Consequently, the number of self-conjugate partitions of $N$ is equal
to the number of partitions of $N$ with distinct, odd parts $k(N)$.
The latter has a generating function
\[K(x) = \sum_N k(N)x^N = \prod_{j \geq 1}(1+x^{2j-1}) = (1+x)(1+x^3) \cdots
\]
For example, $k(4)=1$ corresponds to $\lambda=(1,3)$ as well
as the self-conjugate partition $\lambda=(2,2)$ in Example 3.5.     
Furthermore, $p(N)-k(N)$ is an even number which equals the number of
partitions in class (I). 

In order to further characterize the $\tau$-function $\tau_\lambda$ corresponding
to classes (I) and (II), it is necessary to describe an involution
symmetry among the Schur and skew Schur functions of a partition $\lambda$ 
and its conjugate $\lambda'$. It will be useful for that purpose
to introduce the elementary symmetric
functions of of degree $r$ in variables $(x_1,\ldots,x_n)$. These are defined as
$e_r = \sum x_{i_1}x_{i_2}\cdots x_{i_r}, \,\, 
1 \leq i_1<i_2 \cdot < i_r \leq n$ and form a basis for the symmetric 
functions over the integer ring like the complete symmetric polynomials
$h_r$ introduced earlier in Section 3.1. When expressed in terms of the
power sums $j\theta_j = (x_1^j+\cdots+x_n^j)$ it can be shown 
(see e.g.,~\cite{M37,M95}) that
$e_r(x_1,\ldots,x_n) = p_r(\theta_1, -\theta_2, \ldots, \pm \theta_r)$
i.e., by reversing the sign $\theta_j \to -\theta_j$ when $j$ is {\it even} 
in the generalized Schur polynomials. For brevity, we denote this 
involution on the ring of symmetric polynomials 
by $\omega(\theta_j)=(-1)^{j-1}\theta_j, \,\, j \geq 1$ on the power
sums. Then it follows that 
$$\omega(h_r) = p_r(\omega(\theta)) = e_r\,, 
\qquad \theta = (\theta_1,\theta_2,\ldots) $$
from the relationships between the symmetric polynomials $e_r, h_r$ 
and the generalized Schur polynomials $p_r$.

Let $\lambda=(\lambda_1,\ldots,\lambda_n)$ be a partition of $N \in \mathbb{N}$
and $\lambda'=(\lambda'_1,\ldots,\lambda'_k)$ be its conjugate so that
$\lambda_n=k$ and $\lambda'_k=n$. If $\mu, \mu'$ be another pair of
conjugate partitions then a well known result (see e.g.,,~\cite{M37,M95})
from the theory of symmetric functions shows that the skew Schur 
function $W_{\lambda/\mu}$ can be expressed
via both sets of polynomials $\{e_r\}$ and $\{h_r\}$ as follows:\,
$W_{\lambda/\mu} = \det(h_{\lambda_j-\mu_i+j-i})_{i,j=1}^n =
\det(e_{\lambda'_j-\mu'_i+j-i})_{i,j=1}^k$. Setting $\mu = \emptyset$
an analogous result for the Schur function $W_{\lambda}$ is also obtained.
Translated in terms of generalized Schur polynomials, this result
states that $W_{\lambda/\mu} = 
\det\big(p_{\lambda_j-\mu_i+j-i}\big)_{i,j=1}^n(\theta)  
= \det\big(p_{\lambda'_j-\mu'_i+j-i}\big)_{i,j=1}^k(\omega(\theta))$.
An immediate consequence of this result is the following duality
under the involution $\omega$~\cite{M37,M95}
\begin{equation*}
W_{\lambda'}(\theta) = W_{\lambda}(\omega(\theta))\,, \qquad
W_{\lambda'/\mu'}(\theta) = W_{\lambda/\mu}(\omega(\theta))\,,
\qquad \omega(\theta_j) = (-1)^{j-1}\theta_j\,, \,\, j \geq 1 \,.  
\end{equation*} 
Next we apply the above involution property to obtain 
the $\tau$-function $\tau_P$ for conjugate partitions $\lambda'$
whose degree vector is given by $(m'_1, \ldots, m'_k), \,\, 
m'_j=\lambda'_j+j-1, \,\, 1 \leq j \leq k$. However, note that
since $k=\lambda_n$ and $\lambda'_k=n$, $m'_k = \lambda'_k+k-1
= n+k-1 = \lambda_n+n-1 = m_n$. Then the matrices
$P,Q,W$ in Section 2.2 corresponding to partitions $\lambda$ and
$\lambda'$, have the {\it same} number
$m_n+1=m_k+1$ number of rows and $U$ is the same upper triangular matrix
for both partitions. The (skew) Schur functions associated with $\lambda'$
are the $k \times k$ minors of the $m_n+1 \times k$ matrix
$$
W'_{rj} = p_{m'_j-r}\,, \quad 0 \leq r \leq m_n=m'_k\,, \,\,
1 \leq j \leq k \quad \text{and} \quad  p_j=0, \,\, j<0 
$$ 
analogous to $W$ in \eqref{P}. The minors $W'({\bf r'})$
are labeled by the multi-indices ${\bf r'} = (r'_1 \ldots r'_k)$ such that
$0 \leq r'_1 < \cdots < r'_k \leq m_n$ which form the degree vector of the
corresponding partition $\lambda'({\bf r'})$ whose parts are given
by $\lambda'_j({\bf r'}) = r'_j-j+1$. The total number of the such
multi-indices $\{{\bf r'}\}$ are $\tbinom{m'_k+1}{k} = \tbinom{k+n}{k} = \tbinom{k+n}{n} 
= \tbinom{m_n+1}{n}$, which is equal to the total number of multi-indices $\{{\bf r}\}$
which label the $n \times n$ minors of $W$ given by \eqref{schur}.
Therefore, for every partition $\lambda({\bf r})$ there is a unique 
multi-index ${\bf r'}$ such that $\lambda({\bf r})' = \lambda'({\bf r'})$.
In fact, there is a one-to-one correspondence among the elements in
the sets $\{{\bf r'}\}$ and $\{{\bf r}\}$ with respect to the linear ordering 
in Definition 2.1. That is, ${\bf 0} \equiv {\bf 0'}, \ldots, {\bf m} \equiv {\bf m'}$.
Consequently, we have the following identification among the Schur functions
of a given partition $\lambda$ and its conjugate
\begin{equation}
W_{\lambda'}(\theta) = W'({\bf 0'})(\theta) = W_{\lambda}(\omega(\theta))\,, \qquad
W_{\lambda'/\lambda'({\bf r'})}(\theta) = W'({\bf r'})(\theta)
= W_{\lambda/\lambda({\bf r})}(\omega(\theta))  
\label{dualschur}
\end{equation}
Inserting \eqref{dualschur} to compute $Q({\bf 0'}), Q({\bf r'})$ in
\eqref{skewschur} and substituting the resulting expressions \eqref{square-b}
enables one to compute $\tau_P$ corresponding to the conjugate partition 
$\lambda'$ in terms of $\lambda$. The above results are collected below.  
\begin{proposition}
Let $\lambda$ be a partition of a positive integer $N$ 
with degree vector $(m_1,\ldots,m_n)$ and let $\lambda'$ be 
the conjugate partition with degree vector $(m'_1,\ldots, m'_k)$
where $k=\lambda_n$ and $m'_k=m_n$. Let $\lambda'({\bf r'})$ be
the partition with degree vector ${\bf r'} = (r'_1 \ldots r'_k), \,\, 
0 \leq r'_1 < \cdots < r'_k \leq m_n$.
Then the $\tau$-function $\tau_{\lambda'}$ is obtained 
from $\tau_\lambda$ given in Proposition 3.1 by first identifying 
a unique ${\bf r'}$ by the relation $\lambda({\bf r})' = \lambda'({\bf r'})$
and then applying the involution symmetry \eqref{dualschur}
to obtain $W_{\lambda'}(\theta)$ and $W_{\lambda'/\lambda'({\bf r'})}(\theta)$.   
\end{proposition}
\begin{remark}
\begin{itemize}
\item[(a)] It is important to note that although the involution symmetry 
in \eqref{dualschur} applies to the Schur and skew Schur functions,
it does {\it not} apply to the $\tau$-functions themselves, i.e., 
$\tau_{\lambda'}(\theta) \neq \tau_\lambda(\omega(\theta))$. 
This is because the coefficients $U\tbinom{{\bf r}}{{\bf s}}$
and the phase factors $i^{|\lambda({\bf r})|}$ are not the same when the
multi-indices ${\bf r}$ are replaced by the corresponding ${\bf r'}$ in the
expression for $\tau_{\lambda}$ in Proposition 3.2 in order to 
obtain $\tau_{\lambda'}$.
\item[(b)] The cardinality of the sets $\{{\bf r}\}$ and $\{{\bf r'}\}$ given 
by $\tbinom{n+k}{k}$ is the total number of Young diagrams that fits a $n \times k$
(or a $k \times n$) rectangle.
\end{itemize}
\end{remark}
\noindent For self-conjugate partitions $\lambda=\lambda'$, 
hence $\tau_{\lambda}=\tau_{\lambda'}$. Then \eqref{dualschur} implies that
$W_{\lambda} (\theta) = W_{\lambda}(\omega(\theta))$ and
$W_{\lambda/\lambda({\bf r})}(\theta) = W_{\lambda/\lambda({\bf r})}(\omega(\theta))$      
where $\omega(\theta_j) = (-1)^{j-1}\theta_j$. Consequently, the following
result holds.
\begin{corollary}
The $\tau$-function corresponding to a self-conjugate partition is either
{\it independent} of, or a polynomial of {\it even} degree in the
variables $\theta_{2j}, \,\, j \geq 1$.
\end{corollary}
The classification scheme for the KPI rational solutions
is now complete and summarized below.
\begin{proposition}
For a given positive integer $N$, the $\tau$-function 
for the $N$-lump solutions of KPI fall into 2 distinct classes, (I) and (II).  
The class (I) $\tau$-functions correspond to partitions
$\{\lambda: n \leq \lambda_n\}$ and their conjugates $\lambda'$, where
$n$ is the number of non-zero parts and $\lambda_n$ is the largest part
of the partition $\lambda$. The relation between $\tau_\lambda$ and $\tau_{\lambda'}$
is given by Proposition 3.2. The class (II) $\tau$-functions correspond to self-conjugate 
partitions $\lambda=\lambda'$ and satisfy Corollary 3.1.
The total number of distinct $N$-lump solutions in class (I) is $p(N)-k(N)$
while that of class (II) is $k(N)$, where $p(N)$ is the total number of partitions 
of $N$ and $k(N)$ is the total number of partitions of $N$ into distinct, odd parts. 
\end{proposition}  
\noindent We conclude this section with an illustrative example of class (I) and (II)
partitions.
\begin{example}
Consider $N=4$ whose partitions are enumerated in Example 3.5. For this example,
we pick class (I) partitions $\lambda=(1,3)$ and and its conjugate $\lambda'=(1^2,2)$ with
degree vectors $(m_1,m_2) = (1,4)$ and $(m'_1,m'_2,m'_3)=(1,2,4)$. The associated
matrices $W$ and $W'$ are $5 \times 2$ and $5 \times 3$ respectively, and are given by 
\[W=\begin{pmatrix} p_1 & p_4 \\ p_0 & p_3 \\ 0 & p_2 \\ 0 & p_1 \\ 0 & p_0 \end{pmatrix}\,,
\qquad \qquad \quad 
W'=\begin{pmatrix} p_1 & p_2 & p_4 \\ p_0 & p_1 & p_3 \\ 0 & p_0 & p_2 \\ 0 & 0 & p_1 \\ 
0 & 0 & p_0 \end{pmatrix} \,. 
\] 
Each give rise to $10$ maximal minors out of which $W(23)=W(24)=W(34)=0$ and
correspondingly, $W'(034)=W'(134)=W'(234)=0$. The Schur functions are   
\[W_{(1,3)} = W(01) = \begin{vmatrix} p_1 & p_4 \\ p_0 & p_3 \end{vmatrix}
= \sf{\theta_1^4}{8} +\sf{\theta_1^2\theta_2}{2}-\sf{\theta_2^2}{2}-\theta_4\,,
 \quad W_{(1^2,2)} = W'(012) = \begin{vmatrix} p_1 & p_2 & p_4 \\ 
p_0 & p_1 & p_3 \\ 0 & p_0 & p_2 \end{vmatrix} = \sf{\theta_1^4}{8} 
-\sf{\theta_1^2\theta_2}{2}-\sf{\theta_2^2}{2}+\theta_4 \,,  \] 
after using \eqref{pn} to calculate the $p_n$'s. Notice the symmetry 
$ W_{(1^2,2)}(\theta_1,\theta_2,\theta_3,\theta_4) = 
W_{(1,3)}(\theta_1,-\theta_2,\theta_3,-\theta_4)$ given by \eqref{dualschur}.
Next, consider the multi-index ${\bf r}=(12)$ and the associated partition
$\lambda(12) = (1,1)$. Thus, 
$Y_{1^2} = \raisebox{0.2em}{\ytableausetup{boxsize=0.4em}\ydiagram{1,1}}
\subset
Y_{1,3} = \raisebox{0.2em}{\ytableausetup{boxsize=0.4em}\ydiagram{3,1}}$.  
The conjugate $\lambda(12)' = \lambda'({\bf r'}) = (0,0,2)$, and
$Y_{(0^2,2)}=\raisebox{0.2em}{\ytableausetup{boxsize=0.4em}\ydiagram{2,0,0}} 
\subset  
Y_{(1^2,2)}=\raisebox{0.4em}{\ytableausetup{boxsize=0.4em}\ydiagram{2,1,1}}$.
Then the corresponding multi-index ${\bf r'}=(014)$ is obtained by the relation
$r'_j = \lambda'_j+j-1, \,\, 1 \leq j \leq 3$. 
The associated skew Schur functions are 
\[ W_{(1,3)/(1^2)}=W(12)=\begin{vmatrix} p_0 & p_3 \\ 0 & p_2 \end{vmatrix}
= \sf{\theta_1^2}{2}+\theta_2 \,, \qquad \quad
W_{(1^2,2)/(0^2,2)} = W'(014) = \begin{vmatrix} p_1 & p_2 & p_4 \\ 
p_0 & p_1 & p_3 \\ 0 & 0 & p_0 \end{vmatrix}=\sf{\theta_1^2}{2}-\theta_2\,,   \]  
which again demonstrates the symmetry 
$W_{(1^2,2)/(0^2,2)}(\theta_1,\theta_2) = W_{(1,3)/(1^2)}(\theta_1,-\theta_2)$
of \eqref{dualschur}.   

Now consider a self-conjugate partition of $N=4$ as an example of class (II). 
There is only one such partition $\lambda=(2^2)$ according to Example 3.5.
The degree vector is $(m_1,m_2) = (2,3)$. The multi-index sets  
${\bf r} = (r_1r_2)$ and the partitions $\lambda({\bf r})$ are enumerated in 
Example 2.1 and Example 3.3, respectively.
The corresponding matrix $W$ from \eqref{P}, the Schur function 
$W_{(2^2)}=W(01)$ and the skew Schur function
$W_{(2^2)/(0,1)} = W(02)$ are   
\[ W = \begin{pmatrix} p_2 & p_3 \\ p_1 & p_2 \\ p_0 & p_1 \\ 0 & p_0 \end{pmatrix} \,,
\qquad \quad W_{(2^2)} = \begin{vmatrix} p_2 & p_3 \\ p_1 & p_2 \end{vmatrix} = 
\sf{\theta_1^4}{12} + \theta_2^2 - \theta_1\theta_3\,,  \quad \qquad
 W_{(2^2)/(0,1)} = \begin{vmatrix} p_2 & p_3 \\ p_0 & p_1 \end{vmatrix} = 
\sf{\theta_1^3}{3} - \theta_1\theta_3 \,. \]
Notice that $W_{(2^2)}$ is of degree 2 in $\theta_2$ while $W_{(2^2)/(0,1)}$ is 
independent of $\theta_2$, consistent with Corollary 3.1.
\end{example}
\section{Long time asymptotics of $N$-lumps} 
In this section we will examine the solution structure and certain 
properties of the $N$-lump solutions of the KPI equation. 
In spite of having an exact solution it is
often difficult to analyze $u(x,y,t)$ for arbitrary choices of variables or 
underlying parameters unless one (or more) of them are assumed to be
very small (or large). A natural choice that is of physical interest is to 
investigate the behavior of the solution $u(x,y,t)$ including the wave pattern 
in the $xy$-plane when $|t| >> 1$.

The examples in Section 2.3 present evidence that the $N$-lump solution
for $N=2,3$ separates into $N$ distinct peaks whose heights approach the
$1$-lump peak height asymptotically as $|t| \to \infty$.
Moreover, the peak locations scale as $|t|^{1/p}, \, p=2,3$ 
and (generically) admit an asymptotic expansion as $|t| \to \infty$ of the form  
\[ Z_j(t) := r_j(t) + is_j(t) \sim |t|^{1/p}\big(\xi_{j0} + \xi_{j1} \epsilon
+\xi_{j2} \epsilon^2 + \cdots \big)\,, \qquad \epsilon = |t|^{-1/p} \,,\]
where $Z_j(t)$ is the $j^{\mathrm th}$ peak location, $j=1,2,3$.
In this section, we will show that the above features also hold for
an arbitrary positive integer $N$. Further evidence of such behavior exhibited  
by a special family of multi-lump solutions was presented in a recent 
paper~\cite{CZ21} by the authors.

The key point of the analysis is to establish the fact that to leading order
(in time), the solution $u(x,y,t)$ is localized around a finite number of
peaks (local maxima) in the $xy$-plane; and the dynamics of these peaks occur
at a slow time scale $|t|^{1/p}$ for some $p>0$, in the co-moving frame of \eqref{rs}.
As evident from Proposition 2.1, The KPI solution $u(x,y,t)$ given by \eqref{u} 
is a globally regular rational function for each fixed $t$ in the $xy$-plane, 
decaying as $\sqrt{x^2+y^2} \to \infty$. Hence $u$ has local maxima and
minima in the $xy$-plane but they are too complicated to calculate exactly, in
general. Instead, we note first that the expression
\[u = 2 \ln(\tau_\lambda)_{xx} = 
2\left(\frac{\tau_{\lambda xx}}{\tau_\lambda}
-\big(\frac{\tau_{\lambda x}}{\tau_\lambda}\big)^2\right) \]
suggests that local maxima for $u$ occur approximately
near the {\it minima} of $\tau_\lambda$ where $\tau_{\lambda x} = 0$
and $\tau_{\lambda xx} > 0$ so that in the above expression for $u$
the first term is positive and the second (negative term) vanishes.
Secondly, from either \eqref{square} or
Proposition 3.2 together with \eqref{skewschur-a}, it follows that 
$\tau_\lambda$ is approximately minimized 
when the leading order $|\cdot|^2$ term in $\tau_\lambda$ vanishes.
That is, when the leading order maximal minor $Q({\bf 0})=0$. Therefore,
the peaks of the $N$-lump solution $u(x,y,t)$ are located approximately
near the zeros of $Q({\bf 0})$. It can be shown that for $|t| \gg 1$,
the exact and approximate peak locations differ by $O(|t|^{-1/p})$ in a
very similar manner as outlined in the Appendix of~\cite{CZ21} for a 
special class of $N$-lump solutions. Those calculations will not be
repeated here.
\subsection{Asymptotic peak locations} 
Recall from equations \eqref{square}, \eqref{schur} and \eqref{skewschur-a}
that $Q({\bf 0})$ is a weighted
homogeneous polynomial of degree $N$ in $\theta_j$'s and the parameter $b$
(weight$(b)=-1$). In order to investigate the zeros of $Q({\bf 0})$ in the 
$xy$-plane we first need to express the Schur and skew Schur functions
in terms of the $\theta_j$-variables. Such a representation is available
from the representation theory of symmetric group $S_N$ where the Schur
function $W_{\lambda}$ expresses the irreducible characters of $S_N$ in 
terms of the symmetric functions (see e.g.~\cite{FH91,M37,M95}). $W_\lambda$
is given by
\begin{equation}
W_\lambda(\theta) = \sum_{\alpha_j \geq 0} \chi^\lambda(\alpha)
\frac{\theta_1^{\alpha_1}\cdots\theta_N^{\alpha_N}}{\alpha_1!\cdots\alpha_N!}\,,
\qquad \alpha_1 + 2\alpha_2+\cdots+N\alpha_N = N\,,
\label{chi}
\end{equation}
where $\chi^\lambda(\alpha) \in \mathbb{Z}$ is the character of $S_N$ 
corresponding to the irreducible representation $\lambda$, and the class 
denoted by $(\alpha) =(1^{\alpha_1},2^{\alpha_2},\ldots,N^{\alpha_N})$ 
which is the cycle-type of all permutations in a given class of $S_N$.
Note from Remark 3.1(b) that $(\alpha)$ also denotes partition of a 
positive integer $N$. Hence, both the irreducible representations and 
characteristic classes of $S_N$ are enumerated by integer partitions. 
\begin{example}
For $N=3$, there are 3 partitions:\, $\{(1^3), (1,2), (3)\}$ labeling 
the irreducible representations of $S_3$ which is the permutation group
of 3 indices. $S_3$ has 6 elements which (in the one-line notation of
permutations) can be listed as follows:\, $(123)$ which has 3 1-cycles;
$(213), (321), (132)$ consisting of a 1- and a 2-cycle; and 2 3-cycles
$(312), (231)$. Thus $S_3$ has 3 classes denoted by the cycle-types
$(1^3), (1^1,2^1), (3^1)$ where the superscripts denote the multiplicities
$\alpha_j$ of the $j$-cycle. Thus there are 9 characters $\chi^\lambda(\alpha)$ 
for $S_3$.

Consider the partition $\lambda=(3)$, then $n=1$ and $m_1=3$. From
\eqref{ss}, $W_{(3)} = p_3 = \sf{\theta_1^3}{6}+\theta_1\theta_2+\theta_3$.
Comparing with \eqref{chi}, it follows that 
$\chi^{(3)}(1^3)=\chi^{(3)}(1^1,2^1)=\chi^{(3)}(3)=1$. If $\lambda=(1,2)$ then 
from Example 3.4, $W_{(1,2)}=p_1p_2-p_3 = \sf{\theta_1^3}{3}-\theta_3$.
The characters in this case are:\, $\chi^{(1,2)}(1^3)=2,\,\chi^{(1,2)}(1^1,2^1)=0,\,
\chi^{(1,2)}(3) =-1$.
\end{example}
\noindent The characters $\chi^{\lambda}(\alpha)$ 
satisfy the orthogonality relations  
\[ \sum_{(\alpha)}\frac{\chi^{\lambda}(\alpha)\chi^{\lambda'}(\alpha)}{N_\alpha} 
= \delta_{\lambda\lambda'}\,, \quad \qquad
\sum_{\lambda} \chi^{\lambda}(\alpha)\chi^{\lambda}(\alpha') 
= N_\alpha \delta_{\alpha\alpha'}\,, \qquad 
N_\alpha = \prod_{j=1}^Nj^{\alpha_j}\alpha_j! \,,
\]
where the first sum is over all classes $\{(\alpha)\}$ of $S_N$ and the second
sum is over all partitions $\{\lambda\}$ of $N \in \mathbb{N}$.
The orthogonality relations for the characters are now used to obtain a 
result that will be useful for this section.
The Schur functions $W_\lambda(\theta)$ satisfies a property under the shift of variables
$\theta+h := (\theta_1+h_1, \cdots \theta_N+h_N)$ similar to
\eqref{propc} namely,
\begin{equation}
W_{\lambda}(\theta+h) = \sum_\mu W_\mu(h)W_{\lambda/\mu}(\theta)\,, \qquad 
\mu \subseteq \lambda \,.    
\label{shift} 
\end{equation}
Equation \eqref{shift} can be derived by a Taylor expansion of $W_{\lambda}(\theta+h)$ 
followed by the orthogonality relations of the $\chi^{\lambda}(\alpha)$, and 
then using the result regarding the skew Schur functions mentioned 
in Remark 3.1(c)~\cite{OSTT88,K18,M95}.    
Applying \eqref{shift} to the sum on the right hand side of \eqref{skewschur-a} 
$Q({\bf 0})$ can be expressed as a {\it single} Schur function
\begin{equation}
Q({\bf 0})(\theta) =  i^{|{\bf 0}|}\Big(W_{\lambda} (\theta)+ 
\sum_{\emptyset \neq \lambda({\bf s}) \subset \lambda} 
W_{\lambda({\bf s})}(h)W_{\lambda/\lambda({\bf s})}(\theta)\Big)
=i^{|{\bf 0}|}W_\lambda(\theta+h) \,, 
\label{Wshift}
\end{equation}
by appropriately choosing the $h_j$'s such that
$W_{\lambda({\bf s})}(h) = i^{|{\lambda(\bf s})|}U\tbinom{\bf 0}{\bf s}$. 
For example, if ${\bf s}=(01 \cdots (n-2) n):={\bf 1}$, 
then the $n \times n$ minor $U\tbinom{\bf 0}{\bf 1}=\sf{n}{2b}$
from the $U$ defined in \eqref{Q}. The size of the partition
$\lambda({\bf 1}) = (0,\cdots,0,1)$ is $|\lambda({\bf 1})|=1$. 
Hence from \eqref{ss},  
$W_{\lambda({\bf 1})}(h) = \Wr(p_0,p_1,\ldots,p_{n-2},p_n) = p_1(h)=h_1$.
Thus $h_1=\sf{i n}{2b}$, and all other $h_j$'s can be 
computed successively. The first 3 values of $h_j$ which will be used below
are listed here
\begin{equation}
h_1 = \frac{i n}{2b}\,, \qquad h_2= - \frac{n}{2(2b)^2}\,, \qquad
h_3= - \frac{i n}{3(2b)^3} \,. 
\label{h}
\end{equation}
Equation \eqref{Wshift} implies that the approximate
location of the peaks of the KPI $N$-lump solutions are given by
\begin{equation}
Q({\bf 0})(\theta) = 0 \quad \Rightarrow \quad W_\lambda(\tilde{\theta}) = 0 \,, 
\quad \tilde{\theta} := \theta+h \,, 
\label{Wzero}
\end{equation}
that is, by the zeros of the shifted Schur function $W_\lambda(\theta+h)$.

Recall from \eqref{theta}, that it is the first three variables 
$\theta_j \,\, j=1,2,3$ that depend on $x,y,t$, and in particular,
the $t$-dependence occurs via $\theta_2, \theta_3$ which are linear in $t$.
The rest of the variables $\theta_j =i\gamma_j, \, j>3$ are arbitrary 
$O(1)$ constants with higher weights since  
weight$(\theta_j)=j$. Thus when $|t| \gg 1$, one can set 
$\alpha_j=0, \,\, j>3$ in \eqref{chi} to obtain the dominant behavior
\begin{equation}
W_\lambda(\tilde{\theta}) \sim 
\sum_{\alpha_1,\alpha_2,\alpha_3 \geq 0}\!\!\! 
\chi^\lambda(\alpha')
\frac{\tilde{\theta}_1^{\alpha_1}\tilde{\theta}_2^{\alpha_2}\tilde{\theta}_3^{\alpha_3}}
{\alpha_1!\alpha_2!\alpha_3!} = \Wr(p_{m_1}, \ldots, p_{m_n}) \,,
\label{Wasymp}
\end{equation}
where $(\alpha'):=(1^{\alpha_1}, 2^{\alpha_2}, 3^{\alpha_3})$ with
$\alpha_1+2\alpha_2+3\alpha_3=N$,  
denote those classes of $S_N$ which consist only of 1-, 2-, and 3-cycles.
The wronskian form of \eqref{Wasymp} follows from \eqref{ss} by restricting
the generalized Schur polynomials to depend only on the first three variables 
by setting $\theta_j=0, \, j>3$ in \eqref{pn}, i.e., 
$p_{m_j} = p_{m_j}(\tilde{\theta}_1, \tilde{\theta}_2,\tilde{\theta}_3)$.
In order to locate the $N$-lump peaks, \eqref{Wzero} is to be 
solved asymptotically for $\theta_1(t)$ as $|t| \gg 1$ applying \eqref{Wasymp}.
The coefficient of $\theta_1^N$ in $W_\lambda$ is non-zero because
$\chi^\lambda(1^N) \neq 0$, being the dimension of the irreducible representation
$\lambda$~\cite{M37}. So the dominant balance for $|t| \gg 1$ arises from
$\tilde{\theta}_1^N \sim 
\tilde{\theta}_1^{\alpha_1}\tilde{\theta}_2^{\alpha_2}\tilde{\theta}_3^{\alpha_3}$
where $\alpha_1+2\alpha_2+3\alpha_3=N$. This yields, 
$\tilde{\theta}_1 \sim |t|^{1/p}, \, |t| \gg 1$,
where $\sf{1}{2} \leq \sf{1}{p} = \sf{\alpha_2+\alpha_3}{2\alpha_2+3\alpha_3} \leq \sf{1}{3}$
as claimed in~\cite{ACTV00}. However, it was observed in~\cite{YY21} that
the extreme values $p=2,3$ are the only two possibilities although
no explanation was offered. In what follows, we shall first establish 
that is indeed the case.

In terms of their long time behavior the $N$-lump solutions can be grouped
into two mutually exclusive classes depending on whether $W_{\lambda}(\theta)$ 
does or does not depend on the variable $\theta_2$. If $W_{\lambda}$ is
independent of $\theta_2$ then $\tilde{\theta}_1^N \sim 
\tilde{\theta}_1^{\alpha_1}\tilde{\theta}_3^{\alpha_3}$
where $\alpha_1+3\alpha_3=N$, is the dominant balance in \eqref{Wzero}, implying that
$\tilde{\theta}_1 \sim |t|^{1/3}, \, |t| \gg 1$. But if 
$W_{\lambda}$ does depend on $\theta_2$ then we show in Section 4.1.2
that there exists at least one positive integer
$q$ such that the coefficient of the $\theta_1^{N-2q}\theta_2^q$ is non-zero
in $W_\lambda$. Then the dominant balance required to solve \eqref{Wzero} is    
$\tilde{\theta}_1^N \sim \tilde{\theta}_1^{N-2q}\tilde{\theta}_2^{q}$ 
implying that $\tilde{\theta}_1 \sim |t|^{1/2}, \, |t| \gg 1$. 
Furthermore, an interesting asymptotic feature is observed in this case,
that arises from yet another dominant balance $\tilde{\theta}_1 \sim |t|^{1/3}$
(see also~\cite{YY21}). We provide an explanation of this wave phenomena
in Section 4.1.2. These two distinct classes of $N$-lump
solutions exhibit very different surface wave patterns, and are described in 
more details below. 
\subsubsection{Triangular multi-lumps}
First we consider the case when the Schur function $W_\lambda$ is 
independent of $\theta_2$. The corresponding class of $N$-lump solutions 
will be referred to as the triangular lumps. In a sense the $1$-lump solution
belong to this class but the simplest non-trivial example is the $3$-lump
solution in Section 2.3.3. The following result provides the main characterization
of this class.
\begin{lemma}
Let $\lambda$ be a partition of $N \in \mathbb{N}$ with degree vector
$(m_1,\dots,m_n)$ and let $W_{\lambda}$ and $W_{\lambda/\mu}$ be as 
in \eqref{ss}. Then $W_\lambda, W_{\lambda/\mu}$ are independent of $\theta_2$ 
if an only if the degree vector $(m_1,\dots,m_n)=(1,3,\ldots,2n-1)$ such that 
$\lambda = (1,2,\ldots,n)$ is a self-conjugate partition
of the triangular number $N=\sf{n(n+1)}{2}$. In fact, $W_\lambda$ is 
independent of all the even variables $\theta_{2j}, \, j\geq 1$
if $\lambda = (1,2,\ldots,n)$.
\end{lemma} 
\begin{proof}
If $(m_1,\dots,m_n)=(1,3,\ldots,2n-1)$ then differentiating each of the
two determinants in \eqref{schur} with respect to $\theta_2$ splits
it up into $n$ determinants where the $j^{\mathrm th}$ column is
differentiated in the $j^{\mathrm th}$ component determinant.
Applying \eqref{propa} renders the first column of the first component
determinant to vanish, and the $(j-1)^{\mathrm th}$ and $j^{\mathrm th}$ 
columns are identical for the $j^{\mathrm th}$ component determinant 
for $j \geq 2$.

Conversely, differentiating $W_\lambda$ in \eqref{ss} with respect to 
$\theta_2$ leads to
$\sum_j\Wr(p_{m_1},\ldots,p_{m_j-2},\ldots,p_{m_n})=0$
after using \eqref{propa} for $1 \leq j \leq n$. Each term in the 
sum is itself a Schur function 
corresponding to a partition of $N-2$. This set of Schur functions
is linearly independent since it forms a basis 
for the vector space of symmetric functions of degree $N-2$ over integers. 
Hence, $\Wr(p_{m_1},\ldots,p_{m_j-2},\ldots,p_{m_n})=0$ for each $j$.
For each wronskian to vanish, either a column must identically vanish or 
two successive column must be identical since $m_1< \cdot < m_n$. 
This implies that $m_{j-1}=m_j-2$ for 
$1 < j \leq n$, and for $j=1$, $p_{m_1-2}=0$. That is, $m_1-2 < 0$
from \eqref{propa}. But by hypothesis $m_1 \geq 1$, hence $m_1=1$.
The rest of the proof is similar.
\end{proof} 
A triangular number $N=n(n+1)/2$ can be always expressed as either $N=3m$ or
$N=3m+1$ for some $m \in \mathbb{N}$. Then applying Lemma 4.1 to \eqref{Wshift} 
and denoting the shifted Schur function as $W_\triangle$, yield 
\begin{equation*}
W_{\triangle}(\tilde{\theta}) = \Wr(p_1,p_3,\ldots,p_{2n-1})
\sim \sum_{r=0}^m \chi^\triangle(1^{N-3r},3^r)
\frac{\tilde{\theta}_1^{N-3r}\tilde{\theta}_3^{r}} {(N-3r)!r!} \,, 
\qquad |t| \gg 1\,,
\end{equation*}
where \eqref{Wasymp} is utilized in order to collect the dominant terms
in $t$ and neglect terms involving $\tilde{\theta}_{2j+1}, \, j > 1$.
Then it is clear from above that the dominant balance required 
to asymptotically solve $W_\triangle = 0$ must be 
$\tilde{\theta}_1^N \sim \tilde{\theta}_1^{N-3r}\tilde{\theta}_3$
so that $\tilde{\theta}_1 \sim |t|^{1/3}$ for $|t| \gg 1$.
Substituting from \eqref{thetars}
$\tilde{\theta}_1=iz, \,\, z=r+h_1+is+\gamma_1, \,\, \tilde{\theta}_3=i(t+h_3+\gamma_3)$
where $h_1, h_3$ are from \eqref{h}, and using the dominant balance to 
rescale $z=-(\sf{t}{a})^{1/3}\xi$ where $a \neq 0$ is a constant,
in the above asymptotic expression of $W_\triangle$, lead to
\begin{equation}
W_\triangle(r,s,t) \sim 
t^{N/3}\Big[(\sf{i}{a})^{N/3}\frac{\chi^\triangle(1^N)}{N!}\,Q_n(\xi)+O(t^{-1})\Big]\,. 
\label{triangle}
\end{equation}
The polynomial $Q_n(\xi)$ in \eqref{triangle} is derived as follows. Lemma 4.1 
implies that the wronskian $\Wr(p_1,p_3,\ldots,p_{2n-1})$ is independent of 
$\theta_2$, hence it can be evaluated by setting $\tilde{\theta}_2=0$ in 
the generalized Schur polynomials, yielding  
$$p_r(\tilde{\theta}_1, 0,\tilde{\theta}_3) 
\sim p_r(iz, 0, it) = \Big(\frac{-it}{a}\Big)^{r/3}\, q_r(\xi,a)
$$
to leading order when $|t| \gg 1$. The polynomials $q_r(\xi,a)$
together with the generating function are readily obtained from \eqref{pn} and \eqref{gen}
\begin{equation*}
q_r(\xi,a)= \sum_{k,l \geq 0}\frac{\xi^k a^l}{k!l!}, \quad k+3l=r\,, \qquad
\exp(\alpha \xi+\alpha^3 a) = \sum_{r=0}^\infty q_r(\xi, a)\alpha^r\,.
\end{equation*}
Then the leading order term in $W_\triangle=\Wr(p_1,p_3,\ldots,p_{2n-1})$ 
corresponds to the first term in \eqref{triangle} where $\Wr(q_1,q_3,\ldots,q_{2n-1}) =
\sf{\chi^\triangle(1^N)}{N!}Q_n(\xi), \, N=n(n+1)/2$. The polynomials
$Q_n(\xi)$ normalized by choosing the scale factor $a=-\sf{4}{3}$ 
are known as the Yablonskii-Vorob'ev polynomials which were
originally studied to obtain special rational solutions of the second Painlev\'e 
equation PII~\cite{Y59,V65}. These
are monic polynomials of degree $N=n(n+1)/2$, with integer coefficients, defined by
\[ Q_n(\xi)= \xi^{N-3m}\sum_{r=0}^m (-\sf{4}{3})^rc_r\xi^{3(m-r)} \,,
\qquad \quad c_r = \frac{\chi^\triangle(1^{N-3r},3^r)N!}{\chi^\triangle(1^N)(N-3r)!r!}\,,
\quad \quad m=\left\{ \begin{matrix} \sf{N}{3} \,, & \quad N=3m  \vspace{0.05 in} \\    
\sf{N-1}{3} \,, & \quad N=3m+1 \end{matrix} \right. \,,
\]
and are known to have $N$ distinct roots in the complex plane~\cite{FOU00}.
Then it follows immediately from the above expression for $Q_n(\xi)$ that $\xi=0$ is a root
if and only if $N=3m+1$ (also, then $n \equiv 1$ mod $3$) and the non-real roots 
arise as complex conjugate pairs. Moreover, the roots have a ``triangular''
symmetry:\, $\xi \to \xi e^{2\pi i/3}$ since $Q(\xi) = \xi^{N-3m}P(\eta), \,
\eta=\xi^3$. Then for each root $\eta^{(j)}, \, j=1,\ldots,m$ of 
$P(\eta)$, $Q(\xi)$ has a corresponding triplet of roots $\xi^{(j)}_i, \, i=1,2,3$ 
which lie $\sf{2 \pi}{3}$ apart on a circle of radius $|\eta^{(j)}|^{1/3}$ centered at 
the origin. There exists an extensive literature on the Yablonskii-Vorob'ev 
polynomials and patterns of its roots in the complex plane (see e.g.~\cite{CM03}
and references therein). They are also related to rational solutions of the KdV and 
modified KdV equations~\cite{KO96}, so it is natural for them to appear in the context of
the rational solutions of KPI.

Note that the $O(t^{-1})$ term in \eqref{triangle} arises from all
the subdominant terms not included in \eqref{Wasymp} as well as from 
terms linear in the shift $h_3$ obtained when expanding
$\tilde{\theta}_3^r$ in the wronskian $\Wr(p_1,\ldots,p_{2n-1})$.
Then the solution $W_{\triangle}(r,s,t)=0$ in \eqref{triangle} has 
the asymptotic form:\, $\tilde{\theta}_1 =iz(t)$,
$$z(t) = t^{1/3}(z_0 + \epsilon z_1 + \epsilon^2 z_2 + \cdots)\,,
\qquad \epsilon = t^{-1} \,, \qquad z_0=-\frac{\xi_j}{a^{1/3}} \,,  
$$
where $\xi_j$ is a root of $Q_n(\xi)=0$.
Next, solving for $z=r+h_1+is+\gamma_1$, the approximate locations of the $N$-lump 
peaks for $|t| \gg 1$ are obtained as 
\begin{equation}
Z_j(t):= r_j(t) + is_j(t) = (r_0-\sf{n}{2b} + is_0)-
(\sf{3}{4}t)^{1/3}\big(\xi_j+O(t^{-1})\big)\,, \qquad j=1,2, \ldots, N \,, 
\label{trianglezero}
\end{equation}
where $r_0+is_0 = -\gamma_1$ and $h_1=\sf{n}{2b}$ from \eqref{h}.
Therefore, the corresponding rational solutions of the KPI equation have
$N$ {\it distinct} peaks which form a ``triangular'' pattern in the
co-moving $rs$-plane. Hence they were referred to as the {\it triangular} 
$N$-lump solutions at the top of this subsection. 
The wave pattern has an {\it almost} time-reversal symmetry 
in the sense that $Z_j(t)+Z_j(-t) \sim (r_0+\sf{n}{2b}, s_0)$ is independent
of $t$ for $|t| \gg 1$. 

If $n=3k+1$ for some $k=0,1,\ldots$ then $Q_n(\xi)$ has a root at $\xi=0$ 
and $N=3m+1$ for some $m \in \mathbb{N}$. The corresponding peak for the
$N$-lump solution is approximately located at 
$(r_0-\sf{3k+1}{2b}, s_0) + O(|t|^{-2/3})$ for $|t| \gg 1$. The 
$k=0$ case is the $1$-lump solution. The solution corresponding to $k=1$ i.e., 
$n=4$ is shown below in Figure~\ref{nlumppeaks}. 
The simplest, non-trivial triangular lump solution 
corresponds to $n=2, \, N=3$ which was discussed in details earlier in Section 2.3.3. 
Figure~\ref{nlumppeaks} below compares the exact and the approximate peak 
locations given by \eqref{trianglezero}  
for the cases $n=3,4,5$ corresponding to $N=6, 10, 15$ respectively.
\begin{figure}[h!]
\begin{center}
\includegraphics[scale=0.42]{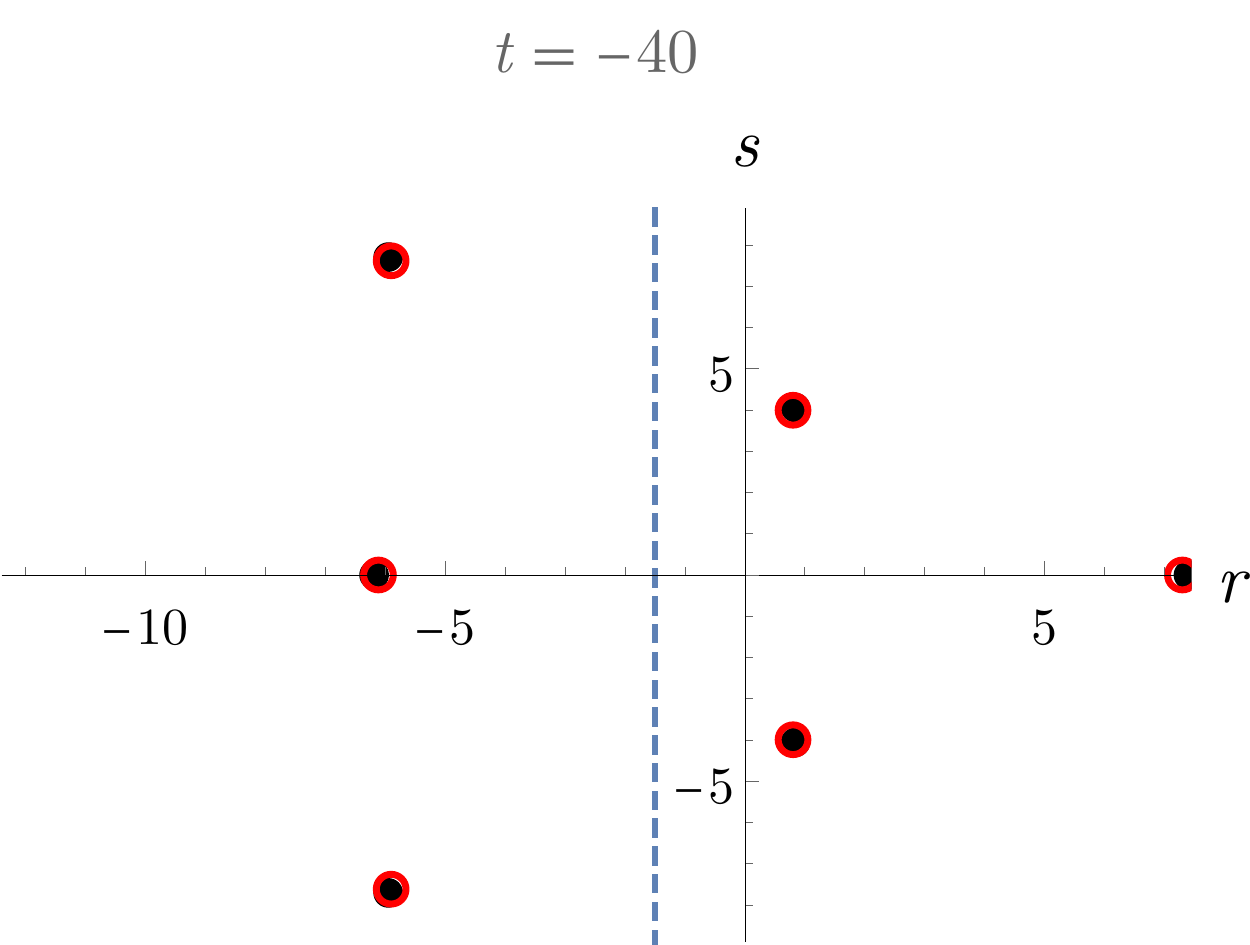} \quad
\includegraphics[scale=0.42]{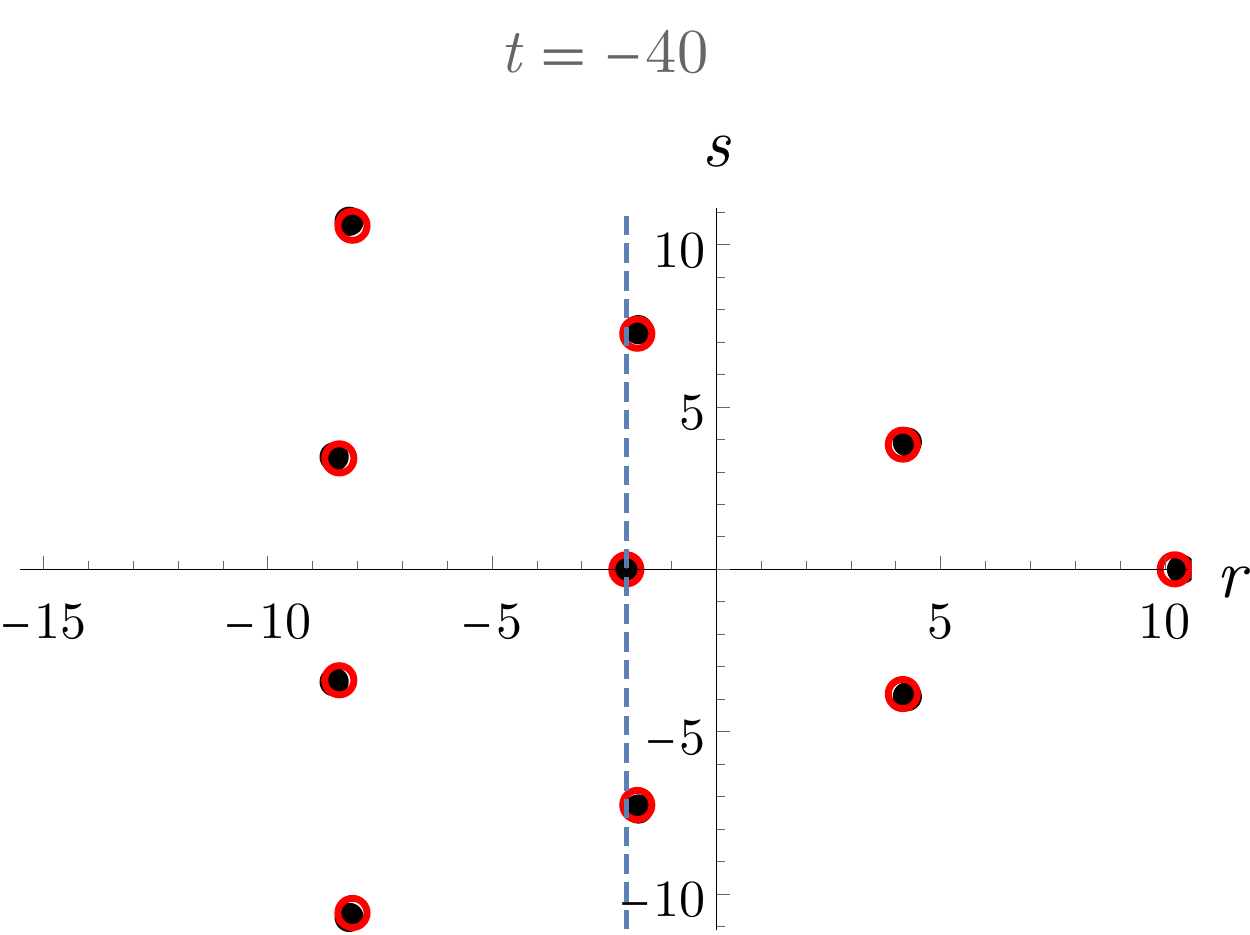} \quad 
\includegraphics[scale=0.42]{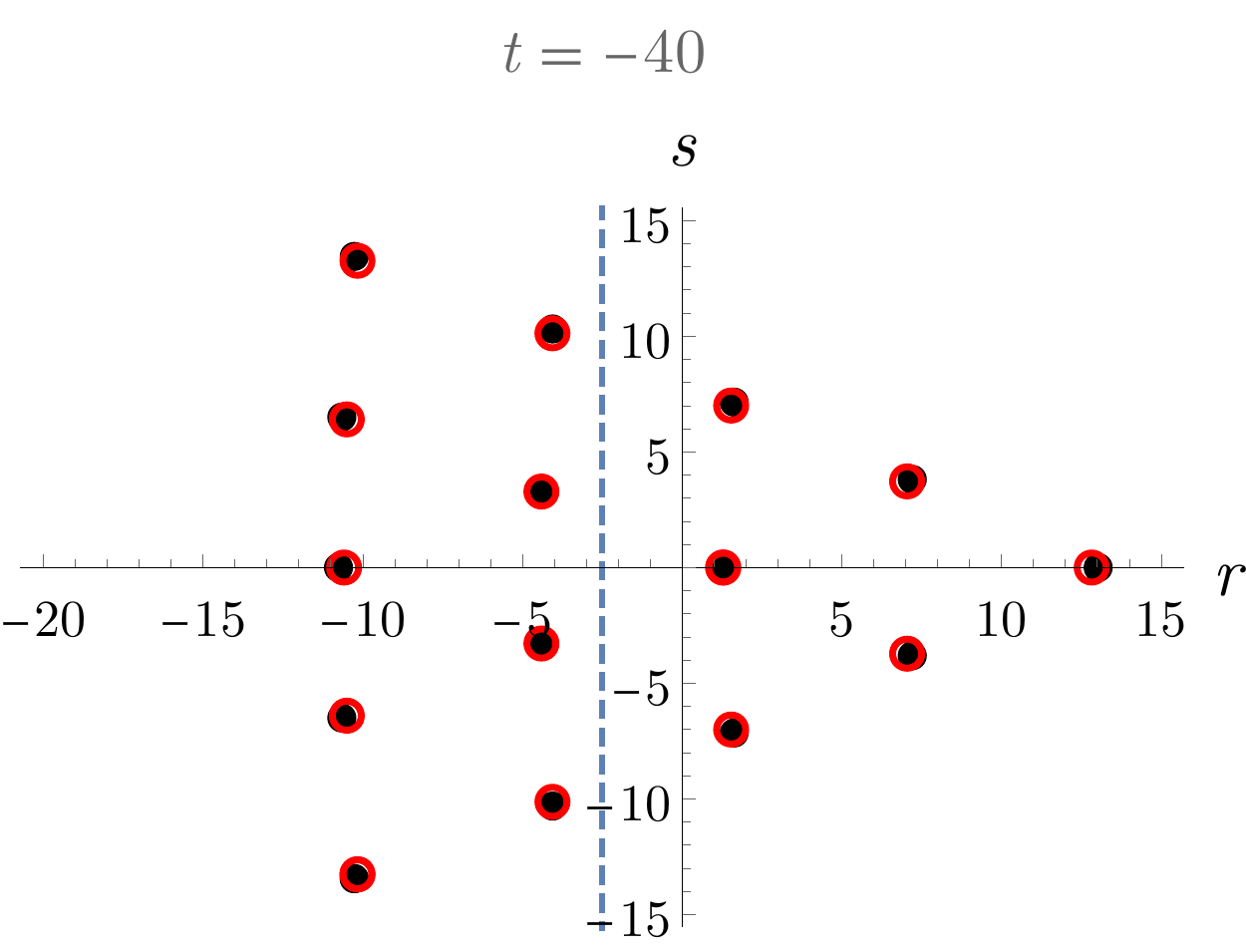} \\
\includegraphics[scale=0.42]{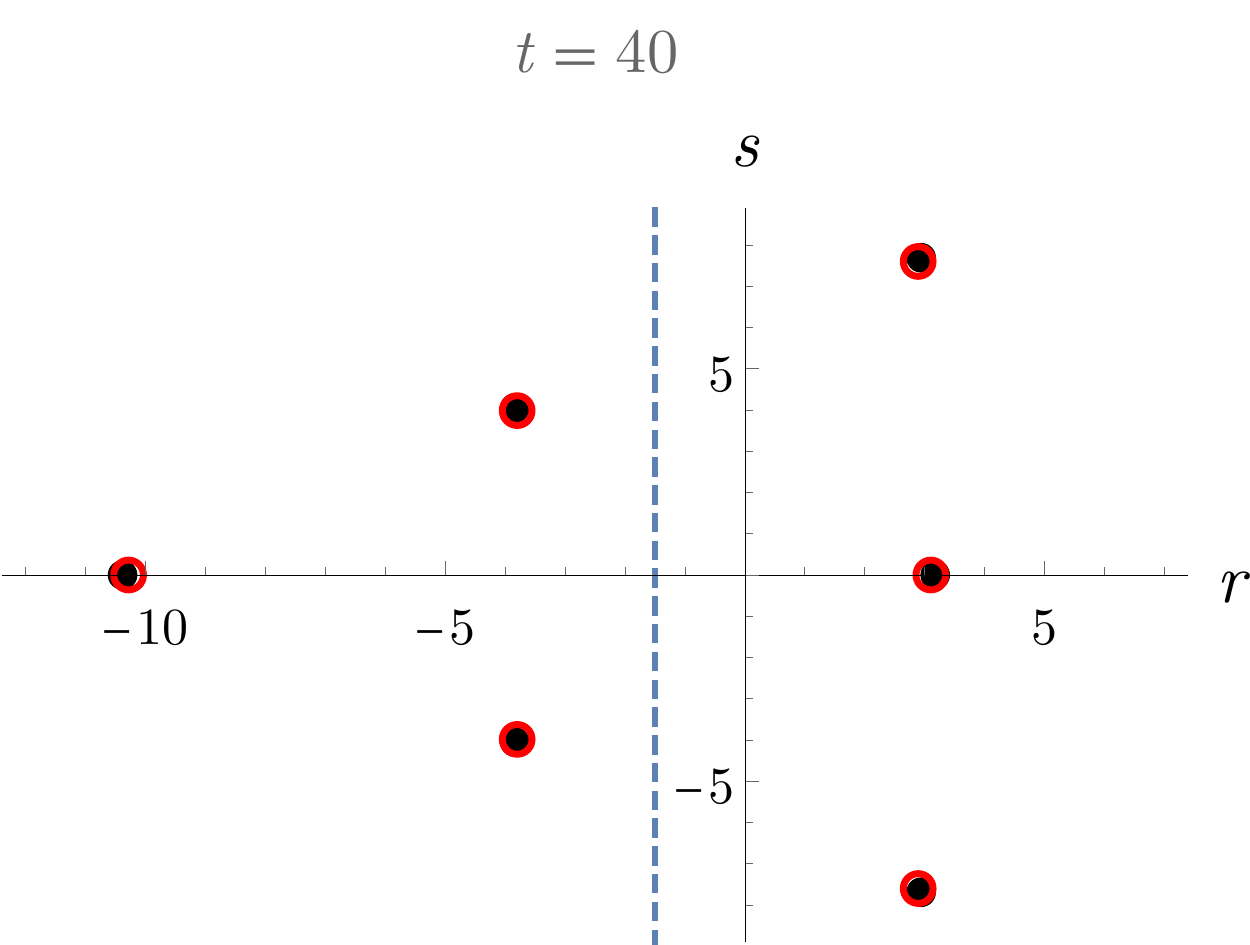} \quad
\includegraphics[scale=0.42]{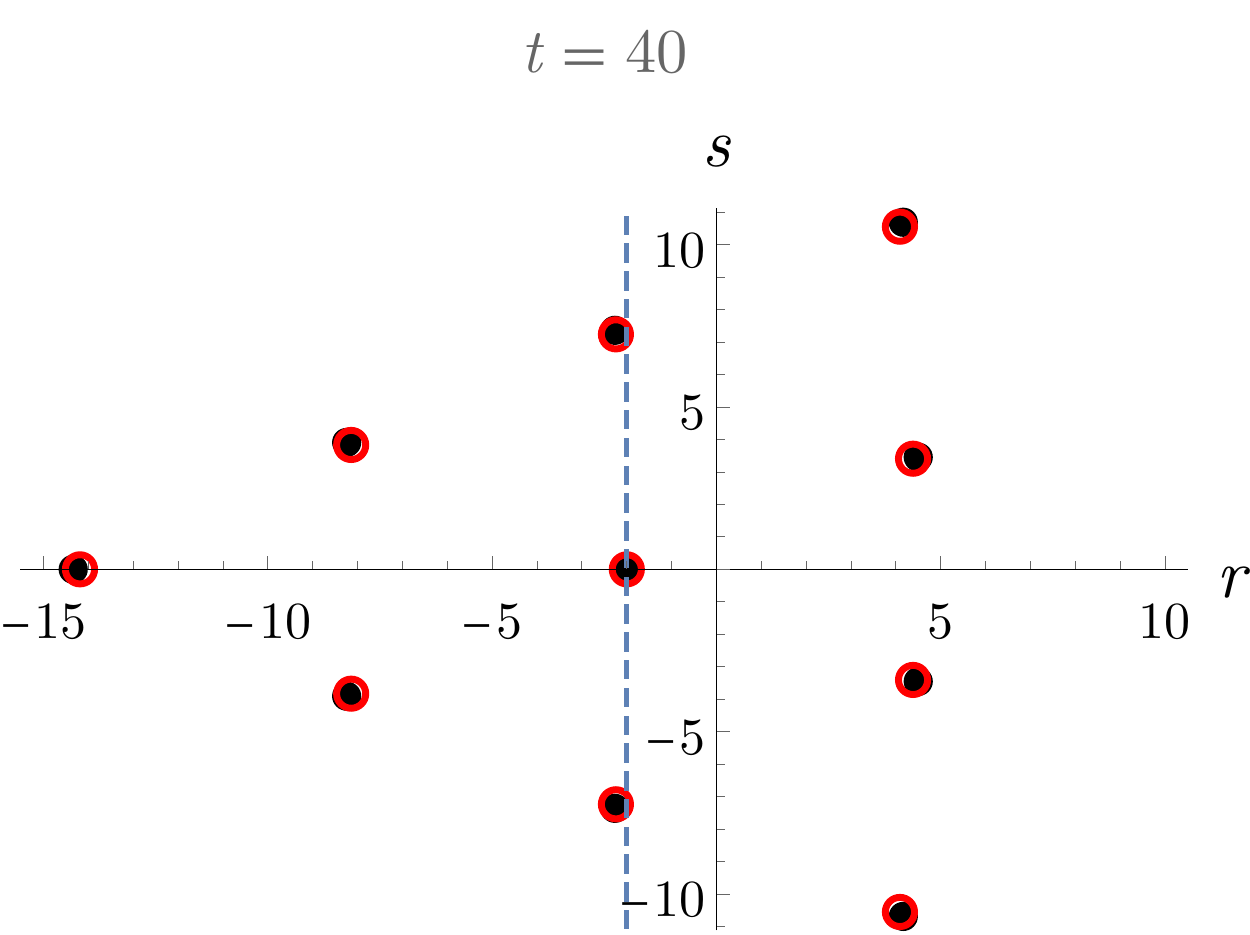} \quad 
\includegraphics[scale=0.42]{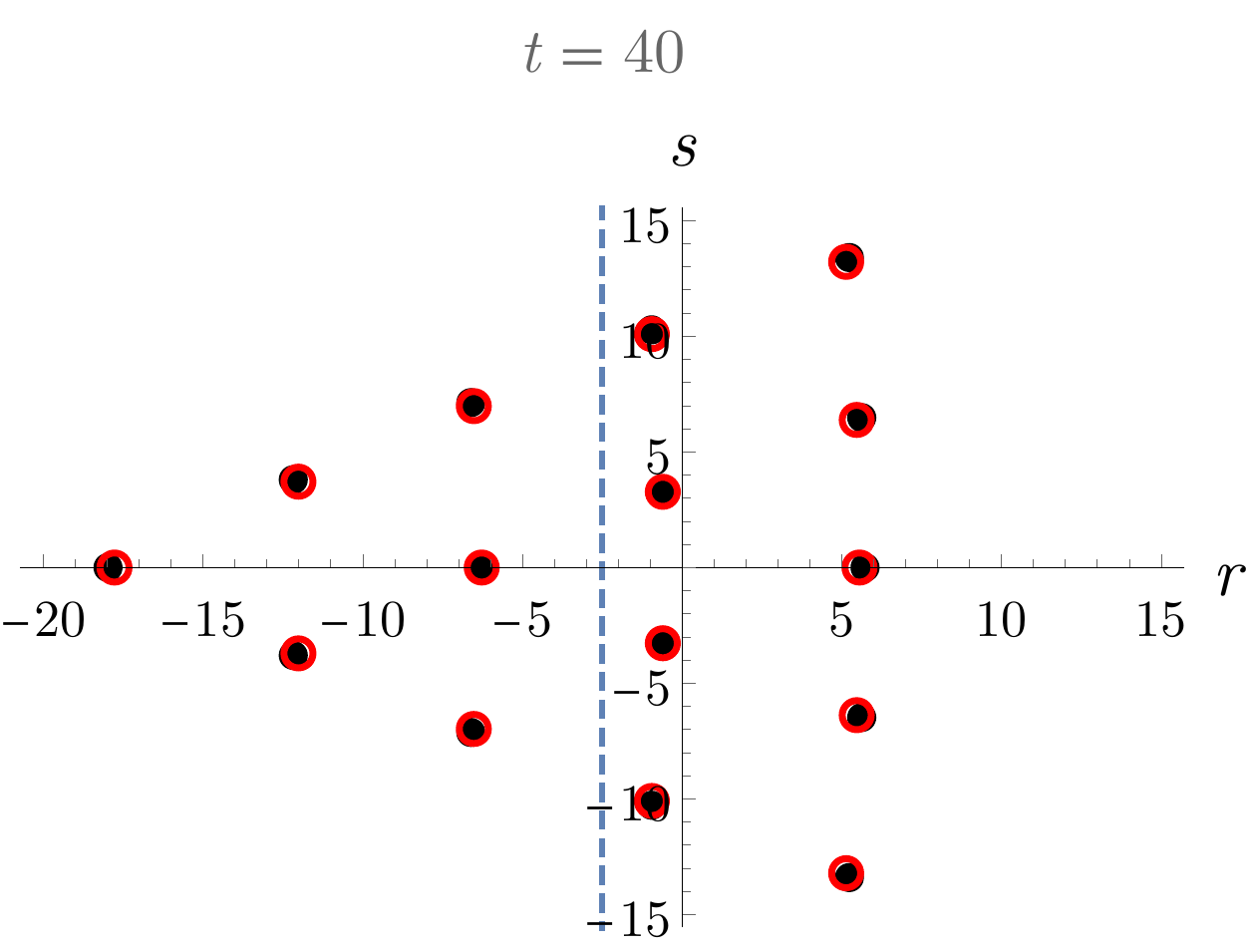}
\end{center}
\vspace{-0.2in}
\caption{Triangular lump peak locations for $n=3,4,5$ at $t=-40$ (top panel) and $t=40$ 
(bottom panel). Total number of peaks $N=\sf{n(n+1)}{2}$ (see text).
Black dots represent locations $z_j$ where $Q({\bf 0})=0$, red open 
circles are exact locations. Under time reversal ($t \to -t$) the peak locations 
are reflected across the dashed vertical line:\, $r=-\sf{n}{2b}$. 
KP parameters: $a=0, b=1$ and $\gamma_j=0, \, j \geq 1$.}
\label{nlumppeaks}
\end{figure}
\begin{remark}
\begin{itemize}
\item[(a)] It should be emphasized that \eqref{trianglezero} gives the
approximate locations of the triangular $N$-lump solutions. The absolute
difference between the exact and approximate locations is $O(|t|^{-1/3})$
when $\xi_k \neq 0$ and $O(|t|^{-1})$ when $\xi_k = 0$.
\item[(b)] The triangular waveform patterns in KPI rational solutions
were previously found in~\cite{GPS93,ACTV00} and more recently
in~\cite{G18,DLZ21,YY21} but no
previous attempt was made to classify these KPI rational solutions.
According to Proposition 3.3, the triangular lumps belong to Class (II) 
corresponding to self-conjugate partitions. 
\item[(c)] The leading order approximate peak locations
is also governed by a well known dynamical
system studied in~\cite{AMM77}. It can be shown using Lemma 4.1 (see also 
Remark 2.2(b)) that the wronskian 
$\Wr(p_1,p_3,\ldots,p_{2n-1})$ is in fact a $\tau$-function $\tau(x,t)$ of 
the KdV equation:\, $4u_t+6uu_x+u_{xxx} =0$ where the generalized
Schur polynomials are redefined as $p_n = p_n(x,y,t)$ instead of 
$\theta_j,\, j\geq 1$ (see Remark 2.1(b)). The corresponding rational
solution of KdV is $u(x,t) = 2(\ln \tau(x,t))_{xx}$
where $\tau(x,t) = (x-x_1(t))(x-x_2(t))\cdots(x-x_N(t))$ are the Adler-Moser
polynomials~\cite{AM78} which are known to have distinct roots if
$N=n(n+1)/2$ for some $n \in \mathbb{N}$.
Plugging this form of $\tau(x,t)$ back into the KdV equation and following
the method outlined in~\cite{DK16} one recovers the first order dynamical system 
for the roots $x_j(t)$ together with the constraints
\[x_{jt} = \sum_{i \neq j}\frac{3}{(x_j-x_i)^2}\,, \quad \qquad
\sum_{i=1, i \neq j}^N\frac{1}{(x_j-x_i)^3}=0\,, \quad j=1,2,\ldots,N \,,\]
found in~\cite{AMM77}. A closer examination of this dynamical system may 
provide further insights into the interaction properties of the triangular
$N$-lump solutions but we do not pursue this matter here.
Substitution of $x_j(t)=-(\sf{3}{4}t)^{1/3}\xi_j$ 
in the dynamical system and the constraints, yield
\[\xi_{j} = -\sum_{i \neq j}\frac{12}{(\xi_j-\xi_i)^2}\,, \quad \qquad
\sum_{i=1, i \neq j}^N\frac{1}{(\xi_j-\xi_i)^3}=0\,, \quad j=1,2,\ldots,N \,,\]
which form a nonlinear system of algebraic equations that may be useful
to compute the roots for the Yablonskii-Vorob'ev polynomials as well as
study their properties.  ` 
\item[(d)] Another possible interesting application 
of the Yablonskii-Vorob'ev polynomials is to obtain explicit formulas for
the characters $\chi^\triangle(1^{N-3r},3^r)$ since there is no easy
algorithm to derive them in general. There are some studies which
derive formulas for the coefficients $c_r$ of $Q_n(\xi)$. For example, it was
shown that $c_{m-1}=0$ if $N=3m+1$~\cite{Ta00} and a formula for $c_m$ has
been derived in~\cite{KO03}, both utilize the recurrence relations for $Q_n(\xi)$.
\end{itemize}
\end{remark}
\subsubsection{General multi-lumps}
We now consider a partition $\lambda \neq (1,2,\ldots,n)$ of $N$ and the 
zeros of the associated shifted Schur function $W_\lambda(\tilde{\theta})$
whose asymptotic form is given in \eqref{Wasymp}. In general, 
$W_\lambda(\tilde{\theta})$ will depend on $\tilde{\theta}_2$. We first look 
for a dominant balance of the type:\, 
$\tilde{\theta}_1^N \sim \tilde{\theta}_1^{N-2j}\tilde{\theta}_2^j$ which
means that the corresponding character coefficient $\chi^\lambda(1^{N-2j},2^j)$
has to be non-zero. To examine this possibility, it is necessary again to
briefly review some ideas from partition theory. We will do so from the
monograph~\cite{M95} (see also the recent work~\cite{BDS20}).

{\bf Murnaghan-Nakayama rule and 2-core of a partition}:\,
A combinatorial way of computing the character $\chi^\lambda(1^{N-2j},2^j)$ 
of the symmetric group $S_N$ is to recursively apply the Murnaghan-Nakayama 
rule which for the present purpose reads as follows
\[\chi^\lambda(1^{N-2j},2^j) = \sum_{\mu \subset \lambda}
(-1)^{ht(S)}\chi^\mu(1^{N-2j},2^{j-1})\,, \]
where the sum is over all partitions $\mu \subset \lambda$ such that
$S=\lambda/\mu$ is either a vertical or horizontal 2-block, i.e. 
either $Y_{\lambda/\mu} = \raisebox{0.05in}{\ytableausetup{boxsize=0.4em}\ydiagram{1,1}}\,,$
or $Y_{\lambda/\mu} = \ytableausetup{boxsize=0.5em}\ydiagram{2}\,,$ and
$ht(S)= \#$ of rows in $S$ minus 1. Thus $ht(S)=1$ for a vertical 2-block
and $ht(S)=0$ for a horizontal one. Successive applications of the 
Murnaghan-Nakayama rule can exhaust all the 2-cycles, which leads to a form
\[\chi^\lambda(1^{N-2j},2^j) = \sum_{\rho}\pm \chi^\rho(1^{N-2j})\,,\]
where $\rho \subset \lambda$ are partitions whose diagrams $Y_\rho$ are
obtained by peeling off $j$ 2-blocks from $Y_\lambda$. The parity $\pm$ depends
on the number of horizontal and vertical 2-blocks peeled away during the
process of obtaining a specific $Y_\rho$ starting from $Y_\lambda$, and the 
sum is over all possible partitions $\rho$ whose diagram $Y_\rho$ consists
of $N-2j$ boxes. The character $\chi^\rho(1^{N-2j}) \neq 0$ is known since
it is the dimension of the irreducible representation $\rho$. Note that if
$\rho=\emptyset, \, \chi^\emptyset=1$. However, this reduction process may lead
to obstructions where a stage is reached when the resulting partition 
$\tilde{\lambda}$ is such that {\it no more} 2-blocks can be peeled away from
$\tilde{\lambda}$ and the recursive process has to prematurely terminate. The partition
$\tilde{\lambda}$ is called the 2-core of the partition $\lambda$. It can be
shown that 2-cores are precisely triangular partitions discussed in Section 4.1.1. 
The key ingredient of this process is to peel off one 2-block at a time so that the
resulting diagram is still a Young diagram of some partition $\mu \subset \lambda$.
The terminating 2-core partition $\tilde{\lambda}$ is always the same no matter
which path is followed. Moreover the parity $\pm$ which depends on the
number of vertical 2-blocks removed along a path is also invariant across all
paths. Thus there exists a positive integer $q$ such that the character
coefficient $\chi^\lambda(1^{N-2q},2^q)$ of $\tilde{\theta}_1^{N-2q}\tilde{\theta}_2^q$ 
of the Schur function $W_{\lambda}$ in \eqref{Wasymp} takes the form
\[\chi^\lambda(1^{N-2q},2^q) = \pm N(\lambda, \tilde{\lambda})\chi^{\tilde{\lambda}}(1^{N-2q}) = 
C(\lambda,\tilde{\lambda}) \neq 0 \]  
for some triangular partition $\tilde{\lambda}$ and where
$N(\lambda, \tilde{\lambda})$ is the total number of possible paths
from $Y_\lambda$ to $Y_{\tilde{\lambda}}$. The expression for 
$C(\lambda,\tilde{\lambda})$ is given in~\cite{BDS20}. More important for our discussion
is the fact that then $\chi^\lambda(1^{N-2q-2},2^{q+1})=0$ since after
removing $q$ 2-blocks the 2-core $\tilde{\lambda}$ is reached.
The example below illustrates the process.
\begin{example}
Let $\lambda = (1,1,3,4)$. All possible sequences of peeling off a vertical or a
horizontal 2-block starting from $Y_\lambda$ and the intermediate
diagrams $Y_\mu$ are shown below. Finally all (3) possible paths indicated
by the arrows terminate at the Young diagram
$Y_{\tilde{\lambda}}=\raisebox{0.05in}{\ytableausetup{boxsize=0.5em}\ydiagram{2,1}}$ at the
lower right corner, which corresponds to a triangular partition.
\[\begin{matrix}
\raisebox{0.1in}{\ytableausetup{boxsize=0.6em}\ydiagram{4,3,1,1}} & \longrightarrow & 
\raisebox{0.1in}{\ytableausetup{boxsize=0.6em}\ydiagram{4,1,1,1}} & \longrightarrow & 
\raisebox{0.1in}{\ytableausetup{boxsize=0.6em}\ydiagram{2,1,1,1}} \\
\big\downarrow &  & \big\downarrow &  & \downarrow \\
\ytableausetup{boxsize=0.6em}\ydiagram{4,3} & \longrightarrow & 
\ytableausetup{boxsize=0.6em}\ydiagram{4,1} & \longrightarrow & 
\ytableausetup{boxsize=0.6em}\ydiagram{2,1}  
\end{matrix} \]
Notice that along each path 1 vertical and 2 horizontal 2-blocks      
are removed to arrive at $Y_{\tilde{\lambda}}$ before the process terminates.
Thus, $\chi^{(1,1,3,4)}(1^3,2^3) = -3\chi^{(1,2)}(1^3)$ but 
$\chi^{(1,1,3,4)}(1,2^4) =0$
since no further 2-blocks can be peeled off from the diagram $Y_{\tilde{\lambda}}$.
\end{example} 

It is also possible to determine the positive integer $q$ which is the maximum
number of 2-blocks removed from $Y_{\lambda}$, precisely as follows.
Let $m=(m_1,\ldots,m_n), \, 1 \leq m_1 < \dots < m_n$ be the degree vector 
of a partition $\lambda$ of
$N \in \mathbb{N}$. After removing a 2-block, one arrives at an (unordered)
sequence $(m_1,\ldots,m_{j}-2,\ldots,m_n)$ for some $1 \leq j \leq n$. After 
rearranging the resulting sequence in ascending order, yields the
sequence $m'$ which is the degree vector
of a new partition $\mu$ so long as the components of $m'$ are all {\it distinct}
and non-negative. Thus, the terminal 2-core $\tilde{\lambda}$ partition is reached
when the successive components of the corresponding degree vector $\tilde{m}$ 
differ by {\it at most} 2. Next consider the initial degree vector $m$
as an unordered set which is a union of $k$ odd and $l$ even positive integers
with $k+l=n$. After removing all possible 2-blocks from the odd set, the residual set
is $\{1,3,\ldots,2k-1\}$ while the residual even set becomes $\{0,2,\ldots,2l-2\}$.
Let us assume $k>l$. Then after rearranging the union of these residual sets in 
ascending order one recovers the degree vector 
$\tilde{m} = (0,1,\ldots,2l-2,2l-1,2l+1,\ldots 2k-1\}$ corresponding to the 2-core 
$$\tilde{\lambda} = (\underbrace{0,0, \ldots,0}_{2l}, 
\underbrace{1,3,\ldots,d}_{k-l})$$
of length $l(\tilde{\lambda}) = d=k-l$, 
using the relation $\tilde{\lambda}_j = \tilde{m}_j-j+1$. $\tilde{\lambda}$
is a triangular partition of size $|\tilde{\lambda}|=\sf{d(d+1)}{2}$.
The total number of 2-blocks removed is 
$$q = \sf{1}{2}(|\lambda|-|\tilde{\lambda}|)\,, \quad \text{since} \quad
|\lambda|-|\tilde{\lambda}|=(m_1+\cdots+m_n)-(1+\cdots+ 2k-1)-(2+\cdots+2l-2)$$
is an even number. If $k<l$ then a similar argument shows that $d=l-k-1$,
and if $k=l$ then $d=0$.
\begin{example}
For $\lambda=(1,1,3,4)$ as in Example 4.2, the degree vector $m=(1,2,5,7)$
where $m_2=2$ is even and the remaining 3 components are odd.
Then $\{2\} \cup \{1,5,7\} \to \{0\} \cup \{1,3,5\}$ after removing  
$(2-0)+(1-1)+(5-3)+(7-5) =6$ boxes from $Y_\lambda$. So $q=3$.
The degree vector for the core partition is $\tilde{m} = (0,1,3,5)$ so that
$\tilde{\lambda} = (0,0,1,2)$ and $d=k-l=3-1=2$.  
\end{example}
We conclude our brief excursion to partition theory with the following.
\begin{proposition}
Let $\lambda \neq (1,2,\ldots,n)$ be a non-triangular partition of
$N \in \mathbb{N}$. Suppose the degree vector of $\lambda$
has $k$ odd and $l$ even components, $k+l=n$.
Then there exists a largest positive integer $q \leq \lfloor\sf{N}{2}\rfloor$ 
such that the character coefficient $\chi^\lambda(1^{N-2q},2^q) \neq 0$ 
in the associated Schur function $W_\lambda(\theta)$. 
The integer $q=\sf{N-|\tilde{\lambda}|}{2}$
where the $n$-tuple $\tilde{\lambda} = (0,\dots,0,1,2,\ldots,d)$ is the 
2-core of $\lambda$ and $d=k-l$ if $k \geq l$ and $d=l-k-1$ if $l > k$. 
Hence, $N-2q =|\tilde{\lambda}|=\sf{d(d+1)}{2}$. Furthermore, the coefficient of 
$\theta_2^q$ in $W_\lambda(\theta)$ is $C_qW_{\tilde{\lambda}}(\theta)$ for
some constant $C_q$ that determines the value of $\chi^\lambda(1^{N-2q},2^q)$. 
\end{proposition} 
The last part of Proposition 4.1 will be explained later
in more details. The first part of Proposition 4.1 
implies the dominant balance $\tilde{\theta}_1^N \sim 
\tilde{\theta}_1^{N-2q}\tilde{\theta}_2^{q}$ in \eqref{Wzero} so that 
$\tilde{\theta}_1 \sim |t|^{1/2}, \, |t| \gg 1$. Collecting all the dominant terms
that are of the form $\tilde{\theta}_1^{N-2j}\tilde{\theta}_2^{j} \sim O(|t|^{N/2})$
from \eqref{Wasymp}, yields
\begin{equation}
 W_\lambda(\tilde{\theta}) \sim \sum_{j=0}^q 
\chi^\lambda(1^{N-2j},2^j)\frac{\tilde{\theta}_1^{N-2j}\tilde{\theta}_2^{j}}{(N-2j)!j!}  
= \Wr(p_{m_1}, \ldots,p_{m_n}) \,, 
\label{Wgen}
\end{equation}
where $p_r = p_r(\tilde{\theta}_1,\tilde{\theta}_2,0)$ in the above wronskian.
From \eqref{thetars}, $\tilde{\theta}_1=iz$ where $z=(r+h_1+is+\gamma_1)$ and
$\tilde{\theta}_2 = \frac{is}{2b}-3bt+h_2+\gamma_2$. Proceeding similarly as
in Section 4.1.1, one rescales $z=|t|^{1/2}\xi$, and uses \eqref{pn} to obtain
$$ p_r(\tilde{\theta}_1,\tilde{\theta}_2,0) \sim 
i^r|t|^{r/2}\big(H_r(\xi, \eta) + O(|t|^{-1/2})\big) \,.  $$
The heat polynomials $H_r(\xi,\eta)$ and their generating function 
are obtained directly from \eqref{pn} and \eqref{gen}
\[ H_r(\xi,\eta)= \sum_{j,k \geq 0}\frac{\xi^j\eta^k}{j!k!}, \quad j+2k=r\,,  
\qquad 
\exp(\alpha \xi+\alpha^2\eta) = \sum_{r=0}^\infty H_r(\xi, \eta)\alpha^r\,,
\quad \quad \eta =\left\{ \begin{matrix} 3b \,, & \quad t \geq 0 \vspace{0.05 in} \\    
-3b \,, & \quad t < 0 \end{matrix} \right. \,.
\]
These heat polynomials $H_r(\xi,\eta)$ were also introduced recently  
in order to investigate a special family of KPI lumps~\cite{CZ21}.
Substituting the asymptotic expression of $p_r$ into \eqref{Wgen} leads to 
\begin{subequations}
\begin{equation}
W_\lambda(r,s,t) \sim i^N|t|^{N/2}\big[H_\lambda(\xi,\eta)+O(|t|^{-1/2})\big]\,,
\label{heat-a}
\end{equation}
where $H_\lambda(\xi,\eta)$ is the wronskian of the heat polynomials given by
\begin{equation}
H_\lambda(\xi,\eta) = \Wr(H_{m_1},\ldots,H_{m_n})=\xi^{N-2q}\sum_{j=0}^q  
\frac{\chi^\lambda(1^{N-2j},2^j)}{(N-2j)!j!}\,\xi^{2(q-j)}\eta^j \,, 
\label{heat-b}
\end{equation}
and the positive integer $q$ is defined in Proposition 4.1. 
\label{heat}
\end{subequations}
The heat polynomials $H_r(\xi,\eta)$ above can be normalized to the Hermite polynomials
(see Remark 4.2(b)) so that the polynomials $H_\lambda$ are related
to the wronskian of Hermite polynomials after certain rescalings.
The wronskian of Hermite polynomials arise in the study of the Schr\"odinger equation:\,
$\psi''(z) = (u(z)-\lambda)\psi(z)$ where $u(z)$ is rational, growing as $z^2$
at infinity, and $\psi(z)$ is single-valued in the complex $z$-plane
for all $\lambda \in \mathbb{C}$~\cite{Ob99}. 

Treating $H_\lambda$ as a polynomial
in $\xi$ only, and $\eta=\pm3b$ as a parameter, it is clear that 
$H_\lambda(\xi)$ is of degree $N$ and has a root at $\xi=0$ with multiplicity $N-2q$.
It is conjectured that the non-zero roots of $H_\lambda(\xi)$ are simple~\cite{FHV12}
and will be assumed to be true in what follows. 
Furthermore, since $\xi^{2q-N}H_\lambda(\xi)$ in \eqref{heat-b} is a polynomial 
in $\xi^2$ of degree $q$ with real coefficients, the remaining non-zero roots 
admit the following symmetries:\, (a)\, if $\xi_j \in \mathbb{C}$, then 
the quartet $\{\pm\xi_j, \pm \bar{\xi}_j\}$, and (b)\, if $\xi_j \in \mathbb{R}$, 
then the pair $\pm \xi_j$ are all roots of $H_\lambda(\xi)$. Another symmetry
stems from the dependence of the roots on the parameter $\eta$. It is evident
from the definition of the heat polynomials $H_r(\xi,\eta)$ given above
\eqref{heat-a} that $H_r(\xi,-\eta) = i^rH_r(-i\xi, \eta)$ which implies
that $H_\lambda(\xi,-\eta)=i^N H_\lambda(-i\xi, \eta)$ and its roots satisfy
$\xi_j(\eta=-3b) = i\xi_j(\eta=3b)$.  

When $\xi_j \neq 0$, by putting $W_\lambda(r,s,t)=0$ in \eqref{heat-a}
and solving asymptotically for 
$$z(t)=|t|^{1/2}(z_0+\epsilon z_1+\epsilon^2z_2+\cdots)\,, \qquad 
\epsilon=|t|^{-1/2}$$
like in the triangular $N$-lump case, one
obtains the $2q$ approximate peak locations for the general $N$-lump solutions 
when $|t| \gg 1$ in terms of the roots of $H_\lambda(\xi)$.
Denoting the peak locations by $Z^{\pm}_j=r^\pm_j+is^\pm_j$ corresponding to 
positive and negative $t$ respectively, one obtains for $|t| \gg 1$,
\begin{equation}
Z^{\pm}_j(t) = (r_0-\sf{n}{2b}+is_0)+
|t|^{1/2}\big[\xi_j(\eta=\pm 3b)+O(|t|^{-1/2})\big] \,, \quad j=1,2,\ldots,2q \,,
\label{nonzero}
\end{equation}
where as usual, $h_1=\sf{n}{2b}$ from \eqref{h} and $-\gamma_1=r_0+is_0$.
Furthermore, $(r^+_j(t), s^+_j(t)) \sim (-s^-_j(t), r^-_j(t))$ as $|t| \to \infty$
due to the symmetry of the roots:\, $\xi_j(\eta=-3b) = i\xi_j(\eta=3b)$.
The surface wave pattern due to the peak locations 
in the $rs$-plane are markedly different from those of the triangular $N$-lumps
due to the quartet or doublet root patterns 
of $H_\lambda(\xi)$ as shown below in Figure~\ref{6lump}. Also, unlike 
in \eqref{trianglezero}, the $O(1)$ term in \eqref{nonzero}
will change if $z_1 \neq 0$ in the asymptotic expansion of $z(t)$ above.
Further details of this asymptotics will not be pursued in this paper.
\begin{example}
Consider the partition $\lambda=(3,3)$ with $m=(3,4)$ where $N=6$.
The Schur function is given by
\[W_\lambda = \begin{vmatrix} p_3 & p_4 \\ p_2 & p_3 \end{vmatrix} =
\frac{\tilde{\theta}_1^6}{144} + \frac{\tilde{\theta}_1^4\tilde{\theta}_2}{24}
+ \frac{\tilde{\theta}_1^2\tilde{\theta}_2^2}{4} - \frac{\tilde{\theta}_2^3}{2}
-\frac{\tilde{\theta}_1^3\tilde{\theta}_3}{6}+\tilde{\theta}_1\tilde{\theta}_2\tilde{\theta}_3  
-(\frac{\tilde{\theta}_1^2}{2} - \tilde{\theta}_2)\tilde{\theta}_4 \,, \]
after using \eqref{pn}. The first four terms above are dominant consistent with the
scaling $\tilde{\theta}_1 = i|t|^{1/2}\xi$, and leads to
\[W_\lambda \sim i^6|t|^3\big[H_\lambda(\xi,\eta)+O(|t|^{-1/2}\big]\,, \quad
H_\lambda = \Wr(H_3,H_4) = 
\frac{1}{144}\big(\xi^6+6\xi^4\eta+36\xi^2\eta^2-72\eta^3\big)\,, \]  
as in \eqref{heat-a}. Note that $\xi=0$ is not a root of $H_\lambda$
as $q=N/2=3$ in \eqref{heat-b}. Since $\eta=3b$ for positive $t$, setting 
$w=\sf{\xi^2}{3b}$ it can be easily verified that the resulting cubic polynomial 
$w^3+6w^2+36w-72$ has a single real root that is positive and a pair of complex 
conjugate roots. Consequently, 4 of the 6 distinct zeros of $H_\lambda(\xi)$ 
arise in complex conjugate pairs, and remaining 2 are real which corresponds to
the positive real root of the cubic polynomial in $w$. The 
approximate peak locations given by \eqref{nonzero} in this case are shown
in the right panel of Figure~\ref{6lump}. By using the symmetry 
$\xi_j(\eta=-3b)=i\xi_j(\eta=3b)$ of the zeros of $H_\lambda(\xi,\eta)$
one obtains the leading order approximate peak locations for $t<0$. These
are shown in the left panel of Figure~\ref{6lump} confirming that  
$(r^+_j(t), s^+_j(t)) \sim (-s^-_j(t), r^-_j(t))$ for $|t| \gg 1$.      
\end{example}
\begin{figure}[t!]
\begin{center}
\includegraphics[scale=0.42]{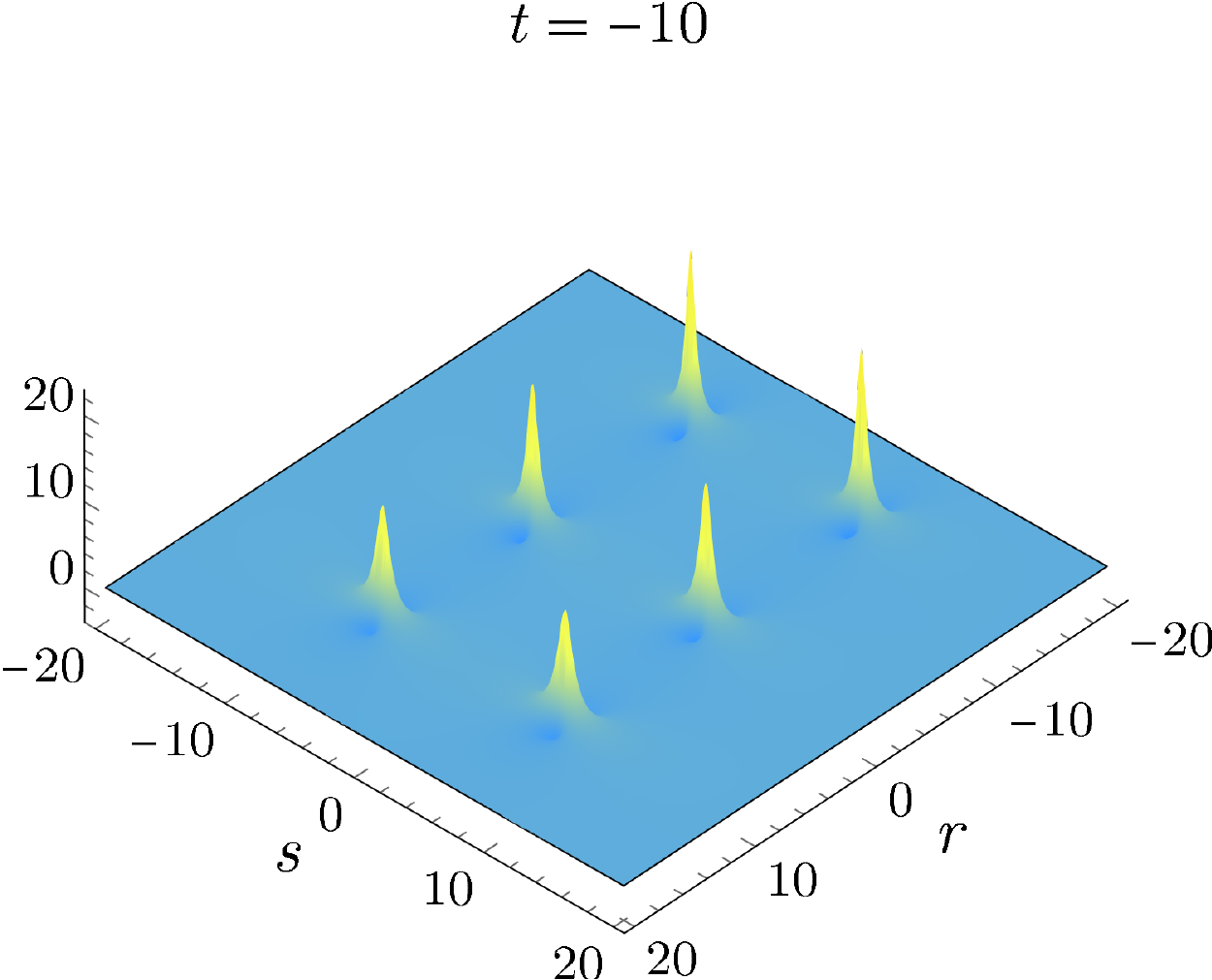} \qquad
\includegraphics[scale=0.42]{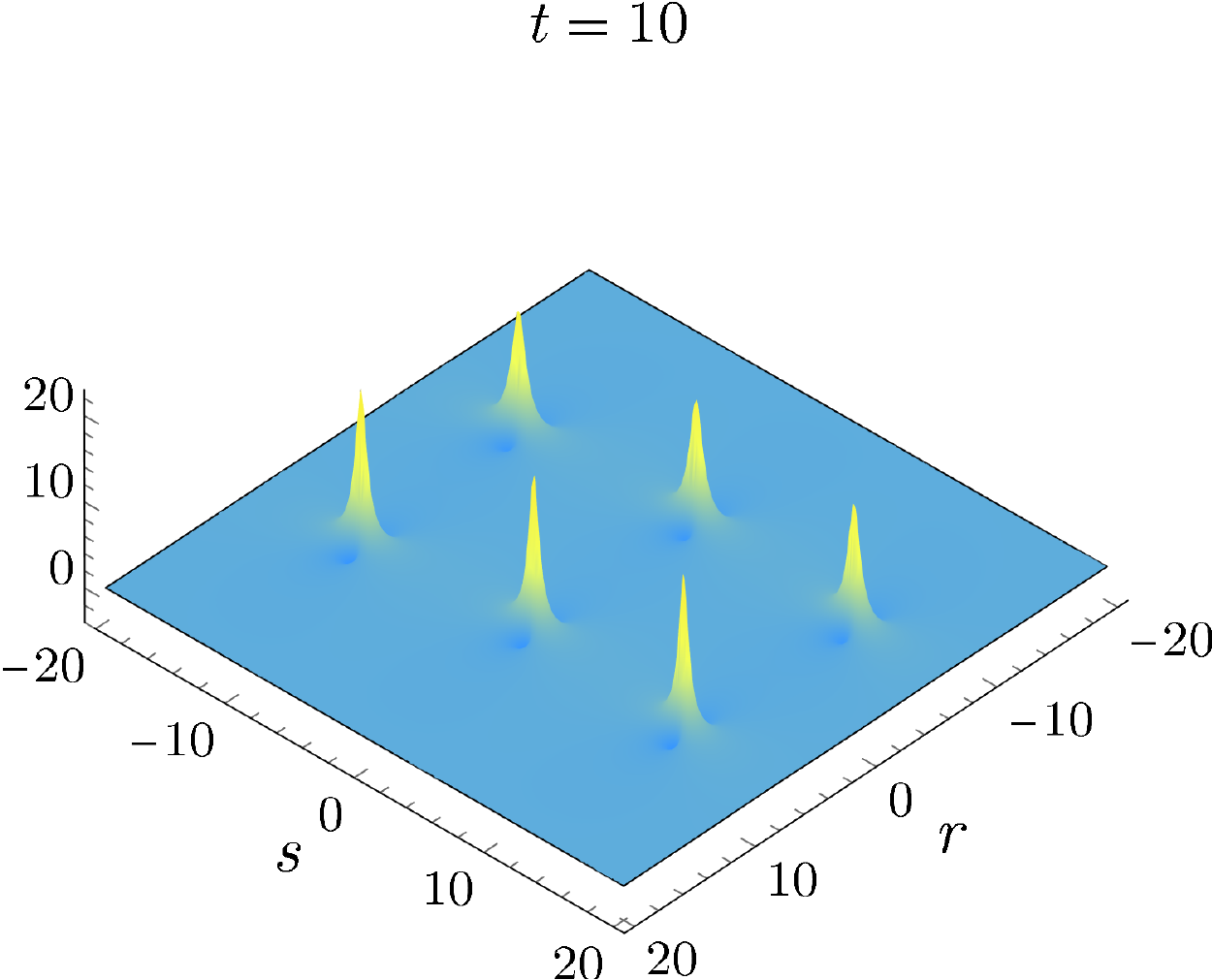} 
\end{center}
\vspace{-0.2in}
\caption{The $6$-lump solutions of Example 4.4. Solutions
corresponding to $t=-10$ (left panel) and $t=10$ (right panel).
KP Parameters:\, $a=0, b=1, \, \gamma_j=0,\, j\geq 1$.}
\label{6lump}
\end{figure}
Yet another interesting result concerns the relation between
the asymptotic peak locations of a pair of class (I) $N$-lump solutions 
in Section 3.2 corresponding to partition $\lambda$ and its conjugate $\lambda'$.
This relationship occurs when the approximate 
peak locations are given by \eqref{nonzero} corresponding to the non-zero 
roots of $H_\lambda$ and $H_{\lambda'}$. Applying the first involution relation 
in \eqref{dualschur} to \eqref{Wgen} yields 
$W_{\lambda'}(\tilde{\theta}_1,\tilde{\theta}_2) = 
W_\lambda(\tilde{\theta}_1,-\tilde{\theta}_2)$.
From $\tilde{\theta}_2 = \frac{is}{2b}-3bt+h_2+\gamma_2$ and the fact
that $s \sim O(|t|^{1/2})$ (since $\theta_1 \sim O(|t|^{1/2}))$ it follows
that $\tilde{\theta}_2(-t) \sim -\tilde{\theta}_2(t)$ for $|t| \gg 1$.
Hence, one obtains the leading order asymptotic relation 
$W_{\lambda'}(r,s,t) \sim W_\lambda(r,s,-t)$ in \eqref{heat-a}. 
Consequently, the asymptotic peak locations given by \eqref{nonzero}
for the $N$-lump solutions corresponding to partitions $\lambda$
and $\lambda'$ satisfy
\[Z^{\lambda'}_j(t) \sim Z^{\lambda}_j(-t)\,, \quad 
j=1,\ldots,2q\,, \qquad |t| \gg 1\,.  \]
In this context, one should note that the 2-core partition $\tilde{\lambda}$
is the {\it same} for both $\lambda$ and its conjugate $\lambda'$ since removing
a horizontal 2-block from $\lambda$ is equivalent to removing a vertical 2-block 
from $\lambda'$ and vice-versa. Hence it follows from Proposition 4.1 that the
positive integer $q$ in \eqref{Wgen} is the same for both Schur functions
$W_\lambda$ and $W_{\lambda'}$. Note that the involution in
\eqref{dualschur} can be applied directly to obtain the relation  
$H_{\lambda'}(\xi,\eta) = H_\lambda(\xi, -\eta)$ reported in~\cite{FHV12,BDS20}.     

When $\xi=0$ is a root of $H_\lambda(\xi)$ with multiplicity $N-2q$, 
the $N-2q$ peak locations are
not separated to leading order at the $O(|t|^{1/2})$ time scale. 
It is then necessary to examine the terms that were subdominant in
\eqref{Wasymp} when $\tilde{\theta}_1 \sim |t|^{1/2}$ and were
neglected to obtain \eqref{Wgen}. Looking back in \eqref{Wasymp} such terms
that are significant for $|t| \gg 1$ are of the form 
$\tilde{\theta}_1^k\tilde{\theta}_2^q\tilde{\theta}_3^{l},\, k+3l=N-2q$
where the exponent of $\tilde{\theta}_2$ is fixed at its largest value $q$.
In other words, one needs to collect the coefficient of $\tilde{\theta}_2^q$ 
from the Schur function $W_\lambda(\tilde{\theta})$ in \eqref{Wasymp}.
For this purpose we split $\tilde{\theta} = \hat{\theta}+\delta$ where
$\hat{\theta} = (\tilde{\theta}_1, 0, \tilde{\theta}_3, \ldots) 
+ (0, \tilde{\theta}_2,0, \ldots)$ and apply \eqref{shift} to 
$W_\lambda(\tilde{\theta})$, which leads to the following 
\[ W_\lambda(\hat{\theta}+\delta) = \sum_{\mu}W_\mu(\hat{\theta})W_{\lambda/\mu}(\delta)\,,
\qquad \mu \subseteq \lambda \,.\]
It can be shown from the definition of $W_{\lambda/\mu}$ in \eqref{ss} that
since $\delta$ depends only on $\tilde{\theta}_2$, $W_{\lambda/\mu}(\delta)$
is nonvanishing only when $\mu$ is obtained from $\lambda$ by removing
successive 2-blocks, i.e., $|\lambda|-|\mu| = 2j, \, j \geq 0$ so that
$W_{\lambda/\mu}(\delta) = C_j\tilde{\theta}_2^j$ for some constant $C_j$. This is
precisely the process of obtaining the 2-core partition of $\lambda$ as
explained at the beginning of this subsection (see also Example 4.2).
Therefore, the above sum starting from $\mu = \emptyset$ must terminate when
$\mu = \tilde{\lambda}$ which is the 2-core of $\lambda$ and 
$W_{\lambda/\tilde{\lambda}}(\delta) = C_q\tilde{\theta}_2^q$ as claimed at the
end of Proposition 4.1. Since $\tilde{\lambda}$ is a triangular partition the
analysis presented in Section 4.1.1 applies to this case. Using \eqref{ss} and
the degree vector $\tilde{m} = (0,1,\ldots,2l-2,2l-1,2l+1,\ldots 2k-1\}$ for
the partition $\tilde{\lambda}$ in Proposition 4.1, one obtains
\[W_{\tilde{\lambda}}(\hat{\theta}) = 
\Wr(p_0,p_1,\ldots,p_{2l-1},p_{2l+1},\ldots,p_{2k-1})
=\Wr(p_1,p_3,\ldots,p_{2d-1}) \,, \qquad d=k-l\,, \quad k \geq l \,.      
\]
A similar expression can be obtained when $k<l$. Also 
when $k=l, \, d=0$, and $\tilde{\lambda} = \emptyset$. In this case,
$N$ is even and $q=N/2$ so that $\xi=0$ is not a root of $H_\lambda(\xi)$
(cf. Example 4.4).

Isolating the coefficient of $\tilde{\theta}_2^q$ in 
$W_\lambda(\tilde{\theta})$ leads to the dominant behavior
\begin{equation}
W_\lambda(\tilde{\theta}) \sim W_{\lambda/\tilde{\lambda}}(\delta)
W_{\tilde{\lambda}}(\hat{\theta}) = C_q\tilde{\theta}_2^q\Wr(p_1,p_3,\dots,p_{2d-1})\,,
\label{Wd}
\end{equation}
and when $W_\lambda(\tilde{\theta})=0$ the corresponding dominant balance is 
$\tilde{\theta}_1 \sim O(|t|^{1/3})$ as in Section 4.1.1. In this case,
when $|t| \gg 1$, $W_\lambda(\tilde{\theta}) \sim O(|t|^{(N+q)/3})$ while the 
terms in \eqref{Wgen}, $\tilde{\theta}_1^{N-2j}\tilde{\theta}_2^j \sim O(|t|^{(N+j)/3})$ 
for $j<q$, are now sub-dominant. Then the remaining $d=N-2q$ approximate peak
locations for the $N$-lump solution corresponding to partition $\lambda$
as well as its conjugate $\lambda'$ are obtained by zeros of the 
Yablonskii-Voro'bev polynomial $Q_d(\xi)$ as in \eqref{trianglezero}
for $j=1,\ldots,d$. Also note that the coefficient $C_q$ in \eqref{Wd}
is recovered from $W_{\lambda/\tilde{\lambda}}(\delta)$. Indeed from \eqref{ss},
\begin{equation}
C_q\tilde{\theta}_2^q = W_{\lambda/\tilde{\lambda}}(\delta) 
=\det(p_{m_j-\tilde{m}_i})(\delta)\,, \qquad  
\delta=(0,\tilde{\theta}_2, 0, \ldots)\,,
\label{Cq}
\end{equation}
where $m, \tilde{m}$ denote the degree vectors of the 
partitions $\lambda, \tilde{\lambda}$, respectively.

In summary, it is established that there are indeed $N$ distinct peaks for the general 
$N$-lump solutions of KPI equation \eqref{kp}. The wave pattern in the
$rs$-plane exhibit two distinct time scales such that the dynamics of 
$q$ out of the $N$ peaks occurs at a faster $O(|t|^{1/2})$ time scale,
while the remaining $N-2q$ peaks evolve at a slower $O(|t|^{1/3})$ time scale.    
Consequently, the $q$ peaks in the far field region of the wave pattern
form a rectangular shape or lie symmetrically along the $r$- or $s$-axis,
whereas the other $N-2q$ peaks form a triangular pattern
in the near-field region of the $rs$-plane.
\begin{example}
Let $\lambda = (1,4)$ corresponding to $m=(1,5)$ and $N=5$. Then the Schur function
from \eqref{Wshift} is
\[W_\lambda(\tilde{\theta}) = \Wr(p_1,p_5) = \frac{\tilde{\theta}_1^5}{30} +
\frac{\tilde{\theta}_1^3\tilde{\theta}_2}{3}+\frac{\tilde{\theta}_1^2\tilde{\theta}_3}{2}
- \tilde{\theta}_2\tilde{\theta}_3-\tilde{\theta}_5 \,.\]
The largest exponent of $\tilde{\theta}_2$ in $W_\lambda(\tilde{\theta})$ is
$q=1$. So the first dominant balance 
$\tilde{\theta}_1^5 \sim \tilde{\theta}_1^3\tilde{\theta}_2$ to solve \eqref{Wzero}
gives the scaling $\tilde{\theta}_1 \sim |t|^{1/2}$, Collecting the dominant terms
which are $O(|t|^{5/2})$, one obtains as in \eqref{Wgen}
$W_\lambda \sim \sf{\tilde{\theta}_1^3}{30}(\tilde{\theta}_1^2+\tilde{\theta}_2$),
the rest of the terms in $W_\lambda$ are subdominant being $O(|t|^2)$.
Next putting $\tilde{\theta}_1=iz, \, z=|t|^{1/2}\xi$, and after
using \eqref{thetars} for $\tilde{\theta}_1, \tilde{\theta}_2$, the dominant terms give
the following asymptotics accordingly as \eqref{heat}
\[ W_{\lambda} \sim i^5|t|^{5/2}\big[H_\lambda(\xi,\eta)+O(|t|^{-1/2})\big]\,, \quad \qquad
H_\lambda(\xi,\eta) = \Wr(H_1, H_5) = \frac{\xi^3}{30}\big(\xi^2+10\eta\big)\,. 
\]
\begin{figure}[t!]
\begin{center}
\includegraphics[scale=0.42]{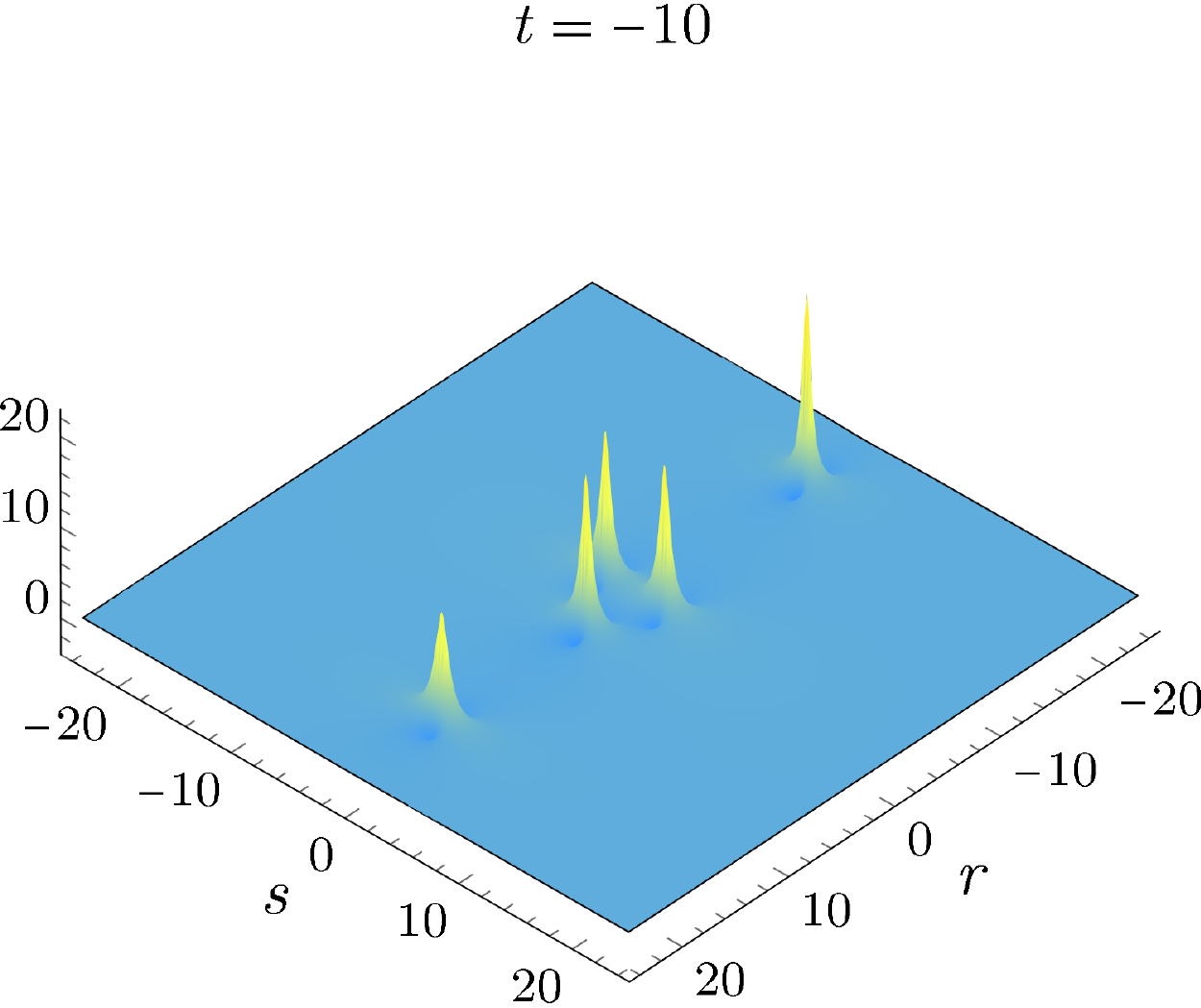} \quad
\includegraphics[scale=0.42]{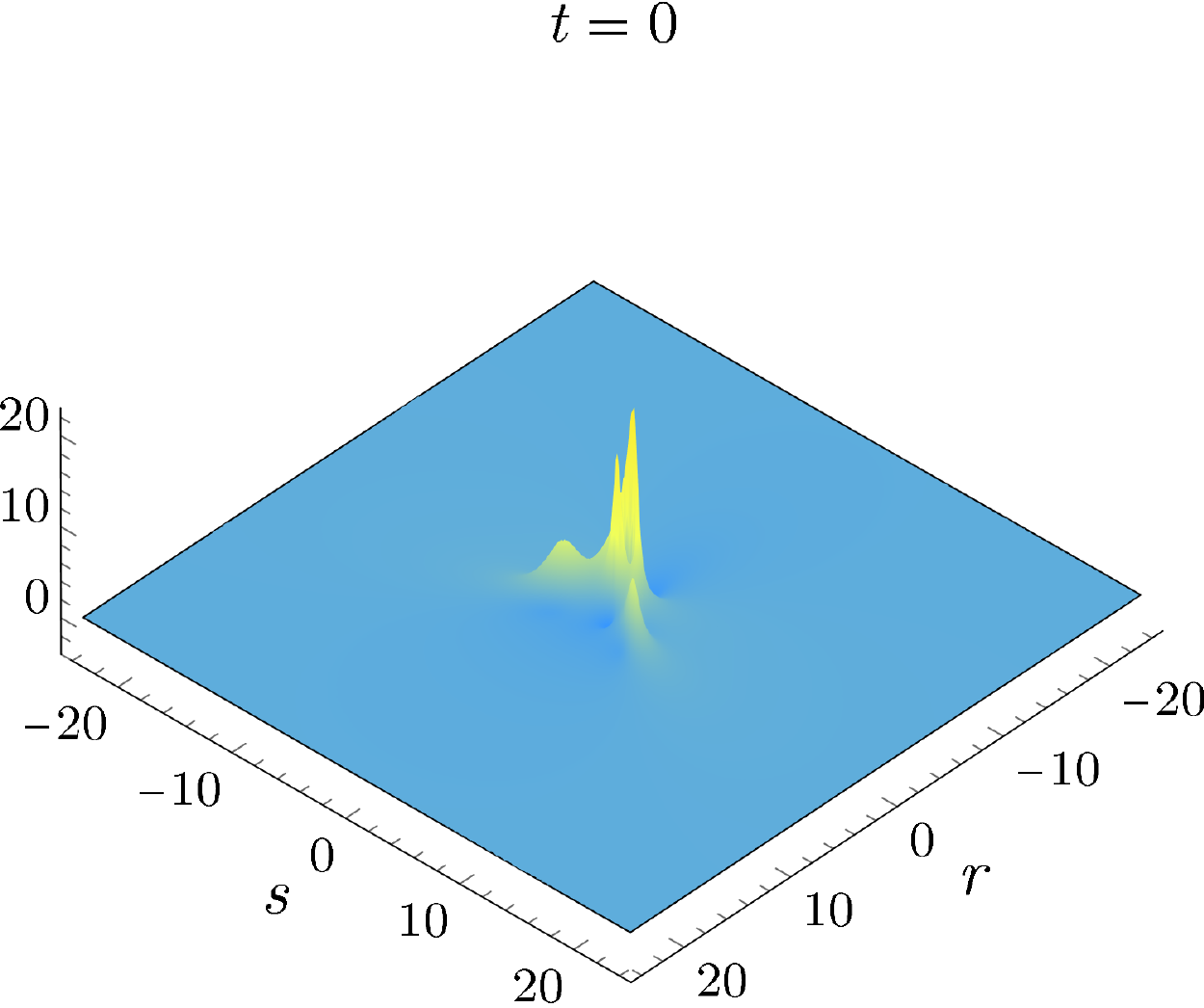} \quad
\includegraphics[scale=0.42]{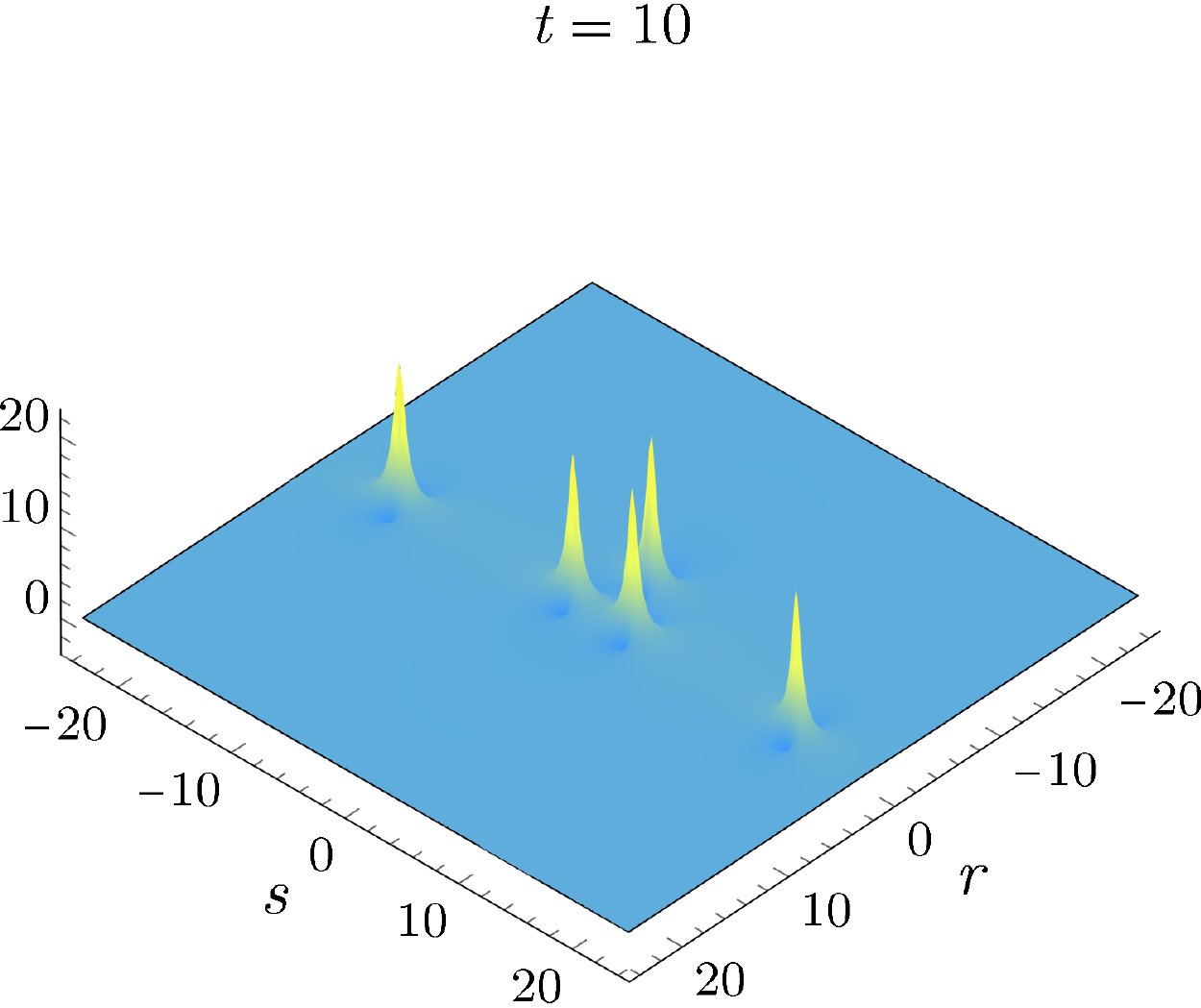} \\
\includegraphics[scale=0.42]{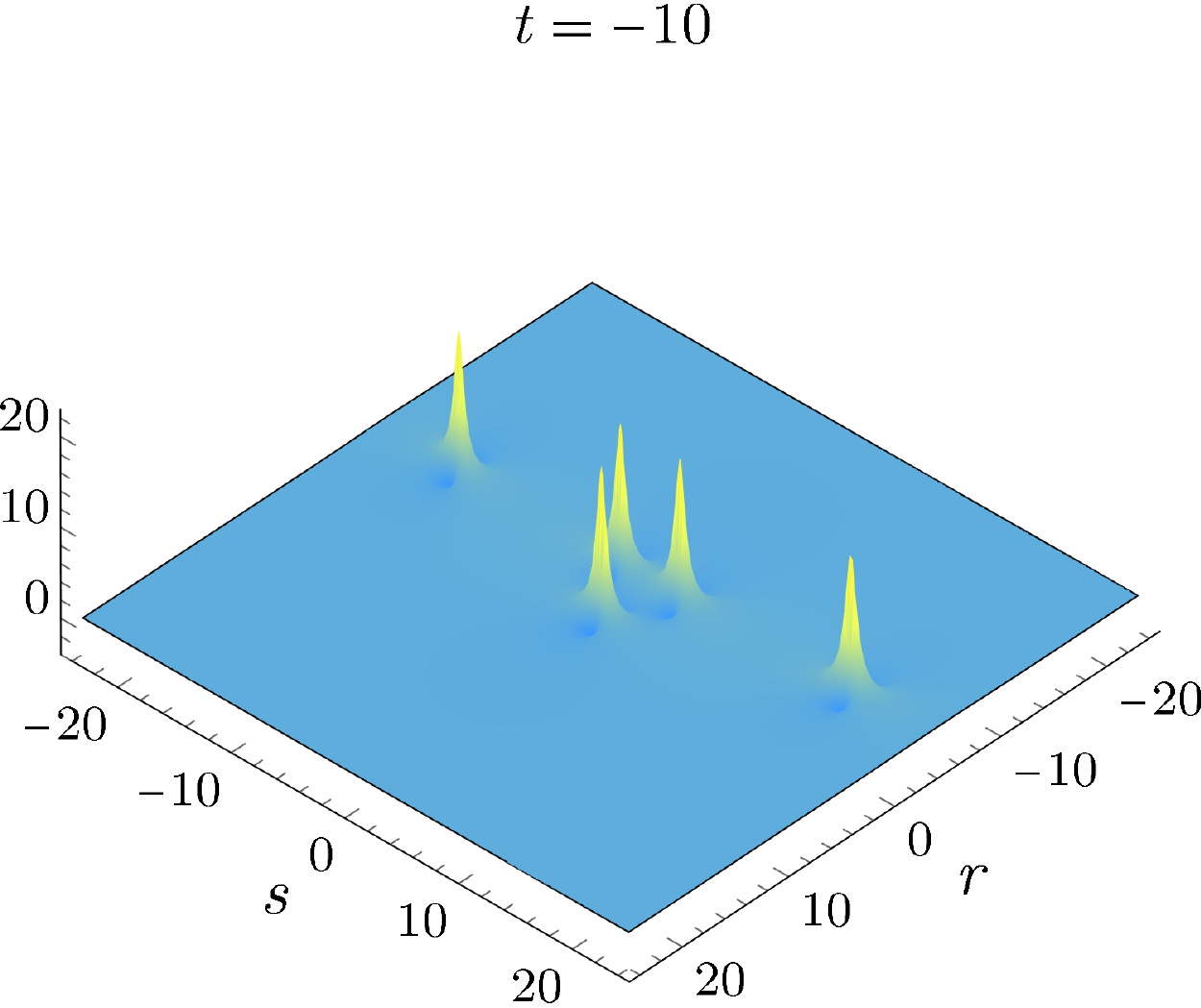} \quad
\includegraphics[scale=0.42]{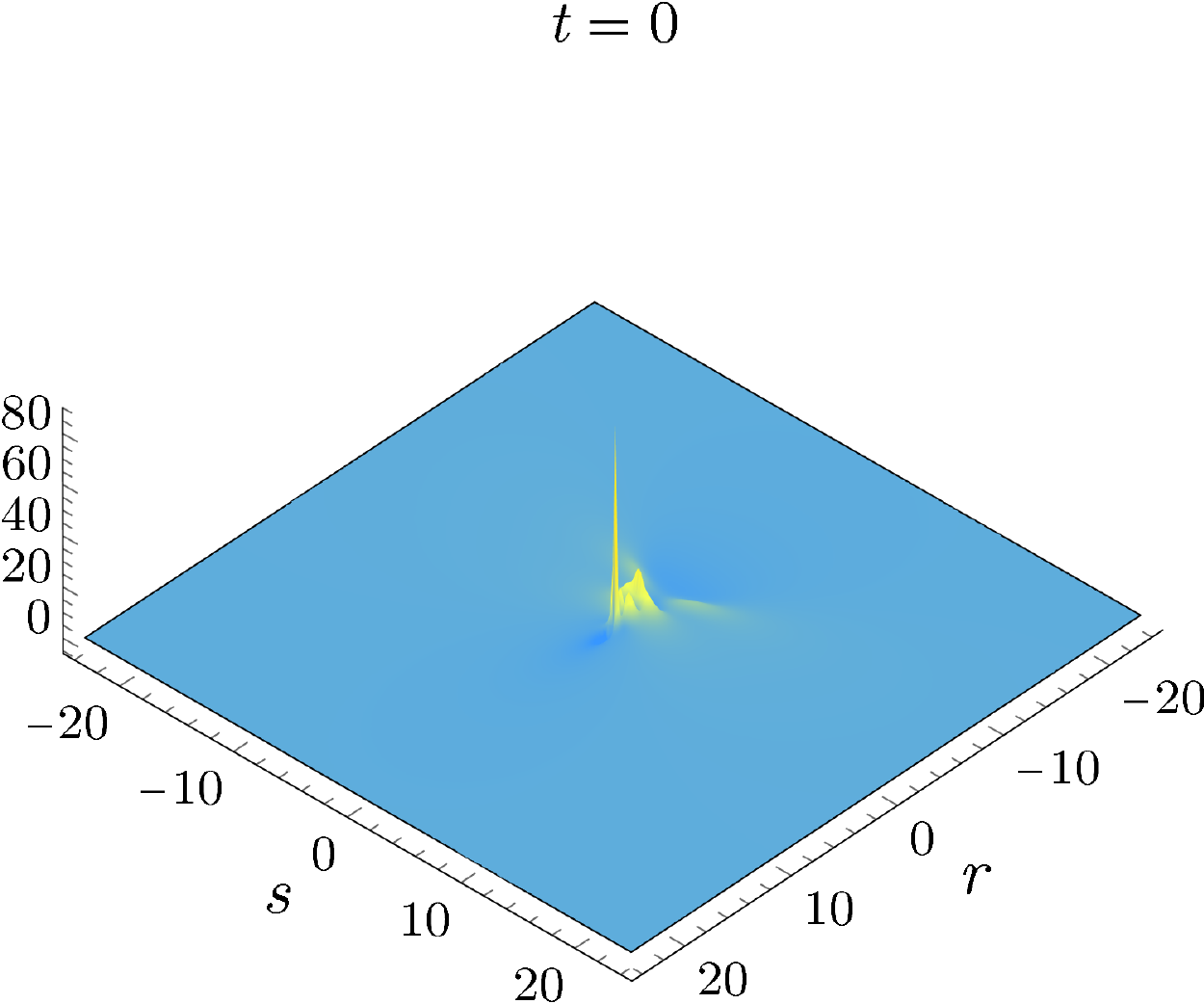} \quad
\includegraphics[scale=0.42]{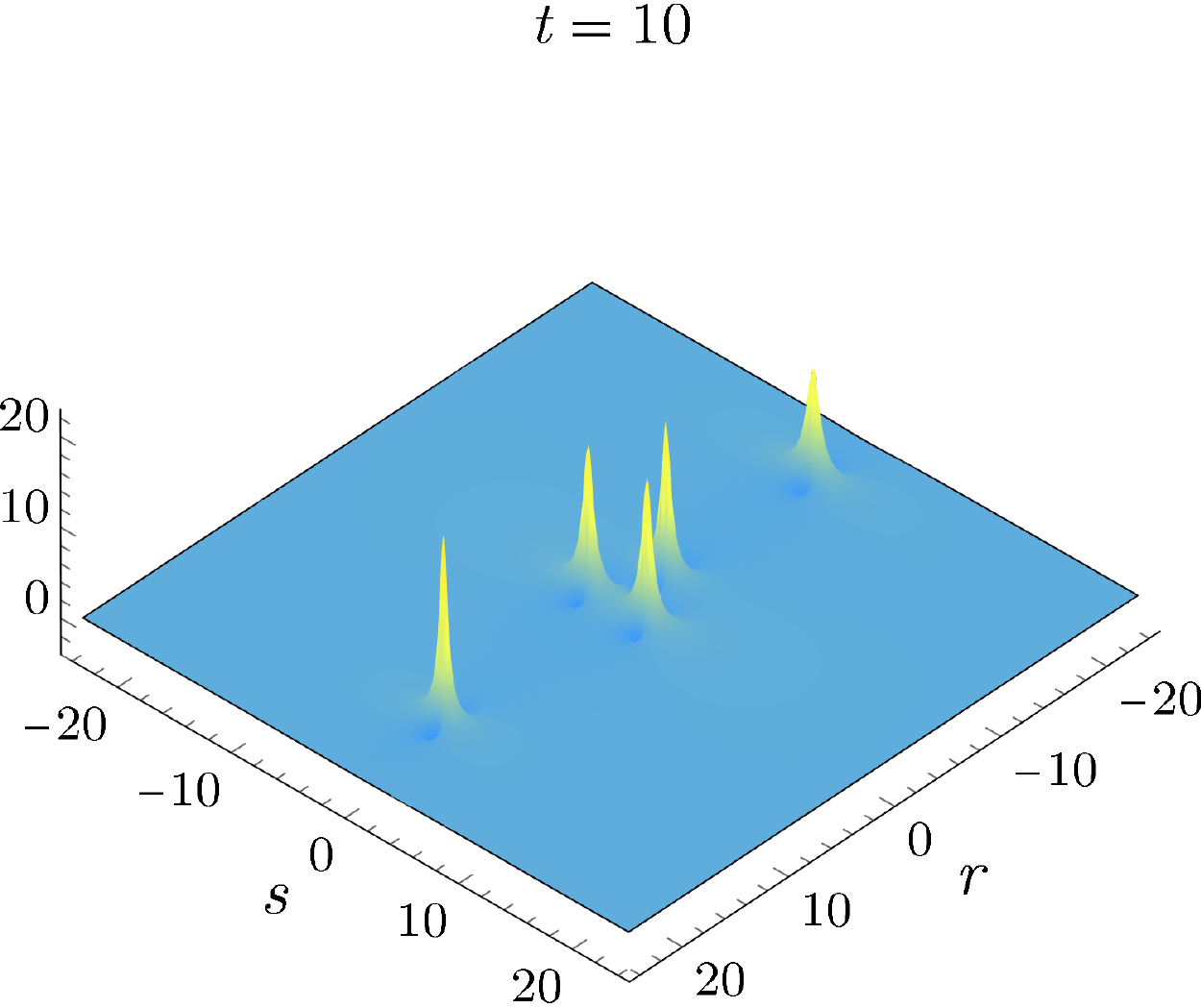} 
\end{center}
\vspace{-0.2in}
\caption{The $5$-lump solutions of Example 4.4. Solutions
corresponding to $\lambda=(1,4)$ (top panel) and the conjugate
$\lambda'=(1,1,1,2)$ (bottom panel). Note that the two peaks in the
outer region of the plane satisfy the symmetry $Z^{\lambda'}_j(t)
\sim Z^{\lambda}_j(-t)$ at $t = \pm 10$. A triangular configuration
of three lumps are formed in the inner region which is the same
for both $\lambda$ and $\lambda'$. KP Parameters: same as in Figure~\ref{6lump}.}
\label{5lump}
\end{figure}
The approximate peak locations for $|t| \gg 1$ are obtained by 
solving $W_\lambda=0$ for $z(t) = |t|^{1/2}(z_0+\epsilon z_1+ \cdots)$
where $\epsilon = |t|^{-1/2}$.
Hence, to leading order the distinct peak locations $z_0$
in this example correspond
to the two non-zero roots $\xi = \pm\sqrt{-10\eta}$ of $H_\lambda(\xi)$ above.
Using $\eta=\pm 3b$ and the shift from \eqref{h}
$h_1=\sf{n}{2b}=\sf{1}{b}$ (since $n=2$ here), the approximate peak locations
for $|t| \gg 1$ are given by 
\[ r^\pm(t)+i s^\pm(t) = (r_0-\sf{1}{b}+is_0) + 
|t|^{1/2}\big(\sqrt{\mp 30b}+O(|t|^{-1/2})\big)\,, \qquad -\gamma_1=r_0+is_0 \,,\]
as in \eqref{nonzero}. When $t<0$, $\eta=-3b$ corresponds to real roots
$\xi = \pm\sqrt{30b}$. Hence for $t<0$ and $|t| \gg 1$, the two peaks lie 
on the $r$-axis symmetrically about the origin. For $t \gg 1$, the
two peaks are instead located on the $s$-axis symmetrically about the origin
consistent with the symmetry $(r^+_j(t), s^+_j(t)) \sim (-s^-_j(t), r^-_j(t))$ 
as $|t| \to \infty$. These features are illustrated in Figure~\ref{5lumppeaks}. 
From the asymptotic expansion of $z(t)$, one can also calculate 
$z_1$ which has a further $O(1)$ contribution to the peak locations.
For example, when $t \gg 1$, $z_1=\sf{11}{5b}$ which shifts the two peak locations
slightly to the {\it right} of the $s$-axis. This effect is evident in the top 
right panel of Figure~\ref{5lumppeaks}.

The remaining approximate peak locations are obtained from a second dominant
balance $\tilde{\theta}_1 \sim O(|t|^{1/3})$ in $W_\lambda(\tilde{\theta})$ above, 
arising from the terms linear in $\tilde{\theta}_2$. That is,
$W_\lambda(\tilde{\theta}) \sim 
\tilde{\theta}_2(\tilde{\theta}_1^3/3-\tilde{\theta}_3) \sim O(|t|^2)$,
and the remaining terms are then $O(|t|^{5/3})$, hence subdominant.
Solving for $\tilde{\theta}_1^3/3-\tilde{\theta}_3 \sim 0$, and using \eqref{h}
with $n=2$ leads to the 
approximate peak locations that are exactly the same as the ones
for the $3$-lump solution given above Figure 4 of Section 2.3.3.
The origin of the second dominant balance is from the 2-core 
$\tilde{\lambda} = (1,2)$ of the partition $\lambda= (1,4)$ obtained by
removing a single horizontal 2-block from its Young diagram i.e., 
$Y_\lambda = \raisebox{0.05in}{\ytableausetup{boxsize=0.4em}\ydiagram{4,1}}   
\longrightarrow \raisebox{0.05in}{\ytableausetup{boxsize=0.4em}\ydiagram{2,1}}  
= Y_{\tilde{\lambda}}$. Since $m_1, m_2$ are both odd, $d=k-l=2$ from Proposition 4.1.
Then $|\tilde{\lambda}| = \sf{d(d+1)}{2}=3$. Again from Proposition 4.1, 
$N-2q = |\tilde{\lambda}|=3$ implies that $q=1$ which has already been found as
the largest exponent of $\tilde{\theta}_2$ in $W_\lambda(\tilde{\theta})$ above. 
Using \eqref{Wd} one can directly obtain the dominant term
$W_\lambda(\tilde{\theta}) \sim C_1\tilde{\theta}_2\Wr(p_1,p_3) =
C_1\tilde{\theta}_2(\tilde{\theta}_1^3/3-\tilde{\theta}_3)$. The coefficient $C_1$ 
is obtained by applying \eqref{Cq}. Here $m=(1,5)$ and $\tilde{m}=(1,3)$.
Hence, $C_1\tilde{\theta}_2 = \big|\begin{smallmatrix} p_0 & p_4 \\ 0 & p_2 
\end{smallmatrix}\big| = p_2(0,\tilde{\theta}_2) = \tilde{\theta}_2 \Rightarrow C_1=1$.
The full wave pattern as well as the peak locations in the $rs$-plane
are illustrated in the top panels of Figures~\ref{5lump} and ~\ref{5lumppeaks}.

Let us also briefly consider the the $5$-lump solution corresponding to
the conjugate partition $\lambda' = (1,1,1,2)$ with degree vector $m'=(1,2,3,5)$.
By direct computation one verifies that 
\[W_{\lambda'}(\tilde{\theta}) = \Wr(p_1,p_2,p_3,p_5) = \frac{\tilde{\theta}_1^5}{30} 
-\frac{\tilde{\theta}_1^3\tilde{\theta}_2}{3}+\frac{\tilde{\theta}_1^2\tilde{\theta}_3}{2}
+ \tilde{\theta}_2\tilde{\theta}_3-\tilde{\theta}_5 \,,\]
which also follows from the first involution relation in \eqref{dualschur}
by simply switching $\tilde{\theta}_2 \to -\tilde{\theta}_2$ in $W_\lambda(\tilde{\theta})$
above. Therefore, again $q=1$ and the dominant balances are same as before.
However, when $\tilde{\theta}_1 \sim O(|t|^{1/2})$,
\[ W_{\lambda'} \sim i^5|t|^{5/2}\big[H_{\lambda'}(\xi,\eta)+O(|t|^{-1/2})\big]\,, 
\quad \qquad H_{\lambda'}(\xi,\eta) = H_{\lambda}(\xi,-\eta)=i^NH_{\lambda}(i\xi,\eta) 
\,,\] 
so that the real and pure imaginary roots are switched from the previous case.
This leads to the symmetry $Z_j^{\lambda'}(t) \sim Z_j^\lambda(-t), \,\, j=1,2$
when $|t| \gg 1$ for the two peaks located in far field region of the $rs$-plane.
However, the approximate configuration of the remaining three peaks that form
the inner core in the near field region remains the same as in the previous case.
Note that in this case $n=4$ so that the shift $h_1=\sf{2}{b}$ is larger than
the previous case. This causes the peak locations to shift further toward the left
of the $s$-axis as clearly evident from e.g., the bottom left frame in 
Figure~\ref{5lumppeaks}. The full wave pattern and the peak locations are shown in 
the bottom panels of Figures~\ref{5lump} and ~\ref{5lumppeaks}.
\end{example}   
\begin{figure}[h!]
\begin{center}
\includegraphics[scale=0.4]{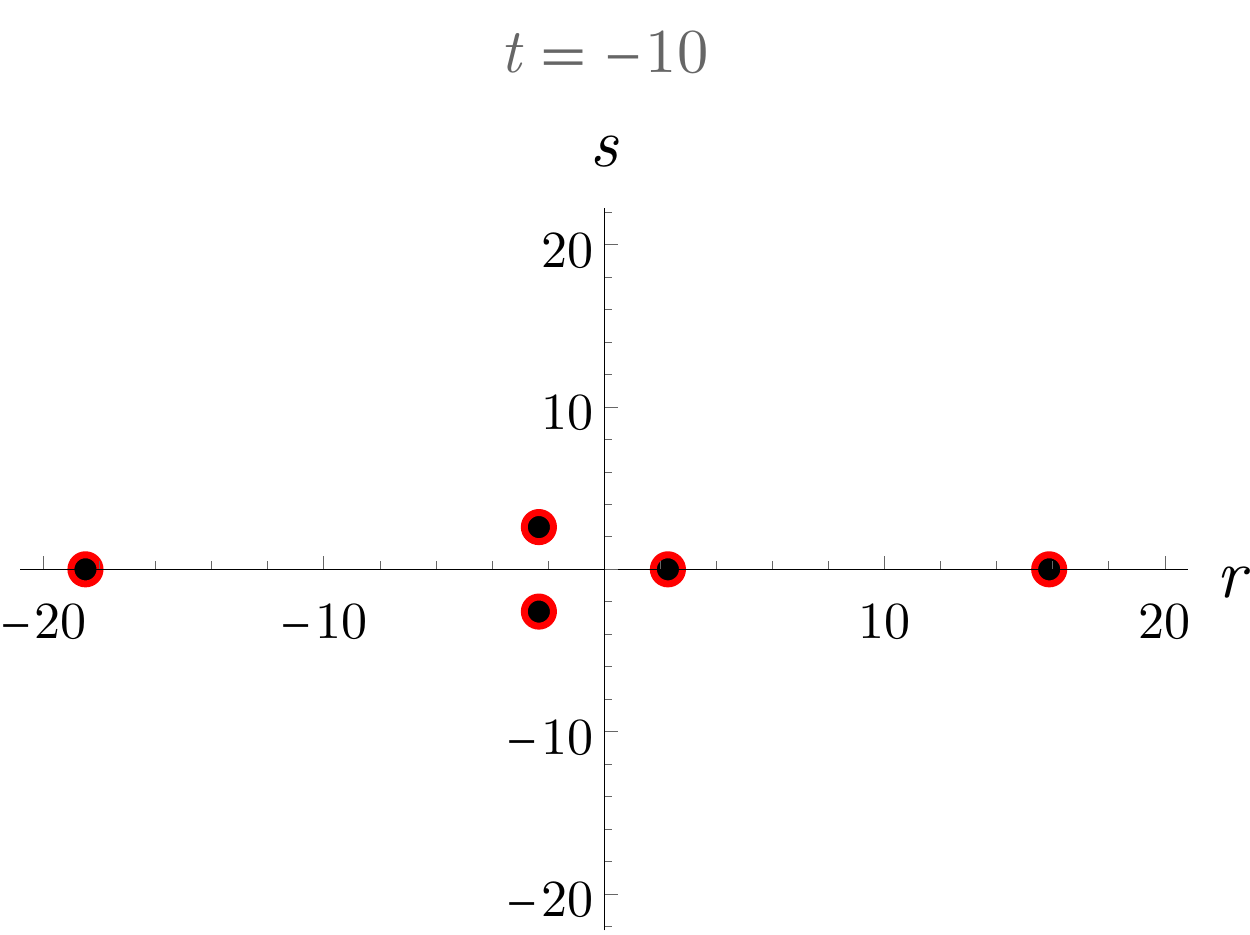} \qquad \qquad
\includegraphics[scale=0.4]{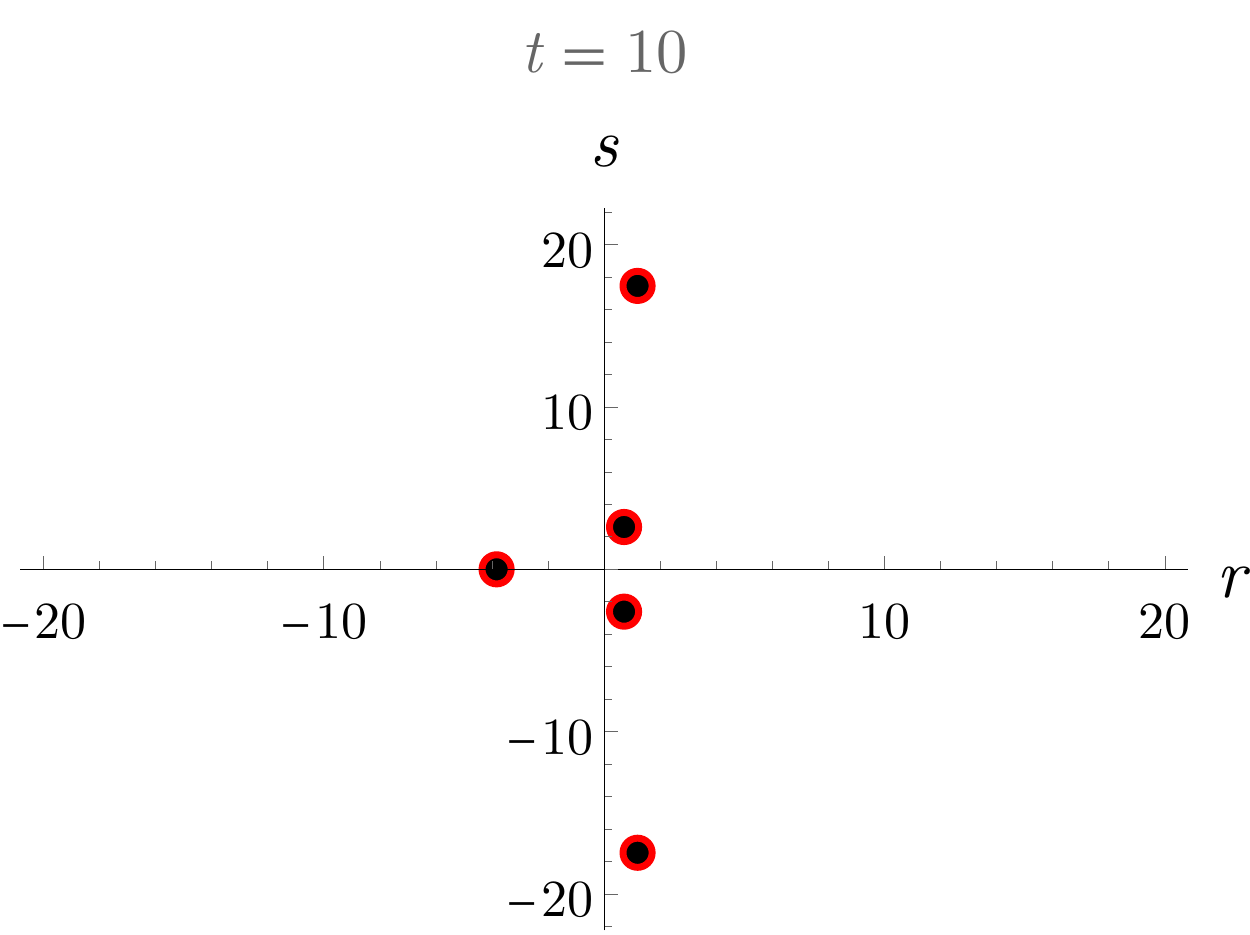} \\
\includegraphics[scale=0.4]{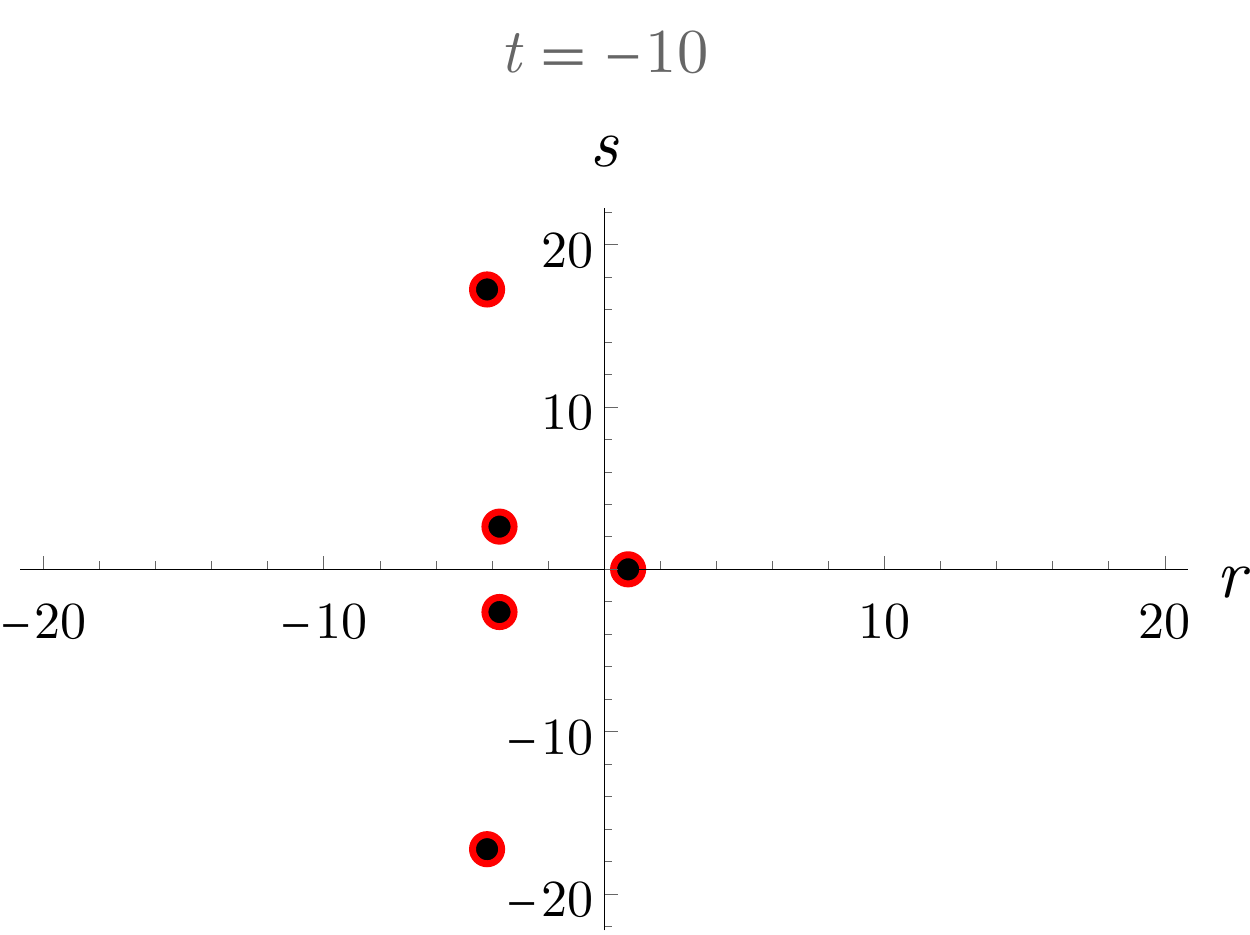} \qquad \qquad
\includegraphics[scale=0.4]{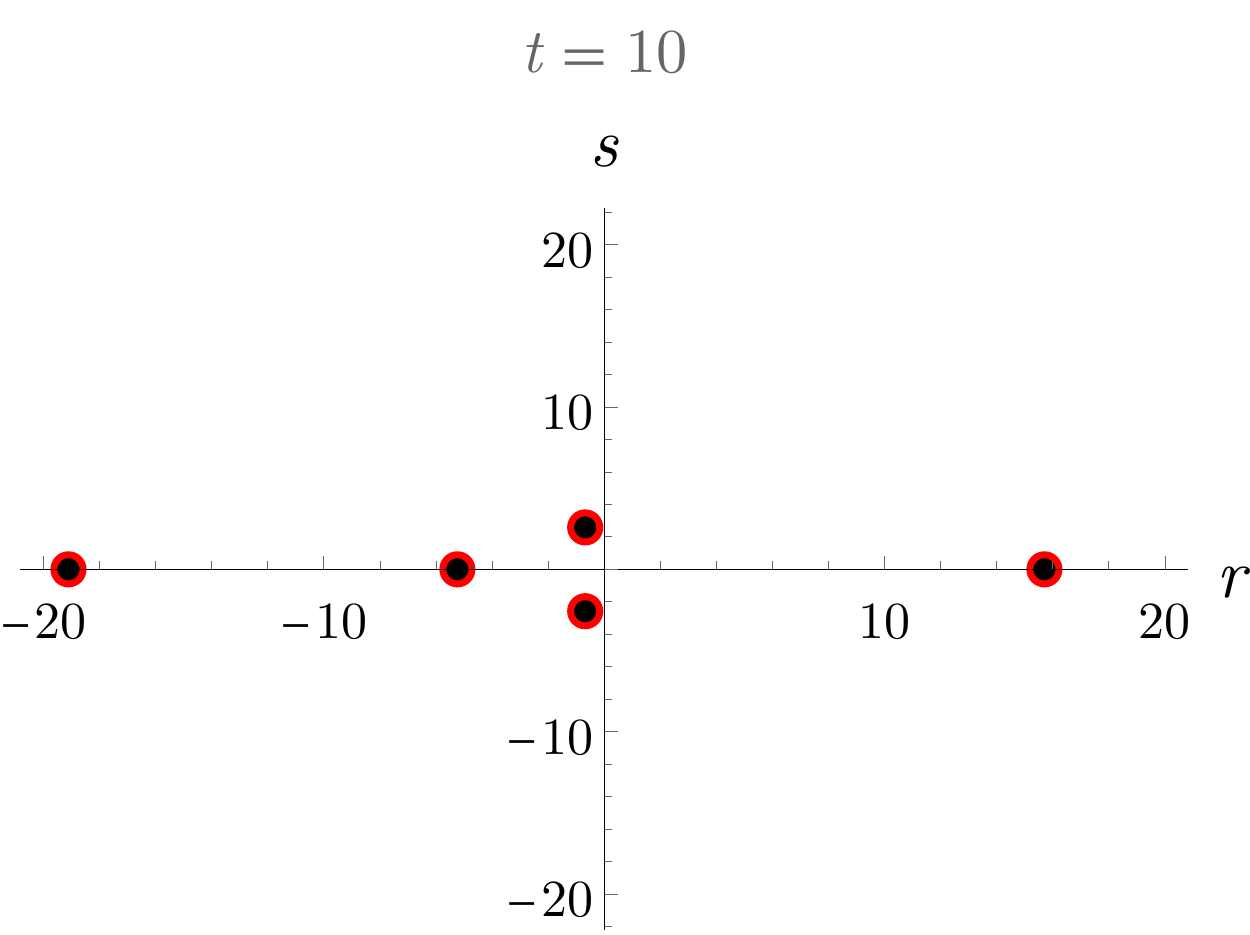}
\end{center}
\vspace{-0.2in}
\caption{Peak locations for Example 4.4 at $t=-10$ (left frames) and $t=10$ (right frames).
Black dots represent approximate peak locations where $Q({\bf 0})=0$
as in \eqref{Wzero}, red open circles are exact locations.
The top panel shows $5$-lump peak locations associated with the partition
$\lambda=(1,4)$, the bottom panel corresponds to its conjugate partition. 
KP parameters: same as in Figure~\ref{5lump}.}
\label{5lumppeaks}
\end{figure}
\begin{remark}
\begin{itemize}
\item[(a)] It should be emphasized that \eqref{nonzero} assumes that
the non-zero roots of the polynomial $H_\lambda(\xi)$ are simple.
Our assumption is based on a conjecture made in~\cite{FHV12} and
we are yet to find a counter-example. 
\item[(b)] The heat polynomials $H_n(\xi,\eta)$ introduced above 
\eqref{heat-a} satisfies the heat equation $h_{\xi\xi} = h_\eta$ in 
two-dimensions with initial condition $h(\xi,0)=\sf{\xi^n}{n!}$. 
They are closely related to the Hermite 
polynomials via $H_n(\xi,\eta) = \sf{(-\eta)^{n/2}}{n!}\,H_n(z), \,\, 
z=\xi(-4\eta)^{-1/2}$. The roots of the heat polynomials $H_n(\xi, \eta)$  
are real (in $\xi$) if the parameter $\eta < 0$ and pure imaginary
if $\eta > 0$. These roots describe the approximate peak locations of
$N$-lump solutions of KPI when the associated partition $\lambda$ has only
one part, i.e., $n=1$ and $\lambda=(N)$. These were recently studied in~\cite{CZ21}.
\item[(c)] The coefficient $C_q$ in Proposition 4.1 is related to the 
character $\chi^\lambda(1^{N-2q},2^q)$ by the formula
\[C_q = \frac{\chi^\lambda(1^{N-2q},2^q)}{\chi^{\tilde{\lambda}}(1^{N-2q})q!} \,. \]
An expression for $C_q$ was derived
in~\cite{BDS20} using the same ideas related to the 2-core of a partition
as the present article. However, we believe that \eqref{Cq} of this paper
provides a more direct formula for $C_q$.
\item[(d)] The long time asymptotics of the approximate peak locations
described in this subsection was also obtained in~\cite{YY21} but from
a different approach as explained in Section 1.    
\end{itemize}
\end{remark}
\subsection{Asymptotic peak heights}
In Section 4.1, the asymptotic form of the approximate peak
locations $Z_j(t)=r_j(t)+is_j(t), \, j=1,\ldots,N$ for the
KPI $N$-lump solutions were obtained by solving \eqref{Wzero}.
That is, the $Z_j(t)$ satisfy 
\begin{equation} Q({\bf 0})(Z_j(t))=0 \quad \Leftrightarrow \quad 
W_\lambda(\tilde{\theta})(Z_j(t))=0\,, \qquad j=1,\ldots,N \,.
\label{Qzero}
\end{equation}
The approximate peak heights are then obtained by calculating
$u_j(t) = u(Z_j(t))$ which estimates the local maximum values of the
solution \eqref{u} in the co-moving $rs$-plane.

In order to obtain an asymptotic expression for $u_j(t)$,
we decompose the $\tau$-function in \eqref{square-b} as 
$$
\tau_p = \frac{|Q({\bf 0})|^2 +l}{(2b)^{n(n-1)}}\,, \qquad
l = \sum_{{\bf r} > {\bf 0}} \frac{|Q({\bf r})|^2}{(2b)^{2|\lambda({\bf r})|}}\,,
$$
where recall that $|{\bf 0}|= \sf{n(n-1)}{2}$ and (from Section 3.1) that  
$|{\bf r}|-|{\bf 0}|=|\lambda({\bf r})|$.
Thus $l$ denotes all other square terms in the sum in \eqref{square-b}
except the first one. Then the $N$-lump solution in \eqref{u} can be expressed as
\[u = 2\partial^2_{x}\ln(|Q({\bf 0})|^2+l) = 
2\,\frac{\big(|Q({\bf 0})|^2\big)_{xx}+l_{xx}}{l}-
2\left(\frac{\big(|Q({\bf 0})|^2\big)_x+l_x}{l}\right)^2 \,. \]
Writing $|Q({\bf 0})|^2=|A+iB|^2=A^2+B^2$ and using \eqref{Qzero},
one obtains
$$\partial_x|Q({\bf 0})|^2 = 2(AA_x+BB_x) =0,, \quad \text{and} \quad  
\partial^2_{xx}|Q({\bf 0})|^2=2(A_x^2+B_x^2)$$ 
at $Z_j(t)$. Furthermore, $A_x^2+B_x^2 = |A_x+iB_x|^2 = |\partial_xQ({\bf 0})|^2$.
In order to calculate $\partial_xQ({\bf 0})$, first note that from either
the definition in \eqref{ss} or the explicit form in \eqref{schur},
$$\partial_xW_\lambda = i\partial_{\theta_1}W({\bf 0}) = iW({\bf 1}) = 
iW_{\lambda/\lambda({\bf 1})}$$ 
by differentiating each row of the determinant $W({\bf 0})$. 
Recall that the multi-index ${\bf 1} := (01 \cdots (n-2)n)$ and the corresponding
partition is given by $\lambda({\bf 1})=(0,0,\ldots,1)$.
Then it follows from \eqref{Wshift} that
$$\partial_xQ({\bf 0})
= i\partial_{\tilde{\theta}_1}i^{|{\bf 0}|}W_\lambda(\tilde{\theta})=
i^{\sf{(n-2)(n+1)}{2}}W_{\lambda/\lambda({\bf 1})}(\tilde{\theta}) \,.$$
Therefore, the expression for each peak height takes the form
\begin{equation*}
u_j(t) = 4\frac{|W_{\lambda/\lambda({\bf 1})}(\tilde{\theta})|^2}{l}+ 2(\ln l)_{xx} \,,
\end{equation*}
where the right hand side is evaluated at $\theta(Z_j(t))$.
The next task is to estimate $W_{\lambda/\lambda({\bf 1})}$ and $l$ 
in the above expression for $u_j$ at the peak locations $Z_j(t)$.

From Proposition 3.1 it follows that the degree of $\theta_1$ in
$W_\lambda$ and $W_{\lambda/\lambda({\bf r})}$ are $N$ and 
$N-|\lambda({\bf r})|$, respectively. Consider first the case when 
$\theta_1 \sim |t|^{1/p}$, $p=2$ or $p=3$ for $|t| \gg 1$. 
Then $W_{\lambda/\lambda({\bf r})} \sim O(|t|^{(N-|\lambda({\bf r})|)/p})$ 
at $Z_j(t)$ for all ${\bf r} \geq 0$. It is then evident from
\eqref{skewschur-b} that the leading order term for $|t| \gg 1$ in 
$l$ is $\sf{1}{4b^2}|Q({\bf 1})|^2$
so that $l \sim \sf{1}{4b^2}|Q({\bf 1})|^2\big[1+O(|t|^{-1/p})\big]$. Like $Q({\bf 0})$ in
\eqref{Wshift}, $Q({\bf 1})$ can also be expressed as a shifted skew Schur
function. Differentiating \eqref{shift} with respect to the shift variables
$h=(h_1,\ldots,h_N)$, one can derive an analogous property for the
shifted skew Schur function~\cite{M95,K18}   
\begin{equation*}
W_{\lambda/\nu}(\theta+h) = \sum_\mu W_{\mu/\nu}(h)W_{\lambda/\mu}(\theta)\,, \qquad 
\nu \subseteq \mu \subseteq \lambda \,.     
\end{equation*}
Then by putting ${\bf r} = {\bf 1}, \, \lambda({\bf 1})=\nu, \,\lambda({\bf s})=\mu$ 
in \eqref{skewschur-b}, and suitably choosing
a set of new shifted variables $c=(c_1,c_2, \ldots, c_N)$ such that  
$W_{\lambda({\bf s})/\lambda({\bf 1})}(c) = 
i^{|\lambda({\bf s})|-|\lambda({\bf 1})|}U\tbinom{\bf 1}{\bf s}$,
the above relation for shifted skew Schur function leads to
\[ Q({\bf 1}) = W_{\lambda/\lambda({\bf 1})}(\hat{\theta})\,, \qquad
\hat{\theta} := \theta + c \quad \text{then,} \quad 
l \sim \sf{1}{4b^2}|W_{\lambda/\lambda({\bf 1})}(\hat{\theta})|^2\big[1+O(|t|^{-1/p})\big]\,.
\]
Since $l$ is a polynomial in $\theta_1$ (and $\bar{\theta}_1$), $(\ln l)_{xx}$
decays as $\theta_1^{-2}$ (and $\bar{\theta}_1^{-2}$), thus 
$(\ln l)_{xx} \sim O(|t|^{-2/p})$ for $|t| \gg 1$. Substituting these estimates 
into the expression of $u_j(t)$ above yields 
\begin{equation}
u_j(t) \sim 16b^2\frac{|W_{\lambda/\lambda({\bf 1})}(\tilde{\theta})|^2}
{|W_{\lambda/\lambda({\bf 1})}(\hat{\theta})|^2}\big(1+ O(|t|^{-1/p})\big) \,,
\label{upeak}
\end{equation}
and where $W_{\lambda/\lambda({\bf 1})}(\tilde{\theta}) \sim 
W_{\lambda/\lambda({\bf 1})}(\hat{\theta}) \sim O(|t|^{(N-1)/p})$ as $|t| \gg 1$.
Then from \eqref{upeak}, it follows that
asymptotically as $|t| \to \infty$, each peak height $u_j(t) \sim 16b^2$
for $j=1,\ldots,N$. Thus the height of each of the $N$ peaks of a KPI $N$-lump
solution asymptotically approaches the $1$-lump peak height.

The asymptotics presented above need to be slightly modified for the case
of triangular $N$-lump solutions. As mentioned in Section 4.1.1, if 
$N=3m+1, \, m \in \mathbb{N}$, then $\xi=0$ is a root of the Yablonskii-Vorob'ev
polynomials $Q_n(\xi)$. The corresponding peak location $Z_0(t) \sim O(1)$  
as mentioned earlier. The asymptotic formula for the approximate peak 
height $u(Z_0(t))$ is obtained in the same way as prescribed above except 
that one needs to re-estimate
$W_{\lambda/\lambda({\bf 1})}$ and the derivatives $l_x, l_{xx}$ of the lower order
term $l$ to obtain a corresponding expression that would replace \eqref{upeak}.
To that end, one needs to first consider the asymptotic expression for 
$W_\triangle(\tilde{\theta})$ given just above \eqref{triangle}. For $|t| \gg 1$,
the dominant contribution arises from the last term 
$\tilde{\theta}_1\tilde{\theta}_3^m$ in the sum since the coefficient
$\chi^\triangle(1,3^m) \neq 0$~\cite{KO03} (see Remark 4.1(d)). Since 
$\tilde{\theta}_1 \sim O(1)$ and $\tilde{\theta}_3 \sim O(t)$, 
\[W_\triangle(\tilde{\theta}) \sim O(t^m)\,, \quad 
\partial_{\tilde{\theta}_1}W_\triangle(\tilde{\theta})=
W_{\triangle/\lambda({\bf 1})}(\tilde{\theta}) = O(t^m)\,, \quad  
W_{\triangle/\lambda({\bf r})}(\tilde{\theta}) \sim  O(t^{m({\bf r})})\,,
\quad m({\bf r}) = \lfloor\sf{N-\lambda({\bf r})}{3}\rfloor\,,  \]
where $\tilde{\theta} = \tilde{\theta}(Z_0(t))$. Interestingly however, 
$\partial^2_{\tilde{\theta}_1}W_\triangle(\tilde{\theta})\sim O(t^{m-2})$ instead of 
$O(t^{m-1})$ because the $\tilde{\theta}_1^4\tilde{\theta}_3^{m-1}$ is missing from
$W_\triangle(\tilde{\theta})$ since $\chi^\triangle(1^4,3^{m-1}) = 0$ when 
$N=3m=1$~\cite{Ta00} (see Remark 4.1(d)). The above estimates then give
\[ l \sim \sf{1}{4b^2}|W_{\triangle/\lambda({\bf 1})}(\hat{\theta})|^2\big[1+O(|t|^{-2})\big]\,,
\quad \quad l_x/l \sim l_{xx}/l = O(t^{-2}) \]
since $l_x$ involves $\partial^2_{\tilde{\theta}_1}W_\triangle(\tilde{\theta})$
and its complex conjugate. Therefore, \eqref{upeak} is modified accordingly as
\begin{equation*}
u_j(t) \sim 16b^2\frac{|W_{\lambda/\lambda({\bf 1})}(\tilde{\theta})|^2}
{|W_{\lambda/\lambda({\bf 1})}(\hat{\theta})|^2}\big(1 + O(t^{-2})\big) \,,
\end{equation*}
where $W_{\triangle/\lambda({\bf 1})}(\tilde{\theta}) \sim 
W_{\triangle/\lambda({\bf 1})}(\hat{\theta}) \sim O(t^{(N-1)/3})$ as $|t| \gg 1$.
Thus $u(Z_0(t)) \sim 16b^2$ as $|t| \to \infty$ like in the other previous cases.  
 
The asymptotic form of $N$-lump solution in the neighborhood of each
peak location is also derived from the estimates provided above.
In the $rs$-plane let $(r,s) = (r_j(t)+\Delta r, s_j(t)+\Delta s)$ where 
$\Delta r, \Delta s \sim O(1)$ and $Z_j(t) = r_j(t)+is_j(t)$ is the 
$j^{\rm th}$ approximate peak location, $1 \leq j \leq N$.  
Then the (re-scaled) KPI $\tau$-function near $Z_j(t)$
has the following form for $|t| \gg 1$
\[(2b)^{n(n-1)}\tau_P(r,s,t)=|Q({\bf 0})|^2+l \sim |W_\lambda(\tilde{\theta})|^2 + 
\sf{1}{4b^2}|W_{\lambda/\lambda({\bf 1})}(\hat{\theta})|^2\big[1+O(|t|^{-1/p})\big]\,,\]
where as before $\tilde{\theta} = \theta+h$ and $\hat{\theta}+\theta+c$.
Using \eqref{thetars}
\[W_\lambda(\tilde{\theta})(r,s,t) = W_\lambda(\tilde{\theta}(Z_j(t))
+ (i\Delta r-\Delta s) \partial_{\tilde{\theta_1}} W_\lambda(\tilde{\theta}(Z_j(t))
+\sf{i \Delta s}{2b} \partial_{\tilde{\theta_2}} W_\lambda(\tilde{\theta}(Z_j(t))
+H(\Delta r, \Delta s) \,, \]
where $H(\Delta r, \Delta s)$ consists of quadratic and higher powers in 
$\Delta r, \Delta s$. Now $W_\lambda(\tilde{\theta}(Z_j(t))=0$ from \eqref{Qzero},
and $\partial_{\tilde{\theta_1}} W_\lambda(\tilde{\theta}(Z_j(t)) = 
W_{\lambda/\lambda({\bf 1})}(\tilde{\theta})(Z_j(t)) = O(|t|^{(N-1)/p})$
for $|t| \gg 1$ where we take $p=2,3$ according to whether 
$\tilde{\theta}_1 \sim |t|^{1/2}$ or $\tilde{\theta}_1 \sim |t|^{1/3}$.
Moreover, differentiating $W_\lambda=W({\bf 0})$ with respect to
$\tilde{\theta}_2$ either in \eqref{schur} or in \eqref{ss}, it follows
that $\partial_{\tilde{\theta_2}} W_\lambda(\tilde{\theta})(Z_j(t)) =
W_{\lambda/\lambda({\bf r})}(\tilde{\theta})(Z_j(t)) \sim O(|t|^{(N-2)/p})$ 
since in this case $|\lambda({\bf r})|=2$. In a similar fashion, it
can be shown that $H(\Delta r, \Delta s) \sim 
O(|t|^{(N-|\lambda({\bf r})|)/p}), \, |\lambda({\bf r})| \geq 2$
since $H(\Delta r, \Delta s)$ involves 
$\partial^j_{\tilde{\theta_1}} W_\lambda(\tilde{\theta})$ and
$\partial^j_{\tilde{\theta_2}} W_\lambda(\tilde{\theta}), \, j \geq 2$ terms.
Collecting all the dominant terms, yields the following asymptotic form of
the $\tau$-function in an $O(1)$ neighborhood of each peak $Z_j(t)$  
\begin{subequations}
\begin{equation}
(2b)^{n(n-1)}\tau_P(r,s,t) \sim |W_{\lambda/\lambda({\bf 1})}(\tilde{\theta})|^2(Z_j(t))
\Big[(\Delta r)^2+(\Delta s)^2 + \sf{1}{4b^2} +O(|t|^{-1/p}) \Big] \,. 
\label{taupeak-a}
\end{equation}
Finally, substituting \eqref{taupeak-a} in \eqref{u} gives  
\[u(r,s,t) \sim u_1(\Delta r, \Delta s) + O(|t|^{-1/p}) \,,\]
near each $Z_j(t), \, j \geq 1$, where $u_1$ is the $1$-lump solution.
For the special case of triangular
$N=3m+1$-lumps where one of the approximate peak locations, 
$Z_0(t) \sim O(1)$, a similar asymptotic analysis as above leads to
\begin{equation}
(2b)^{n(n-1)}\tau_P(r,s,t) \sim |W_\lambda(\tilde{\theta})|^2(Z_0(t))
\Big[(\Delta r)^2+(\Delta s)^2 + \sf{1}{4b^2} +O(|t|^{-1}) \Big] \,. 
\label{taupeak-b}
\end{equation}
We omit the details.
\label{taupeak}
\end{subequations}
Thus it has been shown here that asymptotically 
as $|t| \to \infty$, the $N$-lump solution of KPI splits into $N$
distinct lumps whose approximate peak locations are given in Sections 4.1.1 and 4.1.2.
Furthermore, the solution $u$ in an $O(1)$ neighborhood of each peak locations
is asymptotic to the $1$-lump solution $u_1$ and decays algebraically as
$O(|t|^{1/p})$ elsewhere in the co-moving frame.
In this sense the $N$-lump solution can be viewed as a superposition of 
$N$ distinct $1$-lump solutions. This asymptotics is useful to compute
the conserved quantities. For example, the time-invariance
of $\int\!\!\int_{\mathbb{R}^2}u^2$ implies that
$$\int\!\!\int_{\mathbb{R}^2}u^2 = 
N\int\!\!\int_{\mathbb{R}^2}u_1^2 + O(|t|^{-1/p}) 
= N\int\!\!\int_{\mathbb{R}^2}u_1^2 = N(16\pi b)\,,$$  
for $p=3,2$ or $1$, where the second to last expression above follows
by letting $|t| \to \infty$. 
Other conserved quantities for the KPI $N$-lump solutions can be
computed using same decomposition rule.  
\section{Concluding remarks}
A comprehensive study of a class of rationally decaying, nonsingular,
multi-lump solutions of the 
KPI equation is carried out. These solutions are constructed
employing the binary Darboux transformation. It is
shown that geometrically the KPI $\tau$-function corresponds to a point 
in the complex Grassmannian ${\rm Gr}_{\mathbb{C}}(n,m_n+1)$ endowed with 
a hermitian inner product. Explicit formula for the $\tau$-function as a 
sum of squares is derived in terms of Schur functions associated   
with partition of a positive integer $N$. A classification scheme of the
multi-lump solutions is then developed based on integer partition theory.
Moreover, it is shown that there exists a {\it duality} among the multi-lump
solutions associated with a partition $\lambda$ and its conjugate $\lambda'$
of $N$.

The solution structure of the multi-lumps associated with each partition $\lambda$
of $N \in \mathbb{N}$ is analyzed to show that the waveform consists of
$N$ distinct peaks for $|t| \gg 1$. Moreover, near each peak 
the solution behaves asymptotically 
as a $1$-lump solution and is vanishingly small elsewhere.
The key ingredient in this analysis is the fact that each multi-lump
solution is uniquely characterized by a 
specific Schur function which is the characteristic of the irreducible 
representation of the symmetric group $S_N$.
Consequently, asymptotic analysis of the peak locations for large $|t|$ 
in conjunction with combinatorial methods from integer partition
theory reveal that 
to leading order, the peak distribution in the co-moving plane is described by the
zeros of the Yablonskii-Vorob'ev and Wronskian-Hermite polynomials whose
coefficients are determined by the irreducible characters of the symmetric
group $S_N$. The root structure of these polynomials give rise to novel,
geometric surface wave patterns exhibited by the KPI multi-lumps 
for $|t| \gg 1$. Furthermore, it is shown that the duality of the Schur functions 
associated with a partition and its conjugate leads to certain symmetries 
of the peak distributions for large positive and negative values of $t$.

It is possible to generate a number of interesting rational solutions
by considering special reductions of the KPI multi-lump solutions
described in this article. In certain cases it is also possible to obtain
special rational solutions of $1+1$-equations such as the Boussinesq and the 
non-linear Schr\"odinger equation as symmetry reductions of the
KPI equation. A future direction of study is to investigate
systematically such reduction processes. It is well known that the multi-lump
solutions of KPI form rational potentials associated with the classical
non-stationary Schr\"odinger equation. We plan to report our investigation
on this topic in the future and to relate our work to the IST scheme discussed 
in~\cite{AV97,VA99} in the hope of addressing certain open issues. 

\section{Acknowledgments}
SC thanks Prof. Yuji Kodama (OSU) for useful discussions. 
Initial part of this work was partially supported
by NSF grant No. DMS-1410862.



\end{document}